\documentclass[11pt]{article} %
\usepackage{fullpage}
\usepackage{amsthm}
\usepackage{amsmath, amssymb, graphicx, url}
\usepackage{color}
\usepackage{todonotes}
\usepackage{dsfont}
\usepackage{caption}
\usepackage{subcaption}
\usepackage{comment}
\usepackage{hyperref}
\definecolor{refcol}{rgb}{0.1,0,0.6}
\hypersetup{colorlinks}
\hypersetup{linkcolor=refcol, citecolor=refcol}

\usepackage[nameinlink]{cleveref}

\newcommand{\cB}{\mathcal{B}}

\newcommand{\E}{\mathbb{E}}

\newcommand{\N}{\mathbb{N}}

\newcommand{\Prob}{\mathbb{P}}
\newcommand{\R}{\mathbb{R}}

\newcommand{\U}{\mathbb{U}}
\newcommand{\V}{\mathbb{V}}

\newcommand{\xx}{\mathbf x}
\newcommand{\bb}{\mathbf b}
\newcommand{\yy}{\mathbf y}
\newcommand{\zz}{\mathbf z}

\newcommand{\FF}{\mathbf F}

\newcommand{\PP}{\mathbf P}
\newcommand{\QQ}{\mathbf Q}

\newcommand{\SM}{\mathbf S^M}

\newcommand{\qq}{\mathbf q}
\newcommand{\rr}{\mathbf r}

\newcommand{\PPhi}{\boldsymbol \Phi}
\newcommand{\SSigma}{\boldsymbol \Sigma}
\newcommand{\zzeta}{\boldsymbol \zeta}
\newcommand{\BB}{\boldsymbol B}
\newcommand{\XS}{\boldsymbol S}
\newcommand{\XX}{\boldsymbol X}
\newcommand{\YY}{\boldsymbol Y}
\newcommand{\ppi}{\boldsymbol \pi}
\newcommand{\aalpha}{\boldsymbol \alpha}

\newcommand{\Corr}{\mathrm{Corr}}

\newcommand{\PD}{\mathcal P}
\newcommand{\CM}{\mathcal C}

\newcommand{\bs}{\boldsymbol}

\newcommand{\Nv}{\mathrm N}

\usepackage{enumitem}
\usepackage[normalem]{ulem}

\newcommand{\dd}{{\rm d}}

\DeclareMathOperator{\Tr}{Tr}

\DeclareMathOperator{\diag}{diag}

\usepackage{mathtools}

\usepackage{algorithm,algpseudocode} 
\usepackage{authblk}

\title{Metropolis-adjusted interacting particle sampling}

\author[1]{Bj\"orn Sprungk}
\author[2]{Simon Weissmann}
\author[3]{Jakob Zech}
\date{\today}
\affil[1]{\normalsize Faculty of Mathematics and Computer Science, Technische Universit\"at Bergakademie Freiberg\\

  09599 Freiberg, Germany\\

  \texttt{bjoern.sprungk@math.tu-freiberg.de}
}
\affil[2]{\normalsize
  Universit\"at Mannheim, Institute of Mathematics\\

  68138 Mannheim, Germany\\
  
\texttt{simon.weissmann@uni-mannheim.de}
}
\affil[3]{\normalsize
  Universit\"at Heidelberg, Interdisziplin\"ares Zentrum f\"ur Wissenschaftliches Rechnen\\

  69120 Heidelberg, Germany\\
  
\texttt{jakob.zech@uni-heidelberg.de}
}

\newtheorem{theorem}{Theorem}[section]

\newtheorem{corollary}[theorem]{Corollary}

\newtheorem{example}[theorem]{Example}

\newtheorem{proposition}[theorem]{Proposition}
\newtheorem{remark}[theorem]{Remark}

\newtheorem{assumption}[theorem]{Assumption}

\begin{document}

\maketitle

\begin{abstract}
  In recent years, various interacting particle samplers have been developed to sample from complex target distributions, such as those found in Bayesian inverse problems. These samplers are motivated %
  by the
    mean-field limit
    perspective and implemented as ensembles of particles that move in
    the product state space according to coupled stochastic
    differential equations. The ensemble approximation and numerical
    time stepping used to simulate these systems can introduce bias
    and affect the invariance of the particle system with respect to
    the target distribution. To correct for this, we investigate the
    use of a Metropolization step, similar to the Metropolis-adjusted
    Langevin algorithm. We examine %
    Metropolization of either the whole ensemble or smaller
      subsets of the ensemble, and prove basic convergence of the resulting
    ensemble Markov chain to the target distribution. Our numerical results
    demonstrate the benefits of this correction in numerical examples
    for popular interacting particle samplers such as
    ALDI, %
    CBS, and stochastic %
    SVGD.

\end{abstract}

{\bf Keywords:} Metropolis-Hastings, interacting particle systems, Bayesian inference

\section{Introduction}
Generating samples or computing expectations with respect to a given target distribution $\pi$ in $\R^d$ is a ubiquituous task in applied mathematics, computational physics, statistics, and data science. Applications are %
broad
and include for example Bayesian inference, generative modeling, and hypothesis testing and model fitting.

Various methods tackling this problem have been proposed and analyzed
in the literature.
A classical and nowadays standard method is
Markov-Chain Monte Carlo (MCMC) \cite{MCMCbook} and, in particular, the popular Metropolis-Hastings (MH) algorithm \cite{MRRT1953,H1970}. 
Recently, novel approaches that couple the target distribution with a
reference distribution $\pi_0$ through a deterministic
``transport map'' have emerged, such as polynomial transports
\cite{MR3821485, jaini2019sum}, tensor-train transports
\cite{MR4065222}, normalizing flows \cite{pmlr-v37-rezende15}, and
neural ODEs \cite{NEURIPS2018_69386f6b}.  The resulting sampling
methods aim to transform initial iid~samples following the reference
distribution to samples (approximately) following the target
distribution by applying the transport map samplewise.

Another way to achieve such a transformation of an initial ensemble of particles or samples following $\pi_0$ is by applying suitable stochastic dynamics to the ensemble which for time $t\to\infty$ %
yield particles approximately distributed according to the target $\pi$.
The resulting ensemble dynamics are often interacting, i.e., the drift or diffusion term for each particle depends on the whole ensemble.
Such stochastic interacting particle systems emerge from various ideas and approaches: (i) as ensemble approximations of $\pi$-invariant stochastic differential equations of Langevin or McKean-Vlasov type \cite{GHWS2020, GNR2020, GM1994}, (ii) by adapting methods from particle swarm optimization to construct samplers \cite{CHSV2022}, or (iii) from gradient flows to minimize some objective quantifying the difference between the target $\pi$ and a current approximation $\pi_t$ \cite{LW2016, Nsken2021SteinVG, Gallego2018StochasticGM}. 
For each of these approaches we consider a particular example in this work: (i) an affine invariant interacting Langevin sampler (ALDI) \cite{GHWS2020,GNR2020}, (ii) a consensus-based sampler (CBS) \cite{CHSV2022}, and (iii) a stochastic version of Stein variational gradient descent (SVGD) \cite{Gallego2018StochasticGM}.

In practice, simulating the resulting stochastic dynamical system requires a time discretization and a suitable numerical integration scheme.
For example, approximating the Langevin dynamics via the Euler-Maruyama scheme leads to the so-called unadjusted Langevin algorithm (ULA) \cite{MR1440273}. 
For ULA it can be shown that the time-stepping scheme causes a bias, so that the limiting distribution has an error of the size of the time discretization step, see e.g.~\cite[Theorem~2]{NEURIPS2019_65a99bb7}. 
The Metropolis-adjusted Langevin algorithm (MALA) circumvents this problem by introducing an MH acceptance/rejection-step after each Langevin update \cite{Besag1994, MR1440273}.

In this paper, we propose a similar approach for recent
\emph{interacting} particle systems with $M\in\N$ particles. That is,
we view the time-discrete interactive %
update of the particles as a proposal within an MH scheme.
This results in a Markov chain in the $M$-fold product space, $\R^{d}\times\cdots\times\R^d\simeq\R^{Md}$, that corrects for the bias introduced by the time-discretization of the ensemble   dynamics. Additionally, it offers a natural approach to parallelize Markov chain Monte Carlo sampling, and potentially leads to improved proposals due the whole ensemble's information being used.

The concept of ensemble MCMC was first introduced in \cite{CF2010,
  GW2010}, where particles, referred to as ``walkers'', are updated
individually using so-called walk or stretch moves that involve only
two of the $M$ particles. %
  Following \cite{GW2010} further ensemble MCMC
algorithms have been proposed in recent years
\cite{CW2021,DS2022,LMW2018}.  In these works the ensemble is used to
estimate the target covariance empirically.  The ensemble covariance
is then applied as a preconditioner or covariance for proposing new
states based on a Gauss-Newton update or the (generalized)
preconditioned Crank-Nicolson proposal \cite{CRSW2013,RS2018},
respectively.  Similar to \cite{GW2010}, the authors of
\cite{CW2021,DS2022,LMW2018} use a sequential particle-wise update
and, hence, Metropolization.

\paragraph{Outline}
The remainder of this paper is organized as follows: In this section,
we explain interacting particle systems and present our main ideas,
describe our contributions, and introduce notation.
In \Cref{sec:prelim}, we review the basic methodology of
the Metropolis-Hastings algorithm, including the specific instance
known as ``MALA'', and also discuss classic results related to its
convergence.
In \Cref{sec:MAIPS}, we present %
three general strategies of
Metropolizing interacting particle systems---ensemble-wise, particle-wise and block-wise --- and provide convergence
results for %
each of them.
\Cref{sec:examples} discusses common examples of
  interacting particle systems that are based on various underlying
  stochastic dynamics, and we explain how they align with our
  Metropolization schemes.
Finally, in \Cref{sec:numerics} we report on numerical results
for all presented interacting particle methods.

 \paragraph{Notation and conventions}
 Throughout we consider an underlying probability space $(\Omega,\mathcal F,\mathbb P)$, $\R^d$ to be equipped with the Borel $\sigma$-algebra $\cB(\R^d)$,
 and we assume the target probability distribution $\pi$ on $\R^d$
 to be absolutely continuous with respect to Lebesgue measure.
By abuse of notation, we use the same symbols to denote the Lebesgue densities and the corresponding distributions they represent.
Moreover, we denote by $\PD(\R^d)$ the set of probability densities on $\R^d$.
As usual, ${\rm N}(\mu,\Sigma)$ stands for a normal distribution with mean $\mu\in\R^d$
and covariance $\Sigma\in\cB(\R^d)$, and ${\rm U}([0,1])$ denotes a uniform
distribution on $[0,1]$.

For a measure $\mu$ on $\R^d$ we denote by $L_\mu^1(\R)$ the Lebesgue space of $\mu$-integrable functions $F:\R^d\to\R$. Similarly $L_\mu^2(\R)$ stands for the square integrable functions w.r.t.\ the measure $\mu$. For $F\in L_\mu^1(\R)$ we write $\E_\mu[F]:=\int_{\R^d}F(x)\mu(\dd x)$ and additionally $\V_\mu[F]:=\E_\mu[(F-\E_\mu[F])^2]$ in case $F\in L_\mu^2(\R)$.     

We write $\CM(\R^d) \subset \R^{d\times d}$ for the set of symmetric positive semidefinite matrices of size $d\times d$. The set of all symmetric positive definite matrices is denoted by $\CM^+(\R^d)$ and we write $\sqrt{C}$ for the square root of $C\in\CM^+(\R^d)$. The $n$-dimensional identity matrix is denoted
by ${\rm Id}_n\in\R^{n\times n}$.

We use boldface notation $\xx$ to denote vectors in $\R^{Md}$. They are always interpreted as an ensemble of $M$ vectors in $\R^d$ which in turn are denoted by $x^{(1)},\dots,x^{(M)}$. The ensemble excluding the $i$th particle will be denoted by $\xx^{-(i)}$. More precisely
  \begin{equation}\label{eq:XX}
    \xx := \begin{pmatrix}
    x^{(1)}\\ \vdots\\ x^{(M)}
  \end{pmatrix}
  \in  \R^{Md}\qquad\text{and}\qquad
    \xx^{-(i)} := \begin{pmatrix}
    x^{(1)}\\ \vdots\\ x^{(i-1)}\\ x^{(i+1)}\\ \vdots \\ x^{(M)} 
  \end{pmatrix}\in\R^{M(d-1)}.
  \end{equation}
  For random variables we use upper case notation, for example $\XX=((X^{(1)})^\top,\dots,(X^{(M)})^\top)^\top$. The notation $\XX=\xx$ signifies that a draw of this random variable yielded the value $\xx$ 
  and we use $X^{(i)}\in\R^d$ or $\XX\in\R^{Md}$ as shorthand notation for $X^{(i)}$ being an $\R^d$-valued and $\XX$ being an $\R^{Md}$-valued random variable, respectively.

\subsection{Interacting particle systems}\label{sec:ideas}
The starting point of our method are dynamical systems that transform
a single particle $X_0\sim\pi_0$ %
at time $t=0$ into a particle following the target distribution %
$\pi$ as $t\to\infty$. 
The dynamics of the particle are described by a stochastic differential equation (SDE) of McKean-Vlasov type \cite{PRS2021}:
\begin{equation}\label{eq:SDE_infty}
\dd X_t = \phi(X_t,\pi_t)\dd t + \sqrt{\sigma(X_t,\pi_t)} \dd B_t,
\end{equation}
where $\pi_t:\mathbb{R}^d \to [0,\infty)$ is the probability density of $X_t \in \R^d$ at time $t$, $B_t \in \mathbb{R}^d$ is a (standard) Brownian motion, $\phi:\R^d \times \PD(\R^d) \to \R^d$ is referred to as the \emph{drift}, and $\sigma:\R^d \times \PD(\R^d) \to \CM(\R^d)$ is the \emph{diffusion}. 
Moreover, we assume existence of a unique strong solution $X_t$ of \eqref{eq:SDE_infty} throughout the paper.

\begin{example}[Langevin dynamics]\label{ex:langevin}
  One classical example of \eqref{eq:SDE_infty} is
  \begin{equation}\label{eq:langevin}
    \dd X_t = C\nabla \log \pi(X_t)\dd t+ \sqrt{2C}\dd B_t
  \end{equation}
  for a fixed covariance matrix $C\in \CM^+(\R^d)$.
  The density $\pi_t$ of $X_t$ satisfies the corresponding
  Fokker-Planck equation
  \begin{equation*}
    \partial_t \pi_t = \nabla\cdot(\pi_t C\nabla\log(\pi))+
    \Tr(C\nabla^2\pi_t),
  \end{equation*}
  which describes the gradient flow in the space of probability measures
  w.r.t.~the Wasserstein metric \cite{JKO98}.
  Under suitable assumptions ($\pi$ satisfies a Poincar\'e inequality)
  one can show exponential convergence of $\pi_t$ to $\pi$ as
  $t\to\infty$, e.g., \cite{MR1812873}. %
  
  Note that the %
    drift %
    and diffusion %
    in \eqref{eq:langevin}
    are independent of $\pi_t$.
    This is in contrast to the closely related ``Kalman-Wasserstein dynamics'' \cite{GHWS2020}: replacing $C$ with
    \begin{subequations}\label{eq:WassersteinDyn}
  \begin{equation}\label{eq:Cpi}
    C(\pi_t):=\int_{\R^d}(x - m(\pi_t))(x-m(\pi_t))^\top \pi_t(x)\, \mathrm d x \in \R^{d\times d},
    \qquad
    m(\pi_t) := \int_{\R^d} x \pi_t(x) \, \mathrm d x \in \R^d
  \end{equation}
  yields
  a McKean-Vlasov Langevin dynamic of the form \eqref{eq:SDE_infty} %
  with
  \begin{equation}
    \phi(x, \rho) = C(\rho) \nabla \log \pi(x),
    \qquad
    \sigma(x, \rho) = 2C(\rho),
  \end{equation}
  \end{subequations}
  for $x\in\R^d$ and $\rho\in\PD(\R^d)$.

  For strongly
  log-concave target measures $\pi$ and under the additional assumption that $C(\pi_t)$ does not degenerate, the resulting Markov process $(X_t)_{t\ge0}$ is ergodic with unique invariant distribution $\pi$ and %
it holds exponential convergence of $\pi_t$ to $\pi$ in the Kullback-Leibler divergence as $t\to\infty$ \cite[Proposition~2]{GHWS2020}.
Potential advantages of replacing $C$ by $C(\pi_t)$ are (i) faster
convergence of $\pi_t\to\pi$ due to the preconditioning and (ii)
affine-invariance of the resulting dynamics, see \cite{GNR2020,LMW2018}.
  
\end{example}

\paragraph{Ensemble discretization}
Solving \eqref{eq:SDE_infty} by numerical methods requires %
to discretize.
In terms of the distribution $\pi_t$, this is achieved by replacing $\pi_t$ with the empirical distribution of an ensemble of $M\in\N$ particles $X^{(i)}_t \in \R^d$, $i=1,\dots,M$, initialized iid~as $X_0^{(i)}\sim\pi_0$, $i=1,\dots,M$, at time $t=0$. The particles are collected into a vector $\XX_t\in\R^{Md}$
representing the state of the whole ensemble
(cp.~\eqref{eq:XX}).
Equation \eqref{eq:SDE_infty} then formally becomes a coupled
system of SDEs
\begin{equation}\label{eq:SDE_Ensemble}
  \dd \XX_t = \PPhi(\XX_t)\,\dd t + \sqrt{\SSigma(\XX_t)}\,\dd \BB_t,
\end{equation}
for some suitable drift, diffusion, and standard Brownian motion 
\begin{equation*}
  \PPhi:\R^{Md} \to \R^{Md},\qquad
  \SSigma:\R^{Md} \to \CM(\R^{Md}),\qquad
  \BB_t\in\R^{Md}.
\end{equation*}
Again, we assume well-definedness of the solution %
$\XX_t$
of \eqref{eq:SDE_Ensemble} throughout.

\begin{remark}
  Some papers proceed by first proposing a discrete system of the type
  \eqref{eq:SDE_Ensemble}, and then studying the \emph{mean field
    limit} \eqref{eq:SDE_infty} obtained as $M\to\infty$.
\end{remark}

\begin{example}[Interacting Langevin Dynamics]\label{ex:EKS}

We continue the example of McKean-Vlasov Langevin dynamics from \Cref{ex:langevin}
\begin{equation}\label{eq:MV_langevin}
    \dd X_t = C(\pi_t) \nabla \log \pi(X_t)\dd t + \sqrt{2C(\pi_t)} \dd B_t.
\end{equation}
The computation of $C(\pi_t)$ in \eqref{eq:Cpi} requires to
approximate an integral w.r.t.~$\pi_t$. Using Monte Carlo
  integration based on an ensemble of $M$ particles following these
  dynamics we obtain the ensemble version \cite{GHWS2020} of the
  mean-field dynamics \eqref{eq:MV_langevin}
\begin{equation}\label{eq:EKS}
    \dd X_t^{(i)} = C(\XX_t) \nabla_x\log\pi(X_t^{(i)})+\sqrt{2C(\XX_t)}\dd B_t^{(i)}\qquad i\in\{1,\dots,M\},
\end{equation}
where 
\begin{equation}\label{eq:Cemp}
     C(\XX_t) := \frac{1}{M}\sum_{i=1}^M \left(X_t^{(i)}-m(\XX_t)\right)
     \left( X_t^{(i)}-m(\XX_t)\right)^\top\in \CM(\R^{d}),
     \qquad
     m(\XX_t) = \frac{1}{M}\sum_{i=1}^M X_t^{(i)}\in\R^d
\end{equation}
denote the empirical covariance and mean %
of the ensemble $\XX_t$.  Note that the system of SDEs \eqref{eq:EKS}
is completely coupled
since individual particles %
interact
via $C(\XX_t)$.

It can be shown that the invariant distribution of the particles
  in \eqref{eq:EKS} in general may have %
  a bias, i.e.\ does not equal $\pi$. To address this issue, the authors in \cite{NR19} propose the following modification:
  \begin{equation}\label{eq:ALDI}
    \dd X_t^{(i)} = C(\XX_t) \nabla \log\pi(X_t^{(i)})+ \frac{d+1}{M} (X_t^{(i)}-m(\XX_t))+\sqrt{2C(\XX_t)}\dd B_t^{(i)}\qquad i\in\{1,\dots,M\}.
\end{equation}

\end{example}

\paragraph{Time discretization}
Besides discretizing the distribution, a numerical time-stepping scheme to approximately simulate \eqref{eq:SDE_Ensemble} is required.
For a fixed time step size $h>0$, let
$\zzeta_k\sim \Nv(0,{\rm I}_{Md})$, $k\in\N$, i.e.\
$\zzeta_k\in\R^{Md}$ is normally distributed with mean $0$ and
covariance matrix given by the $Md$-dimensional identity matrix
${\rm
  I}_{Md}$. %
Then the Euler-Maruyama discretization of
\eqref{eq:SDE_Ensemble} %
reads
\begin{equation}\label{eq:discrete}
  \XX_{k+1}=\XX_k+h\PPhi(\XX_k)+\sqrt{h\SSigma(\XX_k)}\zzeta_{k+1}
\end{equation}
with $\XX_k=(X_k^{(1)},\dots,X_k^{(M)})\in\R^{Md}$ and initialized %
as $X_0^{(i)}\sim\pi_0$ iid for $i=1,\dots,M$.

\begin{example}[Unadjusted Langevin algorithm]\label{ex:ULA}
The %
Langevin dynamics \eqref{eq:langevin} do not require an ensemble approximation.
However, we may consider $M$ particles $X_t^{(i)}$, $i=1,\ldots,M$, each individually following \eqref{eq:langevin} without interaction.
This can %
be viewed as (non-interacting) ensemble dynamics of the form \eqref{eq:SDE_Ensemble}. 
The corresponding time-discretized system is %
\begin{equation}\label{eq:ULA}
  X_{k+1}^{(i)}
  =
  X_k^{(i)}
  +
  h C \nabla\log\pi(X_k^{(i)})
  +
  \sqrt{2hC} \zeta_{k+1}^{(i)},\qquad i\in\{1,\dots,M\},
\end{equation}
with $\zeta_{k+1}^{(i)}\sim \Nv(0,{\mathrm{Id}}_d)$ iid.  For
$C={\rm Id}_d$, this is known as the (parallel) ULA \cite{GM1994,
  MR1440273}.
\end{example}

\begin{example}[Unadjusted interacting Langevin dynamics]
  For the interacting Langevin dynamics \eqref{eq:ALDI}, we
  obtain the time-discrete interacting particle system
  \begin{equation}\label{eq:UILD} X_{k+1}^{(i)} = X_k^{(i)} + h
      C(\XX_k) \nabla\log\pi(X_k^{(i)}) + h \frac{d+1}{M} (X_k^{(i)} -
      m(\XX_k)) + \sqrt{2hC(\XX_k)} \zeta_{k+1}^{(i)}\quad
      i\in\{1,\dots,M\},
\end{equation}
with $\zeta_{k+1}^{(i)}\sim \Nv(0,{\mathrm{ Id}}_d)$ and $C(\XX_k)$ as in \eqref{eq:Cemp}.
\end{example}

It is well-known that the introduction of the time discretization
may lead to a bias, i.e.\ in general it does \emph{not} hold that $X_k^{(i)}$ converges in distribution to $\pi$ as $k\to\infty$. %
This applies for example to ULA \cite{NEURIPS2019_65a99bb7}.

\subsection{Main idea and contributions}
Julian Besag \cite{Besag1994} suggested in 1994 to correct the unadjusted Langevin algorithm %
with a Metropolization to obtain a $\pi$-invariant one-particle Markov chain $(X_k)_{k\in\N}$,
which lead to %
MALA.
We adopt this idea to correct for the bias in general time-discrete \emph{interacting particle systems} %
\eqref{eq:discrete}. 
To this end, we view \eqref{eq:discrete} as the proposal mechanism for an ensemble Markov chain $(\XX_k)_{k\in\N}$ in the product state space $\R^{Md}$.
We propose to apply Metropolization in three ways: (i) ensemble-wise, i.e., accept or reject the whole ensemble of all proposed particles, (ii) particle-wise, i.e., accept or reject each proposed particle individually in a sequential manner,
  and (iii) block-wise, i.e., accept or reject each block of particles individually in a sequential manner.
Here a \emph{block} is understood as a fixed subset of particles and identified with $\bb \subset \{1,\ldots,M\}$.
  Methods (i) and (ii) can be seen as a special case of (iii) with either just one batch representing the
  whole ensemble, or batches consisting of only one particle. 
  Due to their different algorithmic behaviour and to present clearly the underlying train of thought
  we will discuss the three versions separately.
  A high-level version of the novel block-wise Metropolization is summarized in \Cref{alg:general}.

  Moreover, we discuss also a \emph{simultaneous} version of particle-wise Metropolization in this work. Such a method is computationally appealing, as it enables particle-wise parallelization. However, it turns out that this strategy can, in general, yield a biased algorithm, i.e., the ensemble Markov chain does not have the correct invariant measure.
  We provide specific examples for the bias of simultaneous particle-wise Metropolization in \Cref{sec:pw-sim}.

\begin{algorithm}[h]
\begin{algorithmic}[1]
  \State fix a partition $\bigcup_{j=1}^L \bb_j = \{1,\dots,M\}$
  \State draw $\xx_0\in\R^{Md}$ according to $\otimes_{i=1}^M\pi_0$ and set initial ensemble state $\XX_0=\xx_0$
  \For{$k=0,\dots, N$}
  \State
  given $\XX_k = \xx_k$ initialize $\xx = (x^{(1)},\ldots,x^{(M)})$ by $\xx=\xx_k$
  \For{$j = 1,\dots,L$}
      \State
      draw block proposal $(y^{(i)})_{i\in\bb_j} \in \R^{|\bb_j|d}$ according particle proposals such as \eqref{eq:ULA} or \eqref{eq:UILD} based on current ensemble $\xx$
      \State update $(x^{(i)})_{i\in\bb_j} = (y^{(i)})_{i\in\bb_j}$ only with block acceptance probability
	   \EndFor
    \State set $\XX_{k+1} = \xx$

  \EndFor
  \State {\bf return}
  ensemble chain $(\XX_k)_{k=0}^{N+1}$
\end{algorithmic}
\caption{Block-wise Metropolization}\label{alg:general}
\end{algorithm}

We emphasize the following potential advantages of combining the Metropolis-Hastings mechanism with interacting particle systems:
\begin{itemize}
\item From an interacting particle sampling perspective we
  \emph{correct for the immanent bias} of the particle dynamics due to
  numerical time-stepping schemes and finite ensemble approximations
  of $\phi(\cdot, \pi_t)$ and $\sigma(\cdot,\pi_t)$ in
  \eqref{eq:SDE_infty}. In particular, this allows in principle to
  take large time steps $h$ in \eqref{eq:discrete} without provoking
  an instability of the time-discrete dynamical system.

\item
From an MCMC sampling perspective the interacting particle dynamics provide in each iteration not just one new state $x_k\in\R^d$ but $M$ new states which can be computed \emph{in parallel} which yields a computational advantage. 
Moreover, in comparison to simply performing, e.g., parallel MALA, the interaction of the particles may yield more \emph{efficient proposal kernels} due to estimating, e.g., the target covariance empirically by the ensemble, and, thus, lead to more efficient MCMC sampling. 
In particular, we obtain \emph{affine-invariant}\footnote{For the benefits of this property we refer to, e.g., \cite{CF2010,GW2010} and \cite{RS2022} where the latter work shows, for instance, the resulting affine-invariance of spectral gaps.} MCMC methods %
if the underlying interacting particle dynamics are affine-invariant themselves such as those proposed in \cite{GHWS2020, GNR2020, CHSV2022}.
\end{itemize}

\paragraph{Contributions}
We summarize the main contributions of this paper: %
\begin{enumerate}
\item We propose a new MCMC method by combining a Metropolization step with interacting particle dynamics in the $M$-fold product state space.
\item We discuss %
  several
  variants of this method %
  for differing sizes of the blocks which are Metropolized.
  Basic convergence results as well as %
  numerical indications suggesting that %
  suitable choices of block sizes lead to more efficient sampling
  algorithms
  are presented.
  
\item We give several %
  concrete examples based on different particle dynamics, including
  the recently introduced stochastic SVGD
 \cite{Nsken2021SteinVG,Gallego2018StochasticGM}, ALDI
  \cite{GHWS2020,GNR2020}, and CBS \cite{CHSV2022}.
\item We present numerical experiments to compare the different
  variants of our algorithm, demonstrating in particular improved
  robustness and higher efficiency due to interaction and
  (block-wise) Metropolization.
\item In the appendix, we provide counterexamples to show that a
    \emph{simultaneous} particle-wise Metropolization strategy does in general not yield an unbiased algorithm.
  \end{enumerate}

  \section{Preliminaries on Markov chain Monte Carlo}\label{sec:prelim}
  We recall the basic terminology and ideas of MCMC sampling required in the following.

    Throughout let $\pi$ be a given absolutely continuous (target) probability distribution on $\R^d$ with the Lebesgue density also denoted by $\pi:\R^d\to[0,\infty)$.
    To approximately sample from $\pi$, we construct a Markov
    chain $(X_k)_{k\in\N}\subseteq\R^d$ that converges to $\pi$ in
    distribution as $k\to\infty$. We denote the associated transition
    kernel by $P:\R^d\times \cB(\R^d)\to [0,1]$, i.e., the chain is
    characterized by $X_k=x$ implying $X_{k+1}\sim P(x,\cdot)$. For a
    measure $\mu$ on $\R^d$ we use the usual notation $\mu P$ for the measure
    \begin{equation*}
      \mu P(A):=\int_{\R^d}P(z,A)\mu(dz)\qquad\forall A\in\cB(\R^d)
    \end{equation*}
    and inductively $\mu P^k:=(\mu P^{k-1})P$ for all $k\ge 2$. Note that
    if $X_0\sim \pi_0$ for some initial distribution $\pi_0$ then $X_k\sim\pi_0P^k$.

    We say that %
    \begin{itemize}
    \item $\pi$ is an \emph{invariant measure} %
      of $P$ iff
     \begin{equation*}
  \pi = \pi P,
  \end{equation*}
  in which case we call $P$ and $(X_k)_{k\in\N}$ \emph{$\pi$-invariant},
\item the Markov chain is \emph{$\pi$-reversible} iff it satisfies
  the \emph{detailed balance condition}
  \begin{equation}
    \int_{A}P(x,B)\pi(\dd x) =
    \int_{B}P(z,A)\pi(\dd z)\qquad\forall A,B\in\cB(\R^d),
  \end{equation}
    \item the Markov chain is \emph{ergodic} iff %
      it holds
      \begin{equation}\label{eq:ergodic}
        \lim_{k\to\infty}d_{\rm TV}(\pi_0P^k,\pi)=0,
      \end{equation}
      where $\pi_0$ is the initial distribution and $d_{\rm TV}$ stands for the total variation distance,
    \item the Markov chain \emph{satisfies a strong law of large numbers} iff %
      \begin{equation}\label{eq:LLN}
        S_N(F) := \frac1N \sum_{k=1}^N F(X_{k})
        \xrightarrow[N\to\infty]{\text{a.\ s.}}\ \E_\pi[F]
	\qquad \forall F \in L^1_\pi(\R),
      \end{equation}

    \item
    the Markov chain \emph{satisfies a central limit theorem for $F\in L_\pi^2(\R)$} iff there exists $\sigma^2_F \in [0,\infty)$ such that %
    \begin{equation}\label{eq:CLT}
	\sqrt{N} \left(S_N(F) - \E_\pi[F]\right) 
    \ \xrightarrow[N\to\infty]{\mathcal L} \ \Nv\left(0, \sigma_F^2 \right).
    \end{equation}
    \end{itemize}

    Due to $X_k\sim\pi_0P^k$, the Markov chain $(X_k)_{k\in\N}$ can be
    interpreted as the realization of a fixed point iteration under
    the mapping $P$. Hence $\pi$ being an invariant measure of $P$ is
    a necessary condition to obtain ergodicity. Additionally, let us
    mention that $\pi$-reversibility is sufficient to ensure $\pi$-invariance.
    
    For Markov chains with $\pi$-reversible transition kernel $P$ the \emph{asymptotic variance $\sigma^2_F$} of $S_N(F)$ in \eqref{eq:CLT} can, in case of existence, be expressed by 
    \begin{equation}\label{eq:sigmaF}
        \sigma_F^2 
        = \V_\pi[F]%
        \left[ 1 + 2 \sum_{k=1}^\infty \Corr(F(\tilde X_0), F(\tilde X_{0+k}) \right],
    \end{equation}
    where $(\tilde X_k)_{k\in\N}$ denotes the $\pi$-reversible Markov chain with transition kernel $P$ starting in stationarity $\tilde X_0\sim \pi$. %
     While a strong law of large numbers
      holds under mild conditions, the central limit theorem is more
      nuanced. We refer to \cite[Section
      5]{RoRo2004} for more details.

  \subsection{Metropolis-Hastings algorithm}
  
  The key question is how to obtain transition kernels that ensure
  ergodicity and a strong law of large numbers. The
  MH algorithm \cite{MRRT1953, H1970} is a standard
  method achieving this under rather mild assumptions.

  The algorithm is based on a \emph{proposal kernel}
  $Q\colon \R^{d} \times \cB(\R^{d})\to[0,1]$, that assigns a
    probability measure $Q(x,\cdot)$ on $\R^d$ to every $x\in\R^d$, in
    combination
  with an acceptance-rejection step. 
Throughout, we assume that $Q(x,\cdot)$ possesses a Lebesgue density for each $x\in\R^d$, i.e.\ there exists $q:\R^d\times\R^d\to [0,\infty)$ such that
  \begin{equation*}
    Q(x,A)=\int_A q(x,y)\ \dd y\qquad\forall A\in\cB(\R^d).
  \end{equation*}
In the $k$th step, if $X_k=x_k$ and $y_{k+1}$ is a proposed value drawn from $Q(x_k,\cdot)$, then $X_{k+1}$ is set to $y_{k+1}$ with probability $\alpha(x_k,y_{k+1})$ defined as follows:
\begin{equation}\label{eq:alpha0}
\alpha(x_k,y_{k+1}) := \begin{cases}
\min\left(1,\frac{\pi(y_{k+1})q(y_{k+1},x_k)}{\pi(x_k)q(x_k,y_{k+1})}\right) &\text{if } \pi(x_k)q(x_k,y_{k+1})>0\\
1 &\text{otherwise}.
\end{cases}
\end{equation}
If $X_{k+1}$ is not set to $y_{k+1}$, it is set to $x_k$.
  The resulting transition kernel is
  \begin{equation}\label{eq:Pr}
    P(x,\dd y) = \alpha(x,y)q(x,y)\dd y+r(x)\delta_x(\dd y),
    \qquad
    r(x):=1-\int_{\R^d}\alpha(x,y)q(x,y)\dd y.
  \end{equation}
    It can easily be checked that $P$ is in fact $\pi$-reversible.
  We present the full algorithm in \Cref{alg:MHMCMC} and refer to {\cite[Section 7.3]{RC2004}} for more details.

\begin{algorithm}[t]
\begin{algorithmic}[1]
\Statex\textbf{Input:} \begin{itemize}
\item target %
  density $\pi$ on $\R^d$
 \item proposal %
   kernel $Q$ with density $q:\R^d\times\R^d\to(0,\infty)$
 \item initial probability distribution $\pi_0$ on $\R^d$
 \end{itemize}
 \Statex \textbf{Output:} Markov chain $(X_k)_{k\in\{1,\dots,N\}}$ in state space $\R^d$
\State draw $x_0\sim\pi_0$ and set initial state $X_0=x_0$
\For{$k=0,\dots, N$}
\State given $X_k = x_k$ %
draw proposal
        $y_{k+1}\sim Q(x_k,\cdot)$
        \State compute acceptance probability
        $\alpha(x_k,y_{k+1})$ in \eqref{eq:alpha0}
	\State draw $u\sim {\rm U}([0,1])$ and set
	$$ 
	X_{k+1} = \begin{cases} 
	y_{k+1} & \text{if}\ u\le \alpha(x_k,y_{k+1})\\ 
	x_k & \text{else}
	\end{cases}
	$$
\EndFor
\end{algorithmic}
 \caption{Metropolis-Hastings}\label{alg:MHMCMC}
\end{algorithm}

  It is left to choose a suitable proposal kernel $Q$.
  As we recall next, ergodicity is already ensured if $Q(x,\cdot)$ has a positive
  Lebesgue density $q$, i.e. $q(x,y) > 0$ for all $x,y\in\R^d$.
  Nonetheless, in practice the efficiency of the algorithm crucially
  depends on the choice of
  $Q$. A standard (albeit crude) proposal satisfying the positivity
  condition is $Q(x,\cdot)=\Nv(x,h{\rm Id}_d)$,
  $h>0$, also known as Random Walk-MH algorithm.

  \begin{theorem}[{\cite[Section 6.7.2 and 7.3.2]{RC2004}}]\label{thm:MH_Conv}
    Let $\pi$ be absolutely continuous and  
    let $Q$ possess a positive Lebesgue density $q:\R^d\times\R^d\to (0,\infty)$, and let  $\pi_0$ be any initial probability distribution.
    Then, the Markov chain $(X_k)_{k\in\N}$ %
    generated by \Cref{alg:MHMCMC}
    with $X_0\sim\pi_0$
    \begin{enumerate}
 \item\label{item:MH1} %
 is ergodic \eqref{eq:ergodic} and satisfies a strong law of large numbers \eqref{eq:LLN},
\item\label{item:MH2} satisfies the central limit theorem
  \eqref{eq:CLT} for any $F\in L_\pi^2(\R)$ for which
  $\sigma_F^2$ %
  in
  \eqref{eq:sigmaF}
  is nonzero and finite.
    \end{enumerate}
\end{theorem}

  \subsection{Metropolis-adjusted Langevin algorithm}
  A popular proposal kernel $Q$ is obtained through the Euler-Maruyama
  discretization
\begin{equation}\label{eq:MALA}
	X_{k+1} = X_k + h\nabla\log\pi(X_k)+\sqrt{2h} \xi_{k+1}, \quad \xi_{k+1}\sim\Nv(0,{\mathrm{Id}_d})
      \end{equation}
      of the Langevin dynamics \eqref{eq:langevin} (with
      $C={\rm Id}_d$) introduced in \Cref{ex:langevin}. Here
      $h>0$ is a fixed step size, and generating a Markov chain
      $(X_k)_{k\in\N}$ through \eqref{eq:MALA} is also known as the
      \emph{unadjusted Langevin algorithm (ULA)}.  While the
      continuous dynamics \eqref{eq:langevin} has $\pi$ as an
      invariant distribution, see e.g., \cite{P14}, it is known that
      the Markov chain \eqref{eq:MALA} has a bias that scales linearly in
      $h$ \cite[Theorem~2]{NEURIPS2019_65a99bb7}.

      Nevertheless, the continuous-time result suggests using \eqref{eq:MALA} as the proposal mechanism, yielding a proposal kernel $Q$ with positive Lebesgue density
\begin{equation*}
	q(x,y) 
    =
    \frac{1}{\left(4\pi h\right)^{d/2}} \exp\left(-\frac{1}{4h}\|y-x-h \nabla\log\pi(x)\|^2\right).
\end{equation*}
\Cref{alg:MHMCMC} with this choice of proposal kernel is known as MALA. %
According to \Cref{thm:MH_Conv}, and
contrary to the Markov chain \eqref{eq:MALA}, the Metropolised Markov
chain generated by \Cref{alg:MHMCMC} necessarily does have $\pi$ as its invariant
distribution. Moreover, it satisfies ergodicity and a law of large
numbers. 
Furthermore, it is known that for sufficiently large $d$ the best performance in terms of a small asymptotic variance $\sigma_F^2$ in \eqref{eq:CLT} is obtained for choosing the step size $h$ such that the average acceptance rate $\bar{\alpha} = \int_{\R^d}\int_{\R^d} \alpha(x,y) q(x,y) \dd y \, \pi(x) \dd x$ is about $57.4\%$ \cite{RoRo2001}.

\section{Metropolis-adjusted interacting particle
  sampling}\label{sec:MAIPS}

We consider an interacting particle system \eqref{eq:SDE_Ensemble}
resulting from a finite ensemble approximation of suitable
$\pi$-invariant McKean-Vlasov dynamics as in \eqref{eq:SDE_infty},
and its time discretization \eqref{eq:discrete} with step size $h>0$,
i.e.\
\begin{equation}\label{eq:SDE_M_Ensemble}
  \XX_{k+1}=\XX_k+h\PPhi(\XX_k)+\sqrt{h\SSigma(\XX_k)}\zzeta_{k+1}.
\end{equation}

For many relevant particle methods, there exist
\begin{equation*}
\Phi:\R^d\times\R^{(M-1)d}\to\R^d,\qquad
\Sigma:\R^d\times \R^{(M-1)d}\to \CM(\R^{d})
\end{equation*}
such that the drift $\PPhi:\R^{Md}\to\R^d$ and diffusion
$\SSigma:\R^{Md}\to\R^{Md\times Md}$ can be written as 
\begin{equation}\label{eq:PhiM_SigmaM_product}
  \PPhi(\XX) = \begin{pmatrix}
    \Phi(X^{(1)},\XX^{-(1)})\\
    \vdots\\
    \Phi(X^{(M)},\XX^{-(M)})
  \end{pmatrix},
  \qquad
  \SSigma(\XX)=%
  \begin{pmatrix}
    \Sigma(X^{(1)},\XX^{-(1)}) &&\\
    &\ddots&\\
    &&\Sigma(X^{(M)},\XX^{-(M)})
  \end{pmatrix},
\end{equation}
where $\XX\in\R^{Md}$ denotes again the ensemble of particles
  $X^{(1)},\dots,X^{(M)}\subseteq\R^d$ and we use the notation
  introduced in \eqref{eq:XX}. In this case the discrete dynamical
system \eqref{eq:SDE_M_Ensemble} takes the form
\begin{equation}\label{eq:discrete3}
  X_{k+1}^{(i)} = X_k^{(i)}+h\Phi(X_k^{(i)},\XX_k^{(-i)})+\sqrt{h\Sigma(X_k^{(i)},\XX_k^{(-i)})}
  \xi_{k+1}^{(i)}\in\R^d\qquad\forall i\in\{1,\dots,M\}
\end{equation}
with $\xi_k^{(i)}\sim\Nv(0,{\rm I}_d)$, iid for all $i\in\{1,\dots,M\}$, $k\in\N$.

In this section we will discuss an ensemble-wise Metropolization for the general dynamics \eqref{eq:SDE_M_Ensemble} and, in addition, a particle-wise Metropolization for the special case \eqref{eq:discrete3} under assumption \eqref{eq:PhiM_SigmaM_product}. We emphasize that there are common interacting particle systems, such as SVGD, which do not satisfy \eqref{eq:PhiM_SigmaM_product}. %

\begin{remark}
  Equation \eqref{eq:SDE_M_Ensemble} generates an ensemble Markov chain $(\XX_k )_{k\in\N}$ in the product state space $\R^{Md}$. We emphasize that %
  the dynamics of each individual particle $(X^{(i)}_k)_{k\in\N}$, $i=1,\ldots,M$, does not necessarily satisfy the Markov property with respect to the filtration $\mathcal F_t^{(i)} = \sigma(X_s^{(i)},s\le t)$, since the particle-wise drift and diffusion term may depend on the entire ensemble $\XX_k$.
\end{remark}

\subsection{Ensemble-wise Metropolization}\label{sec:ensemble_met}
Our goal in the
  following is to generate a Markov Chain $(\XX_k)_{k\in\N}$ that
  converges in distribution to the product target measure
  \begin{equation*}
    \ppi(\mathbf x)  = \prod_{i=1}^M \pi( x^{(i)}).
\end{equation*}
To this end, given the state $\XX_k$, we use \eqref{eq:SDE_M_Ensemble} %
to generate the ensemble proposal
\begin{equation}\label{eq:Ykp1}
  \YY_{k+1}=\XX_k+h\PPhi(\XX_k)+\sqrt{h\SSigma(\XX_k)}\zzeta_{k+1}
\end{equation}
in the $k$th step of the algorithm. Here as before
$\zzeta_{k+1}\sim\Nv(0,{\rm I}_{Md})$, so that the proposal kernel
$\QQ \colon \R^{Md} \times \cB(\R^{Md}) \to [0,1]$ on the product
space $\R^{Md}$ is given by
\begin{equation*}
  \QQ(\xx,\cdot) = \Nv\left( \xx + h\PPhi(\xx), h
    \SSigma(\xx)\right) \qquad \forall\xx \in\R^{Md}.
\end{equation*}
In order for $\QQ(\xx,\cdot)$ to possess a well-defined Lebesgue density,
throughout we assume the following:
\begin{assumption}\label{assum:Cov}
  For each %
  $\xx \in \R^{Md}$ the matrix
  $\SSigma(\xx) \in \R^{Md\times Md}$ in
  \eqref{eq:SDE_M_Ensemble} is regular.
\end{assumption}%
The Lebesgue density of $\QQ(\xx,\cdot)$ then reads
\begin{equation}\label{eq:qq}
  \qq(\xx, \yy)
   = 
  \frac{1}{\det\left(2\pi\, h\, \SSigma(\xx)  \right)^{1/2}}
  \exp\left(- \frac 1{2h} \left\| \SSigma(\xx)^{-1/2} \left(\yy - \xx - h \PPhi(\xx) \right) \right\|^2 \right)\qquad\forall \xx,\yy\in\R^{Md},
\end{equation}
which is a positive but in general \emph{not} symmetric function,
i.e.\ $\qq(\xx, \yy)>0$ for all $\xx,\yy \in \R^{Md}$ but not
necessarily $\qq(\xx,\yy) \neq \qq(\yy,\xx)$ if $\xx\neq\yy$.  The
  existence of a density allows to introduce a Metropolization step as
  follows: The ensemble proposal $\YY_{k+1}$ in \eqref{eq:Ykp1} is accepted
  with probability
\begin{equation}\label{eq:alpha}
\bs \alpha(\XX_k,\YY_{k+1}) := \begin{cases}
  \min\left(1,\frac{\ppi(\YY_{k+1}) \, \qq(\YY_{k+1},\XX_k)}{\ppi(\XX_k) \, \qq(\XX_k,\YY_{k+1})}\right) &\text{if }\ppi(\XX_k) \qq(\XX_k,\YY_{k+1})>0\\
  1 &\text{else,}
  \end{cases}
\end{equation}
in which case we set $\XX_{k+1}=\YY_k$. Otherwise $\XX_{k+1}=\XX_k$.
The resulting Metropolis-adjusted interacting particle sampling method
is summarised in \Cref{alg:MA-particle}.  

The corresponding transition kernel
$\PP \colon \R^{Md} \times \mathcal B(\R^{Md}) \to [0,1]$ of the
ensemble Markov chain $(\XX_k)_{k\in\N}$ generated by
\Cref{alg:MA-particle} is
\[
  \PP(\xx, \mathrm d \yy) = \bs \alpha(\xx,\yy)\, \qq(\xx, \yy)\, \dd
  \yy + \rr(\xx)\, \delta_\xx(\dd \yy), \qquad \rr(\xx) := 1 -
  \int_{\R^{Md}} \bs \alpha(\xx,\yy) \qq(\xx, \yy)\, \dd \yy.
\]
Note that the construction allows for general proposal kernels $\QQ$
which are dominated by the Lebesgue measure.

\begin{algorithm}[h]
  \begin{algorithmic}[1]
    \Statex \textbf{Input:}
    \begin{itemize}
    \item product target %
      density
      $\ppi = \bigotimes_{i=1}^M \pi$ on $\R^{Md}$
    \item proposal kernel %
      $\QQ$ with
      density $\qq:\R^{Md}\times\R^{Md}\to(0,\infty)$
      in \eqref{eq:qq}
\item initial probability distribution $\pi_0$ on $\R^d$
\end{itemize}
\Statex \textbf{Output:} ensemble Markov chain
$(\XX_k)_{k\in\{1,\dots,N\}}$ in state space $\R^{Md}$
\State %
draw $\xx_0\sim \otimes_{i=1}^M\pi_0$ and set initial state $\XX_0=\xx_0$
\For{$k=0,\dots,N$}
\State %
given $\XX_k=\xx_k$ draw proposal $\yy_{k+1}\sim \QQ(\xx_k,\cdot)$ %
\State compute acceptance
probability $\aalpha(\xx_k,\yy_{k+1})$ in \eqref{eq:alpha}
	
         \State draw $u\sim{\rm U}([0,1])$ and set
	$$ 
	\XX_{k+1} = \begin{cases} \yy_{k+1} &\text{if}\ u\le
          \aalpha(\xx_k,\yy_{k+1})\\
          \xx_k &\text{else}
	\end{cases}
	$$
        \EndFor
      \end{algorithmic}
      \caption{Ensemble-wise Metropolized interactive particle sampling
      }\label{alg:MA-particle}
    \end{algorithm}

     \Cref{alg:MA-particle} can be interpreted as a
      Metropolis-Hastings method to sample from the product measure
      $\ppi=\otimes_{i=1}^M\pi$ in the product space
      $\R^d\times\cdots\times\R^d\simeq\R^{Md}$.
      Hence the following
      convergence result is an immediate consequence of
      \Cref{thm:MH_Conv}:

    \begin{corollary}\label{cor:MAPS}
    Let $\pi$ possess a Lebesgue probability
    density on $\R^d$, and let $\pi_0$ be any probability distribution
    on $\R^d$. Let $\QQ$ possess a positive Lebesgue density
    $\qq:\R^{Md}\times\R^{Md}\to (0,\infty)$, and let
    $(\XX_k)_{k\in\N}$ be generated by \Cref{alg:MA-particle}. Then
    \begin{enumerate}
    \item the Markov chain $(\XX_k)_{k\in\N}$ %
      is ergodic and satisfies the strong law of large numbers
  \begin{equation}\label{eq:MAPS_LLN}
    \SM_N(F) := \frac1{NM} \sum_{k=1}^N \sum_{i=1}^M F\left( X^{(i)}_{k} \right) 
    \ \xrightarrow[N\to\infty]{\text{a.\ s.}} 
    \E_\pi[F] \qquad\forall F \in L^1_\pi(\R),
  \end{equation}
\item if additionally %
$F\in L_\pi^2(\R)$ satisfies %
    \begin{equation}\label{eq:MAPS_var}
      \bs \sigma_F^2 
      := 
      \V_\pi[F] \ \left[ 
      1 + 2 \sum_{k=1}^\infty 
      \Corr\left( \frac 1M \sum_{i=1}^M F\left( \tilde X^{(i)}_{0} \right), \frac 1M \sum_{j=1}^M F\left( \tilde X^{(j)}_{k} \right) \right) \right]
      \in (0,\infty),
    \end{equation}
    where $(\tilde \XX_k)_{k\in\N}$ denotes the 
    the stationary Markov chain generated by \Cref{alg:MA-particle} with $\pi_0 = \pi$, %
    then there holds the central limit theorem
    \begin{equation}\label{eq:MAPS_CLT}
      \sqrt{N} \left( \SM_N(F) - \E_\pi[F]\right) 
      \ \xrightarrow[N\to\infty]{\mathcal L} \ \Nv\left(0, \frac 1M \bs\sigma_F^2 \right).
    \end{equation}
    \end{enumerate}

  \end{corollary}

  \begin{proof}
    For $\xx=(x^{(1)},\dots,x^{(M)})\in\R^{Md}$ with $x^{(j)}\in\R^d$ define $G:\R^{Md}\to\R$ via
    \begin{equation*}
      G(\xx):=\frac{1}{M}\sum_{j=1}^M F(x^{(i)}).
    \end{equation*}
    Then with $\ppi:=\otimes_{j=1}^M\pi$ it holds $G\in L_{\ppi}^1(\R)$ since $F\in L_\pi^1(\R)$.
    According to \Cref{thm:MH_Conv} \ref{item:MH1} this implies
    \begin{equation*}
      \SM_N(F) = 
      \frac{1}{N}\sum_{k=1}^NG(\XX_k)
      \xrightarrow[N\to\infty]{\text{a.\ s.}}
      \E_{\ppi}[G] = \E_\pi[F],
    \end{equation*}
    which shows \eqref{eq:MAPS_LLN}. Statement \eqref{eq:MAPS_CLT} follows similarly with \Cref{thm:MH_Conv} \ref{item:MH2} and $\V_{\bs\pi}[G] = \frac 1M \V_{\pi}[F]$.
  \end{proof}

  The factor $\frac 1M$ in the asymptotic variance
  $\frac 1M \bs\sigma_F^2$ of $\SM_N (F)$ in \eqref{eq:MAPS_CLT}
  reflects the reduced variance due to running $M$ Markov chains
  instead of just one.  Moreover, we note that the particles
  $(X_k^{(i)})_{i=1}^M$ within the ensemble $\XX_k$ become iid $\pi$
  distributed as $k\to\infty$, since the limit distribution $\ppi$ is
  of product form.
  However, $X_k^{(i)}$ and $X_{k+1}^{(j)}$ do in general not become
  uncorrelated for $i\neq j$ as $k\to\infty$, which %
  is due to
  the chains interacting:

\begin{example}
  Let $\pi = \Nv(0,1)$, $M=2$ and consider the $\ppi=\pi\otimes\pi$-invariant transition kernel

  \[
    \PP(\xx,\cdot) = \Nv\left(A \xx, \frac 12 I_2 \right), \qquad A
    = \begin{pmatrix}
      \frac 12 & \frac 12\\
      \frac 12 & -\frac 12
    \end{pmatrix}.
  \]
  Then for the Markov chain $(\XX_k)_{k\in\N}$ generated with transition
  kernel $\PP$ holds for all $k\in\N$ that
  $\XX_k \sim \ppi$ implies  $\XX_{k+1}\sim \ppi$, but $\Corr(X_k^{(1)}, X_{k+1}^{(2)}) = \Corr(X_k^{(2)}, X_{k+1}^{(1)}) =   \frac 12$.
\end{example}

Let us comment on the special case of non-interacting particles,
  i.e., assume $(X_k^{(i)})_{k\in\N}$
  are %
  mutually independent Markov chains for all $i=1,\ldots,M$.
  Moreover, assume for simplicity that $X_0^{(i)}\sim\pi$ iid,
  $i=1,\dots,M$.
For such an ensemble of iid Markov chains the estimator
\[
    \SM_N(F)
    =
    \frac 1M \sum_{i=1}^M \left(\frac 1N \sum_{k=1}^N F(X_k^{(i)}) \right)
\]
is simply an average of $M$ iid path average estimators $\frac 1N \sum_{k=1}^N F(X_k^{(i)})$, $i=1,\dots,M$, of the independent particle Markov chains $(X_k^{(i)})_{k\in\N}$.
Thus, the resulting asymptotic variance $\bs \sigma_F^2$ in \eqref{eq:MAPS_var} can %
be written as
\[ 
    \bs \sigma_F^2 = 
    \V_\pi[F] \ \left[ 1 + 2 \sum_{k=1}^\infty
    \Corr\left(F\left( X^{(i)}_{0} \right), F\left( X^{(i)}_{k}
      \right) \right ) \right] = \sigma_F^2
\]
where $i\in\{1,\ldots,M\}$ is arbitrary and $\sigma_F$ as in \eqref{eq:sigmaF}.

In particular, for noninteracting particle
systems, we have for sufficiently large $N$
\[ \V(S_{ MN }(F)) \approx \V(\SM_{N}(F)) \] where
\begin{equation*}
  S_{MN}(F):=\frac{1}{MN}\sum_{k=1}^{MN}F(\tilde X_k),
\end{equation*}
with $(\tilde X_k)_{k=1}^{MN}$ being %
a \emph{one-particle} Markov chain generated by \Cref{alg:MHMCMC} with $\tilde X_0\sim\pi$, as already mentioned in \cite{GW2010}.
In our numerical experiments, see \Cref{sec:numerics}, we observe that the interaction of the $M$ particles can give improvements in the form
\[ 
    \V(S_{ MN }(F)) > \V(\SM_{N}(F)) 
\]
due to allowing each particle chain $(X_k^{(i)})_{k\in\N}$ to take
larger steps resulting from exploiting approximate information on
$\pi$ provided by the ensemble in the proposal kernels $\QQ$.
However, a rigorous proof of this statement is beyond the scope of
this work.

We conclude this subsection with a remark on \Cref{assum:Cov}.

 \begin{remark}\label{rem:subspace_property}
   \Cref{assum:Cov} is not necessarily satisfied in common particle
   dynamics. For example, $\SSigma(\xx)$ often (see 
     \Cref{sec:examples} ahead) takes the form
   \eqref{eq:PhiM_SigmaM_product} with%
   \[ \Sigma(x^{(i)}, \xx^{-(i)}) %
     = \frac 1{M-1} \sum_{i=1}^M (x^{(i)} - \bar\xx)(x^{(i)} -
     \bar\xx)^\top \in \R^{d\times d}, \qquad \bar\xx = \frac 1M
     \sum_{i=1}^M x^{(i)} \in \R^{d}, \] as the empirical covariance
   of the ensemble $\xx$.
   If $M < d$, then
   $\Sigma(x^{(i)},\xx)$ is degenerate and so is $\SSigma(\xx)$. Apart
   from the fact that $\qq$ is not well-defined in \eqref{eq:qq}
   (i.e., we would have to work with the Moore-Penrose inverse of
   $\SSigma(\xx)$), a more important consequence of the degeneracy of
   $\SSigma(\xx)$ is related to the so-called \emph{subspace
     property} of the ensemble dynamics, see e.g. \cite{ILS2013}. This
   means that for certain $\PPhi$
   each particle $X_{k+1}^{(i)}$ of the
   ensemble Markov chain $(\XX_k )_{k\in\N}$ remains in the span of the
   initial ensemble
   $V_0 := \mathrm{span}\,(X_{0}^{(1)}, \ldots, X_{0}^{(M)})$. Thus,
   the ensemble Markov chain will not explore the entire space
   $\R^{Md}$, but will stay within
   $V_0\times \cdots \times V_0 = V_0^M$. In this case, we can at best
   hope for an approximate sample of the (normalized) conditional
     density $\pi|_{V_0}$ %
   but not a sample of $\pi$ itself.
\end{remark}

\subsection{Particle-wise Metropolization}\label{sec:particle_met}
Although, we can control the average acceptance probability for ensemble-wise Metropolization via a sufficiently small stepsize $h$,  
a particle-wise acceptance or rejection seems more flexible and thus preferable.
In fact, previous ensemble MCMC algorithms introduced in \cite{GW2010, LMW2018, CW2021, DS2022} are all based on a \emph{sequentially particle-wise Metropolization} which can be understood as a Metropolis-within-Gibbs mechanism applied to particles as the ``coordinates'' or components of the ensemble state \cite{RoRo2006}.
We now discuss this approach in detail
providing a theoretical description and again a convergence result.
Additionally, we briefly comment on the possibility of
a \emph{simultaneous} (rather than sequential)
Metropolization below.
 
To define particle-wise Metropolization we assume the structural condition
  \eqref{eq:PhiM_SigmaM_product} on the drift $\PPhi$ and diffusion
  $\SSigma$, so that \eqref{eq:discrete3} describes the particle
  dynamics.
In this case, the corresponding proposal kernel $\QQ$ for the ensemble Markov chain follows a product form:
\begin{equation}\label{eq:proposal_structure}
  \QQ(\xx, \dd \yy)
  =
  Q_{\xx^{-(1)}}(x^{(1)}, \dd y^{(1)})
  \otimes
  \cdots
  \otimes	
  Q_{\xx^{-(M)}}(x^{(M)}, \dd y^{(M)}),
\end{equation}
where %
\[
  Q_{\xx^{-}}(z, \cdot) := \Nv(z + h \, \Phi (z, \xx^{-} ),
    h\, \Sigma (z, \xx^{-} ) )\qquad 
  \forall z\in\R^d,~\xx^{-}\in\R^{(M-1)d}.
\]
Under \Cref{assum:Cov}, i.e., $\Sigma(z, \xx^{-})$ is regular for any $z\in\R^d$ and $\xx^-\in\R^{(M-1)d}$, $Q_{\xx^{-}}(z,\cdot)$ %
possesses the Lebesgue density
\begin{equation}\label{eq:qxx}
  q_{\xx^-}(z,y)
  \ = \
  \frac{1}{\det\left(2\pi\, h\, \Sigma(z, \xx^-)  \right)^{1/2}}
  \exp\left(- \frac 1{2h} \left\| \Sigma(z, \xx^-)^{-1/2} \left(y - z - h \Phi(z, \xx^-) \right) \right\|^2 \right) > 0.
\end{equation}

We then introduce the acceptance probability for the $i$th particle as
\begin{equation}\label{eq:alpha2}
  \alpha_{\xx^{-(i)}}(x^{(i)},y) :=
	 \begin{cases}
           \min\left(1, \frac{\pi(y) \, q_{\xx^{-(i)}}(y, x^{(i)})}{\pi( x^{(i)}) \, q_{\xx^{-(i)}}(x^{(i)},y)}\right) & %
           \text{if}\ \pi(x^{(i)}) \, q_{\xx^{-(i)}}(x^{(i)},y) > 0\\
           1 &\text{else}
	 \end{cases}
    \qquad \forall i\in\{1,\dots,M\},
\end{equation}
and define the MH transition kernel $P_{\xx^{-(i)}} \colon \R^{d} \times \mathcal B(\R^{d})\to[0,1]$, $\xx \in \R^{Md}$, for the $i$th particle as
\begin{align}\label{eq:Pi}
      P_{\xx^{-(i)}}(x^{(i)}, \dd y) 
      := 
      \alpha_{\xx^{-(i)}}(x^{(i)},y)\, Q_{\xx^{-(i)}}(x^{(i)}, \dd y) 
      + 
      r_{\xx^{-(i)}}(x^{(i)})
      \delta_x(\dd y), 
\end{align}
where $r_{\xx^{-(i)}}(x^{(i)}) :=  1 - \int_{\R^d} \alpha_{\xx^{-(i)}}(x^{(i)},y)\, Q_{\xx^{-(i)}}(x^{(i)}, \dd y)$ denotes the rejection probability just for the $i$th particle. 
By construction $P_{\xx^{-}}$ is $\pi$-reversible for any $\xx^-\in\R^{(M-1)d}$.
Based on $P_{\xx^{-(i)}}$ we can now describe a sequential and simultaneous Metropolization.

\paragraph{Sequential particle-wise Metropolization}
For this method interpret the particles $X^{(i)}$ as ``block coordinates'' in the product state space $\R^{Md}$, and update them sequentially as in classical Metropolis-within-Gibbs sampling.
To this end, we define transition kernels $\PP^{(i)}\colon \R^{Md} \times \cB(\R^{M,d})\to[0,1]$, $i=1,\ldots,M$, for the whole ensemble, which only update the $i$th particle via
\begin{equation}\label{eq:PPi}
    \begin{split}
      \PP^{(i)}(\xx, \mathrm d \yy)
      :=&
      \delta_{x^{(1)}}(\mathrm d y^{(1)})
      \otimes\cdots \otimes
      \delta_{x^{(i-1)}}(\mathrm d y^{(i-1)})
      \otimes
      P_{\xx^{-(i)}}(x^{(i)} , \mathrm d y^{(i)})\\
      &\otimes
      \delta_{x^{(i+1)}}(\mathrm d y^{(i+1)})
      \otimes \cdots \otimes
      \delta_{x^{(M)}}(\mathrm d y^{(M)}).
      \end{split}
\end{equation}
The sequential particle-wise Metropolization of the interacting particle system \eqref{eq:SDE_M_Ensemble} is defined as the sequential application of the transition kernels %
$\PP^{(i)}$
\begin{align}\label{eq:PP_seq}
      \PP_\mathrm{seq}(\xx, \mathrm d \yy)
      := \PP^{(1)}\cdots\PP^{(M)}(\xx, \mathrm d \yy).
\end{align}
We point out that $\PP_{\rm seq}$ is $\ppi$-invariant since all $\PP^{(i)}$ are $\ppi$-invariant. 
In general $\PP_{\rm seq}$ is not $\ppi$-reversible however (cf.\ deterministic scan Gibbs sampling). 
Reversibility can be retained if the order updating for the particles is not deterministic but random (cf.\ random scan Gibbs sampling), i.e., with $\Psi_M$ denoting the set of all $M!$ permutations $\psi\colon \{1,\ldots,M\} \to \{1,\ldots,M\}$ we set
\begin{align}\label{eq:PP_seq_2}
      \tilde\PP_\mathrm{seq}(\xx, \mathrm d \yy)
      := 
      \sum_{\psi \in \Psi_M} \frac 1{M!} \, \PP^{(\psi(1))}\cdots\PP^{(\psi(M))}(\xx, \mathrm d \yy).
\end{align}
An algorithmic description of $\PP_\mathrm{seq}$ in \eqref{eq:PP_seq} is given in \Cref{alg:MA-particle_seq}.

\begin{algorithm}[t]
  \begin{algorithmic}[1]
    \Statex \textbf{Input:} \begin{itemize}
    \item target density $\pi$ on $\R^d$
    \item ensemble dependent proposal kernel $Q_{\xx^-}$ with density
      $q_{\xx^-}:\R^{d}\times\R^{d}\to (0,\infty)$ in \eqref{eq:qxx}
    \item initial probability distribution $\pi_0$ on $\R^d$
    \end{itemize}
    \Statex \textbf{Output:} ensemble Markov chain
    $(\XX_k)_{k\in\{1,\dots,N\}}$ in state space $\R^{Md}$
    \State
    draw $\xx_0\sim \otimes_{i=1}^M\pi_0$ and set initial state $\XX_0=\xx_0\in\R^{Md}$
    
      \For{$k=0,\dots, N$}
      \State 
      given $\XX_k = \xx_k$ initialize $\xx = (x^{(1)}_k,\ldots,x^{(M)}_k) $
      \For{$i = 1,\dots,M$}
      \State
      draw proposal $y^{(i)} \sim Q_{\xx^{-(i)}}(x^{(i)},\cdot)$
      \State 
      compute particle acceptance probability $\alpha_{\xx^{-(i)}}(x^{(i)}, y^{(i)})\in [0,1]$ in \eqref{eq:alpha2}
      \State 
      draw $u_i\sim{\rm U}([0,1])$ and set
      \[
          x^{(i)} = \begin{cases}
          y^{(i)} &\text{if } u_i \le \alpha_{\xx^{-(i)}}(x^{(i)}, y^{(i)})\\
          x^{(i)} &\text{else}
	       \end{cases}
      \]
	   \EndFor
    \State set $\XX_{k+1} = \xx$
	\EndFor
      \end{algorithmic}
      \caption{Sequentially particle-wise Metropolized interactive particle sampling
      }\label{alg:MA-particle_seq}
    \end{algorithm}

    It is not possible to apply \Cref{thm:MH_Conv} to prove
      the ergodicity of the chains generated by
      \Cref{alg:MA-particle_seq}, as neither $\PP_\mathrm{seq}$ nor
      $\tilde\PP_\mathrm{seq}$ are %
      MH transition kernels. We can however leverage the results of
      \cite{RoRo2006} about Metropolis-within-Gibbs algorithms to
      establish ergodicity:

\begin{theorem}[Cf. \cite{RoRo2006}]\label{theo:MAPS2_seq}
  Let the $\ppi$-invariant Markov chain $(\XX_k)_{k\in\N}$ be
  generated by \Cref{alg:MA-particle_seq}.  Assume that the
  Lebesgue density $q_{\xx^-}$ of $Q_{\xx^-}$ satisfies
  $q_{\xx^-}(z,y) > 0$ for all $z$, $y \in \R^d$ and all
  $\xx^{-} \in \R^{(M-1)d}$, and that
  \begin{equation}\label{eq:assum_Dn}
    \lim_{n\to\infty} \Prob\left( \exists k \in \{1,\ldots,n\} \text{ s.t. } \forall i=1,\ldots,M \text{ holds } %
    X^{(i)}_{k} \neq X^{(i)}_{k-1} \,\middle |\, \XX_0 = \xx\right)
    =
    1
    \qquad
    \forall \xx\in\R^{Md}.
  \end{equation}
  Then $(\XX_k)_{k\in\N}$ is ergodic and satisfies a law of large
    numbers \eqref{eq:MAPS_LLN}.%
\end{theorem}

\begin{remark}
  A Markov chain with the transition kernel $\tilde\PP_\mathrm{seq}$ instead of $\PP_\mathrm{seq}$ satisfies besides ergodocity and a law of large numbers \eqref{eq:MAPS_LLN} also a central limit theorem \eqref{eq:MAPS_CLT} given \eqref{eq:MAPS_var} and given the assumptions of \Cref{theo:MAPS2_seq}.
   However, also for the non-reversible transition kernel $\PP_\mathrm{seq}$ (i.e.\
  \Cref{alg:MA-particle_seq}) sufficient conditions for the central
  limit theorem \eqref{eq:MAPS_CLT} with $\bs\sigma^2_F$ as in
  \eqref{eq:MAPS_var} can be stated. We refer to \cite[Section
  5]{RoRo2004} for more details.
\end{remark}

The additional assumption \eqref{eq:assum_Dn} for ergodicity in \Cref{theo:MAPS2_seq} is rather mild. 
It states that for any initial ensemble $\XX_0 = \xx$, the probability of all proposed particles being accepted at least once for $k\in\{1,\dots,n\}$ tends to $1$ as $n\to\infty$.
In practice, this should be satisfied in any reasonable setting and can also be easily checked during simulation (by observing whether all proposed particles are accepted at least once).

\paragraph{Simultaneous particle-wise Metropolization}
For interacting particle systems described by
\eqref{eq:SDE_M_Ensemble} it seems %
unintuitive to update the particles in the ensemble sequentially,
i.e., to use a different drift and diffusion term for each of the
particles due to interaction and the sequentially updated
ensemble. Moreover, from a computational perspective, a
  \emph{simultaneous} (rather than sequential) particle-wise
  Metropolization would be most desirable as it allows for the
  parallel processing of all particles within the ensemble: this
  method, in each step of the algorithm, decides for each particle in
  the ensemble independently, whether the proposal for
  this particle is accepted or rejected. %
Formally, this is described by the ensemble transition operator 
\[
    \PP_\mathrm{sim}(\xx, \mathrm d \yy)
    =
    P_{\xx^{-(1)}}(x^{(1)} , \mathrm d y^{(1)})
    \otimes
    \cdots
    \otimes
    P_{\xx^{-(M)}}(x^{(M)} , \mathrm d y^{(M)}).
\]
However, it turns out that this approach %
in general fails to ensure ergodicity. A detailed discussion
including counterexamples is provided in \Cref{sec:pw-sim}.

\subsection{Block-wise Metropolization}\label{sec:block_met}
The sequential particle-wise Metropolization requires an inner loop over all $M$ particles which results also in the sequential computation of $M$ different drift and diffusion terms $\Phi(x^{(i)}, \xx^{-(i)})$ and $\Sigma(x^{(i)}, \xx^{-(i)})$, $i=1,\dots,M$.
In order to allow for parallelization, we discuss a \emph{block-wise Metropolization} which can be seen as bridging the ensemble-wise and particle-wise variants.
Instead of applying the particle-wise transition
kernels $P_{\xx^{-(i)}}$ sequentially, we propose to apply 
transition kernels block-wise.

To explain the idea, we first fix a ``block''
  $\bb\subseteq\{1,\dots,M\}$ of indices.  We then write
\begin{equation*}%
    \xx^{(\bb)} 
    := 
    (x^{(j)})_{j \in \bb}
    \in  \R^{|\bb|d}
    \qquad
    \text{and}
    \qquad
    \xx^{-(\bb)}
    :=
    ( x^{(j)} )_{j \in \{1,\ldots,M\}\setminus\bb}
    \in \R^{(M-|\bb|)d}\, .
\end{equation*}
In the following we assume again a product form \eqref{eq:proposal_structure} of the ensemble proposal kernel $\QQ$.
We then define corresponding block-wise proposal kernels
$\QQ_{\xx^{-(\bb)}}\colon \R^{|\bb| d} \times \cB (\R^{|\bb| d})\to[0,1]$
by
\begin{equation}\label{eq:Qx_block}
    \QQ_{\xx^{-(\bb})}(\zz, \cdot) 
    := \Nv\left(\zz + h \, \PPhi(\zz, \xx^{-(\bb)}),
    h\, \SSigma (\zz, \xx^{-(\bb)})\right)
    \qquad 
    \forall \zz\in\R^{|\bb|d},~\xx^{-(\bb)}\in\R^{(M-|\bb|)d}
  \end{equation}
where
\begin{equation*}%
  \PPhi(\zz, \xx^{-(\bb)})
  := \begin{pmatrix}
    \Phi\left( z^{(1)}, \, \left(\zz^{-(1)},\xx^{-(\bb)} \right) \right)\\
    \vdots\\
    \Phi\left( z^{(|\bb|)}, \, \left(\zz^{-(|\bb|)},\xx^{-(\bb)}\right) \right)
  \end{pmatrix}
\end{equation*}
and
\begin{equation*}
  \SSigma(\zz, \xx^{-(\bb)})
  := 
  \begin{pmatrix}
    \Sigma\left(z^{(1)}, \, \left(\zz^{-(1)},\xx^{-(\bb)} \right) \right) &&\\
    &\ddots&\\
    &&\Sigma\left( z^{(|\bb|)}, \, \left(\zz^{-(|\bb|)},\xx^{-(\bb)}\right) \right)
  \end{pmatrix}.
\end{equation*}

Under \Cref{assum:Cov}, for every $\xx^{-(\bb)} \in \R^{(M-|\bb|)d}$,
the proposal kernel $\QQ_{\xx^{-(\bb)}}$ possesses a Lebesgue density $q_{\xx^{-(\bb)}} \colon \R^{|\bb|d} \times \R^{|\bb|d}\to (0,\infty)$ 
\[ 
  q_{\xx^{- (\bb)}}(\zz,\yy)
  \ = \
  \frac{1}{\det\left(2\pi\, h\, \SSigma(\zz, \xx^{-(\bb)})  \right)^{1/2}}
  \exp\left(- \frac 1{2h} \left\| \SSigma(\zz, \xx^{-(\bb)})^{-1/2} \left(\yy - \zz - h \PPhi(\zz, \xx^{-(\bb)}) \right) \right\|^2 \right). 
\]
This allows to define block-wise acceptance probabilities of the form
\begin{equation}\label{eq:alpha_block}
    \bs \alpha_{\xx^{-(\bb)}}(\zz, \yy) :=
	   \begin{cases}
           \min\left(1, \frac{{\ppi}(\yy) \, q_{\xx^{-(\bb)}}(\yy, \zz)}{\ppi(\zz) \, q_{\xx^{-(\bb)}}(\zz,\yy)}\right) & 
           \text{if}\ {\ppi}(\zz) \, q_{\xx^{-(\bb)}}(\zz,\yy) > 0\\
           1 &\text{else}
	 \end{cases}
\end{equation}
with $\yy, \zz\in\R^{|\bb| d}$, where %
${\ppi}(\zz) := \prod_{i=1}^{|\bb|} \pi(z^{(i)})$ for
$\zz\in\R^{|\bb| d}$.  We finally define a block-wise MH transition kernel
$\PP_{\xx^{-(\bb)}}\colon \R^{|\bb| d} \times \mathcal{B}
(\R^{|\bb|} d)\to[0,1] $, $\xx^{-(\bb)} \in \R^{(M-|\bb|)d}$ describing the update
of the $|\bb|$ particles within a block $\mathbf b$ by
\[
      \PP_{\xx^{-(\bb)}}(\zz, \dd \yy) 
      := 
      \bs \alpha_{\xx^{-(\bb)}}(\zz,\yy)\, \QQ_{\xx^{-(\bb)}}(\zz, \dd \yy) 
      + 
      r_{\xx^{-(\bb)}}(\zz) \delta_{\zz}(\dd \yy), 
\]
where $r_{\xx^{-(\bb)}}(\zz) :=  1 - \int_{\R^{|\bb|d}} \alpha_{\xx^{-(\bb)}}(\zz,\yy)\, \QQ_{\xx^{-(\bb)}}(\zz, \dd \yy)$.
Again, the transition kernel $\PP_{\xx^{-(\bb)}}$ is reversible w.r.t.~$\ppi = \Pi_{i=1}^{|\bb|} \pi$ by construction for any $\xx^{-(\bb)} \in \R^{(M-|\bb|)d}$.

Consider now a partition of the ensemble of $M$ particles via
\begin{equation}\label{eq:blocks}
  \{1,\dots,M\}=\dot\bigcup_{i=1}^L \bb_i.
\end{equation}
Analogous to \eqref{eq:PPi} we define transition kernels $\PP^{(\bb_j)}\colon \R^{M d}\times \mathcal B(\R^{M d})\to [0,1]$, which only update
particles corresponding to the block $\bb_j$, via
\begin{equation*}
  \PP^{(\mathbf b_j)}(\xx,\dd \yy)
  = \bigotimes_{k\notin\bb_j}\delta_{x^{(k)}}(\mathrm d y^{(k)})\otimes
  \PP_{\xx^{-(\bb_j)}}(\xx^{(\bb_j)},\dd \yy^{(\bb_j)}).
\end{equation*}
By construction each $\PP^{(\mathbf b_j)}$, $j=1,\ldots,L$, is invariant w.r.t.~product target $\ppi = \Pi_{i=1}^M \pi$. 

The block-wise Metropolization of the interacting particle system
\eqref{eq:SDE_M_Ensemble} is then defined %
as the sequential application of the block transition kernels
$\PP^{({\bb_j})}$, i.e.,
\begin{align}\label{eq:PP_block}
      \PP_\mathrm{block}(\xx, \mathrm d \yy)
      := \PP^{(\bb_1)}\cdots\PP^{(\bb_L)}(\xx, \mathrm d \yy).
\end{align}
We note that ensemble- and particle-wise Metropolization in 
  \Cref{sec:ensemble_met} and \ref{sec:particle_met} can be viewed as
  special cases of the block-wise Metropolization
  with $\bb_j=\{j\}$ for all $j=1,\dots,M$, and $\bb_1=\{1,\dots,M\}$,
  respectively.
Analogous to \Cref{theo:MAPS2_seq} we obtain:

\begin{theorem}[Cf. \cite{RoRo2006}]\label{theo:MAPS_block}
  Let the $\ppi$-invariant Markov chain $(\XX_k)_{k\in\N}$ be
  generated by \Cref{alg:MA-particle_block}. 
    Assume that for any $\bb\subseteq\{1,\dots,M\}$
    the Lebesgue density of $\QQ_{\xx^{-(\bb)}}$ satisfies
    $q_{\xx^{-(\bb)}}(\zz,\yy) > 0$ for all $\zz, \yy \in \R^{|\bb|d}$
    and all $\xx^{-(\bb)} \in \R^{(M-|\bb|)d}$.  If
  \begin{equation}\label{eq:assum_Dn}
    \lim_{n\to\infty} \Prob\left( \exists k \in \{1,\ldots,n\} \text{ s.t. } \forall j=1,\ldots,L \text{ holds } %
    \XX^{(\mathbf b_j)}_{k} \neq \XX^{(\mathbf b_j)}_{k-1} \,\middle |\, \XX_0 = \xx\right)
    =
    1
    \qquad
    \forall \xx\in\R^{Md},
  \end{equation}
  then $(\XX_k)_{k\in\N}$ is ergodic and satisfies a law of large
    numbers \eqref{eq:MAPS_LLN}.
\end{theorem}
Again, also for the non-reversible transition kernel $\PP_\mathrm{block}$ a central limit theorem \eqref{eq:MAPS_CLT} with $\bs\sigma^2_F$ as in \eqref{eq:MAPS_var} can be stated given sufficient conditions. For the latter we refer again to \cite[Section 5]{RoRo2004} for more details.

\begin{remark}
  The definition of $\PP_\mathrm{block}$ can be modified by randomly
  selecting which of the (fixed) blocks $\bb_1,\dots,\bb_L$ to
    update in each step.
  This then yields a reversible transition operator
  $\tilde{\PP}_\mathrm{block}$ satisfying also a central limit theorem
  \eqref{eq:MAPS_CLT} given \eqref{eq:MAPS_var} and given the
  assumptions of \Cref{theo:MAPS_block}.  Moreover, one could also
  think about generalizing the construction by allowing for random
  blocks $\bb$, i.e., updating each time a random selection of $|\bb|$
  particles, where also the number $|\bb|\in\{1,\ldots,M\}$ could be
  drawn randomly %
  in each step. However, we focus in the following on deterministic
  blocks and the deterministic scan block updating.
\end{remark}

\begin{algorithm}[t]
  \begin{algorithmic}[1]
    \Statex \textbf{Input:} \begin{itemize}
    \item target density $\pi$ on $\R^d$
    \item a partition $\dot\bigcup_{j=1}^L\bb_j=\{1,\dots,M\}$
    \item ensemble dependent block proposal kernels $\QQ_{\xx^{-(\bb_j)}}$ with density $\qq_{\xx^{-(\bb_j)}}:\R^{|\bb_j|d}\times\R^{|\bb_j|d}\to (0,\infty)$ for arbitrary $\xx^{-(\bb_j)} \in \R^{(M-|\bb_j|)d}$ and $j \in\{1,\ldots,L\}$ %
    \item initial probability distribution $\pi_0$ on $\R^d$
    \end{itemize}
    \Statex \textbf{Output:} ensemble Markov chain
    $(\XX_k)_{k\in\{1,\dots,N\}}$ in state space $\R^{Md}$
    \State
    draw $\xx_0\sim \otimes_{i=1}^M\pi_0$ and set initial state $\XX_0=\xx_0\in\R^{Md}$
      \For{$k=0,\dots, N$}
      \State 
      given $\XX_k = \xx_k$ initialize $\xx = (x^{(1)}_k,\ldots,x^{(M)}_k) $
      \For{$j = 1,\dots,L$}
      \State
      draw proposal $\yy \sim \QQ_{\xx^{-(\bb_j)}}(\xx^{(\bb_j)},\cdot)$ %
      \State 
      compute block acceptance probability
      $\bs \alpha_{\xx^{-(\bb_j)}}(\xx^{(\bb_j)}, \yy)\in [0,1]$ in \eqref{eq:alpha_block}
      \State 
      draw $u_j\sim{\rm U}([0,1])$ and set
      \[
          \xx^{(\bb_j)} = \begin{cases}
          \yy &\text{if } u_j \le \bs \alpha_{\xx^{-(\bb_j)}}(\xx^{(\bb_j)}, \yy)\\
          \xx^{(\bb_j)} &\text{else}
	       \end{cases}
      \]
	   \EndFor
    \State set $\XX_{k+1} = \xx$
	\EndFor
      \end{algorithmic}
      \caption{Block-wise Metropolized interactive particle sampling
      }\label{alg:MA-particle_block}
    \end{algorithm}

 \section{Algorithmic examples}\label{sec:examples}
 In this section we discuss several interacting particle methods
 which can be used
 to build proposals %
 for the Metropolization schemes
 presented in %
  \Cref{sec:MAIPS}.  
 They are inspired from continuous-time systems 
 \eqref{eq:SDE_infty}, which we translate into a time-discrete
 interacting particle system through ensemble and time discretization.
 For each example, we
   discuss whether the assumptions of \Cref{cor:MAPS} and
   \Cref{theo:MAPS2_seq} are satisfied and whether particle-wise
   metropolization is applicable.

\subsection{Parallel Metropolis-adjusted Langevin algorithm (pMALA)}
For comparison %
we start with a %
\emph{non-interacting} system, namely %
parallel %
MALA. 
This will serve to illustrate the benefits %
of interaction of particles in our numerical examples.

\paragraph{Ensemble update and proposal density}
Recall the particle-wise update scheme from \Cref{ex:ULA}
arising from discretization of the Langevin dynamics
introduced in \Cref{ex:langevin} and \Cref{ex:ULA}, i.e.\
\begin{equation}\label{eq:pMALA}
	X_{k+1}^{(i)} = X_k^{(i)} + h \nabla\log\pi(X_k^{(i)}) + \sqrt{2h} \zeta_{k+1}^{(i)}, 
	\qquad
	k\in\N,\ j = 1,\ldots,M,
\end{equation}
with iid~$\zeta_{k+1}^{(i)}\sim\Nv(0, \mathrm{Id}_d)$. 
This system satisfies the structural condition \eqref{eq:PhiM_SigmaM_product}
and the corresponding proposal kernel $\QQ$ %
has the product form \eqref{eq:proposal_structure} with
($\xx^{-}$-independent) Lebesgue density
\[
	q_{\xx^-}(x,y)
	=
	q(x,y)
    =
    \frac{1}{(4h\pi)^{d/2}}
	\exp\left(-\frac{1}{4h}\|\left(y- (x+h \nabla\log\pi(x))\right)\|^2\right),
	\qquad
	x,y \in \R^d,
\]
i.e., $\Phi(x,\xx) :=  \nabla\log\pi(x)$ and $\Sigma(x,\xx) = 2\mathrm{Id}_d$
in \eqref{eq:PhiM_SigmaM_product}.

\paragraph{Summary} 
By construction \eqref{eq:pMALA} possesses a positive proposal density, thus establishing ergodicity for ensemble-wise Metropolization, see \Cref{cor:MAPS}. Furthermore, the
structural condition in \eqref{eq:PhiM_SigmaM_product} is met, which also ensures ergodicity for both particle- and block-wise Metropolization as stated in \Cref{theo:MAPS2_seq} and \Cref{theo:MAPS_block}.

\subsection{Metropolis-adjusted ALDI (MA-ALDI)}\label{sec:MA-ALDI}

We continue with the ALDI discretization introduced in \Cref{ex:EKS} for
the continuous Wasserstein dynamics described by \eqref{eq:WassersteinDyn}.

\paragraph{Ensemble update and proposal density}
We introduce a sightly modified version of the update scheme \eqref{eq:ALDI}:
\begin{equation}\label{eq:EM_ALDI}
  \begin{aligned}
X_{k+1}^{(i)} = &X_k^{(i)} + h (\gamma {\mathrm{Id}_d} + (1-\gamma)C(\XX_k))\nabla\log\pi(X_k^{(i)})+h(1-\gamma)\frac{d+1}{M}(X_k^{(i)}-m(\XX_k))\\ & + \sqrt{2h(\gamma {\mathrm{Id}_d} + (1-\gamma)C(\XX_k))}\zeta_{k+1}^{(i)}.
\end{aligned}
\end{equation}
Here, $\zeta_{k+1}^{(i)}\sim\Nv(0,{\mathrm{Id}_d})$ iid for all
$i=1,\dots,M$ and $k\ge 0$, and $\gamma\in[0,1]$ is a fixed parameter
that determines the contribution of the covariance matrix $C(\XX_k)$
in the update.

When $\gamma = 0$, the update scheme \eqref{eq:EM_ALDI} reduces to the
ALDI method \eqref{eq:ALDI}, while $\gamma = 1$ corresponds to the
pMALA method discussed in \Cref{ex:ULA} or the previous subsection. By
choosing $\gamma\in[0,1]$, we enable a smooth transition from pMALA to
MA-ALDI. Additionally, the term $\gamma {\mathrm{Id}}$ in
$\Sigma_\gamma(x^{(i)},\xx^{-(i)}) := 2(\gamma {\mathrm{Id}} +
(1-\gamma)C(\xx))$ 
serves as variance inflation
and guarantees the positive-definiteness of
$\Sigma_\gamma(x^{(i)},\xx^{-(i)}) \in \CM^+(\R^d)$ for all
$\xx\in\R^{Md}$ when $\gamma >0$, which breaks the well-known subspace
property \cite{ILS2013}. The structural assumption
\eqref{eq:PhiM_SigmaM_product} which implies the product proposal
structure \eqref{eq:proposal_structure} and thus allows for
particle-wise Metropolization, holds with
$\Phi(x^{(i)},\xx^{-(i)}) := C(\xx)\nabla\log\pi(x^{(i)})$ and
$\Sigma_\gamma(x^{(i)},\xx^{-(i)})$ as above.

The proposal density $q_{\xx^{-}}$ of $Q_{\xx^{-}}$ in \eqref{eq:proposal_structure}
is given by
\[
	q_{\xx^{-}}(z, y)
	=
	\frac1{\det(\sqrt{2\pi\, h\, \Sigma_\gamma(z,\xx^{-})})}
	\exp\left(-\frac{1}{4h}\|\Sigma_\gamma(z,\xx^{-})^{-1/2} \left(y- (z + h \Phi_\gamma(z,\xx^{-}))\right)\|^2\right)
\]
where $\Phi_\gamma(x^{(i)},\xx^{-(i)}) = (\gamma {\mathrm{Id}} + (1-\gamma)C(\xx))\nabla\log\pi( x^{(i)})+h(1-\gamma)\frac{d+1}{M}(x^{(i)}-m(\xx))$.

Variance inflation is not necessary, i.e., we can allow for
$\gamma=0$, if $C(\XX_k) \in \CM^+(\R^d)$ %
for all $k\in\N$.
The latter necessarily requires having $M\geq d+1$ particles $X_k^{(i)}$,
since ${\mathrm{rank}}(C(\XX_k)) \le M-1$.  Note that
$C(\XX_k) \in \CM^+(\R^d)$ implies that
$C(\XX_{k+1}) \in \CM^+(\R^d)$ almost surely for \eqref{eq:EM_ALDI}
with $\gamma = 0$ within \Cref{alg:MA-particle} or
\Cref{alg:MA-particle2}. This holds since the probability of proposing
ensembles which lie in a strict subspace of $\R^d$ is zero.  Thus, we
introduce:

\begin{assumption}\label{assum:MA-ALDI}
  At least one of the following conditions is met: (i) $\gamma > 0$
  or (ii)
  $M>d$ and for the initial ensemble $\xx_0\in\R^{Md}$ it holds that $C(\xx_0) \in \CM^+(\R^d)$.
\end{assumption}

\begin{remark}\label{rem:EKS}
    From an inverse problem perspective, the interacting Langevin system \eqref{eq:EKS} can also be viewed as modification of the ensemble Kalman inversion (EKI) \cite{ILS2013} and allows for derivative-free implementations avoiding the computation of $\nabla \log\pi$. Let $\pi$ be of the form
    \[ \pi(x) \propto \exp(-\frac{1}{2}\|\Gamma^{-1/2}(G(x)-z)\|^2),\]
    where $G:\R^d\to\R^{d_z}$ is a possibly nonlinear differentiable mapping, $z\in\R^{d_z}$, and $\Gamma\in\R^{d_z\times d_z}$ symmetric positive definite. 
    Then $\nabla \log\pi$ computes as
    \[\nabla\log\pi(x) \propto -\nabla G(x)\Gamma^{-1}(G(x)-z),\]
    for $x\in\R^d$. The key idea of the derivative-free implementation is to %
    apply the approximation
    \[C(\xx) \nabla G(x)\Gamma^{-1}(G(x)-z) \approx C^{xG}(\xx)\Gamma^{-1}(G(x)-z)\]
    within the particle system, where
    \[C^{xG}(\xx) := \frac{1}{M}\sum_{i=1}^M (G(\xx^{(i)})-m(G(\xx))(\xx^{(i)}-m(\xx))^\top,\quad m(G(\xx)):=\frac{1}{M}\sum_{i=1}^M G(x^{(i)}).\]
    Using a second order Taylor expansion, one can show that the approximation error scales with the spread of the particle system \cite[Lemma~4.5]{Weissmann_2022}. The derivative-free modification of \eqref{eq:EKS}
    \begin{equation}\label{eq:EKS_derivativefree}
    \dd X_t^{(i)} = -C^{xG}(\XX_t) \Gamma^{-1}(G(X_t^{(i)})-z)+\sqrt{2C(\XX_t)}\dd B_t^{(i)},\qquad i\in\{1,\dots,M\},
    \end{equation}
    is often referred to the Ensemble Kalman sampler (EKS)
    \cite{GHWS2020}. Note that for linear maps $G$ both \eqref{eq:EKS}
    and \eqref{eq:EKS_derivativefree} coincide. For nonlinear maps $G$
    and hence, non-Gaussian distributions $\pi$, localisation of the
    empirical covariance helps to improve the approximation of $\pi$
    through the EKS \cite{RW2021}.  Finally, one can similarly apply
    our proposed Metropolis-adjusted scheme to EKS in order to reduce
    the resulting bias through avoiding computation of derivatives.
\end{remark}

\begin{remark}\label{rem:MAEKS}
  Similarly, we can consider a variance-inflated version of the ALDI
  update following \Cref{rem:EKS}:
\begin{align}\label{eq:EKS_EM}
   X_{k+1}^{(i)} 
   =
   X_{k}^{(i)} 
   -h\, C_\gamma^{xG}(\XX_k) \Gamma^{-1}(G(X_k^{(i)}) - z)
   +
   \sqrt{ 2h\, (\gamma {\mathrm{Id}} + (1-\gamma)C(\XX_k)) } \zeta_{k+1}^{(i)} \qquad i\in\{1,\dots,M\},
\end{align}
with iid~$\zeta_{k+1}^{(i)}\sim\Nv(0,{\mathrm{Id}})$, $\gamma\in[0,1]$, $C(\xx)$ defined as above and 
\[
    C_\gamma^{xG}(\xx)
    =
    \gamma {\mathrm{Id}} + \frac {1-\gamma}{M-1}\sum_{i=1}^M (x^{(i)} - m(\xx)) (G(x^{(i)}) - m(G(\xx)))^\top,
    \qquad
    m(G(\xx)))
    = \frac 1M \sum_{i=1}^M 
    G(x^{(i)}).
\]
The resulting proposal distribution $q_\xx$ is as for ALDI above with the same $\Sigma_\gamma$ but different $\Phi_\gamma$.
\end{remark}

\paragraph{Summary}
Given \Cref{assum:MA-ALDI}, the interacting particle systems
\eqref{eq:EM_ALDI} and \eqref{eq:EKS_EM} posses positive proposal densities, thus establishing ergodicity for ensemble-wise Metropolization, see \Cref{cor:MAPS}. Furthermore, the
structural condition in \eqref{eq:PhiM_SigmaM_product} is met, which also ensures ergodicity for both particle- and block-wise Metropolization as stated in \Cref{theo:MAPS2_seq} and \Cref{theo:MAPS_block}.

\subsection{Metropolis-adjusted consensus based sampling (MA-CBS)} %
Motivated by the consensus based optimization (CBO) scheme
\cite{PTTM2017}, the authors of \cite{CHSV2022}
propose a modification leading to the so-called
consensus based sampling (CBS) method. While
CBO aims to find a global minimizer of some objective function
$\mathcal V:\R^d\to\R_+$, CBS aims to generate approximate samples
from measures of the form
$\pi(x)\propto \exp(-\mathcal V(x)), \ x\in\mathcal X= \R^d$.

The theoretical %
study of CBS in \cite{CHSV2022} %
was based on %
its continuous-time
formulation in the mean-field limit %
represented %
by
the McKean-Vlasov SDE %
\begin{equation}\label{eq:CB}
    \dd X_t = -(X_t-m_\pi(\pi_t))\,{\mathrm d}t + \sqrt{2\lambda^{-1}C_{\pi}(\pi_t)}\,{\mathrm d}B_t,  \quad X_0\sim \pi_0,
\end{equation}
i.e., $\phi(x,\rho) = -(x - m_\pi(\rho))$ and $\sigma(x,\rho) = 2\lambda^{-1} C_{\pi}(\rho)$ in \eqref{eq:SDE_infty}. Here $B_t$ is a $d$-dimensional Brownian motion, $\pi_t$ denotes the probability density function of the %
state $X_t$, $t\geq 0$, and $m_\pi(\rho)$, $C_\pi(\rho)$ denote the weighted mean and covariance of a probability distribution $\rho$ defined as
\begin{align*}
  m_\pi(\rho)&=\frac{1}{\int_{\R^d}\rho(x)\pi(x)\,{\mathrm d}x}\int_{\R^{d}} x \ %
                                           \pi(x)
                  \,\rho(x)\,{\mathrm d}x,\\
  C_{\pi}(\rho) &= \frac{1}{\int_{\R^d}\rho(x)\pi(x)\,{\mathrm d}x}\int_{\R^{d}}\Big((x-m_\pi(\rho))(x-m_\pi(\rho))^\top\Big)\,%
                       \pi(x)
                     \,\rho(x)\,{\mathrm d}x. 
\end{align*}

\paragraph{Ensemble update and proposal density}
Discretizing %
\eqref{eq:CB} by the Euler-Maruyama scheme in time and using an ensemble as empirical approximation for $\pi_t$ we obtain the following %
update
\[
   X_{k+1}^{(i)}
   = 
    X_k^{(i)} - h(X_k^{(i)} - m_{\pi}(\XX_k))+\sqrt{4hC_{\pi}(\XX_k)}\zeta_{k+1}^{(i)},
    \qquad
    i=1,\ldots,M, \ k \in\N,
\]
where $\zeta_{k+1}^{(i)}\sim\Nv(0,{\mathrm{Id}_d})$ iid and
\begin{align*}
    m_{\pi}(\xx)
    &=
      \frac{1}{\sum_{i=1}^M \pi(x^{(i)})}\sum_{i=1}^M \pi(x^{(i)})
      x^{(i)},\\
    C_{\pi}(\xx) &= \frac{1}{\sum_{i=1}^M \pi(x^{(i)})} \sum_{i=1}^M
                   \pi(x^{(i)})
                   \Big((x^{(i)}-m_{\pi}(\xx)) (x^{(i)}-m_{\pi}(\xx))^\top\Big)
\end{align*}
denote the weighted empirical mean and covariance of the ensemble
$\xx\in\R^{Md}$. We mention that in the original work \cite{CHSV2022}
  the authors introduced a rescaled time, so that our system
  slightly differs from the one in \cite{CHSV2022}.
The proposed Metropolization is applicable to both discrete time systems.

Similar to the ALDI method the weighted empirical covariance
$C_{\pi}(\XX_k)$ is not necessarily positive definite.  We can again
either use sufficiently many particles $M>d$ or introduce
variance inflation which yields the %
update
\begin{align}\label{eq:EM_CBS} 
    X_{k+1}^{(i)} 
    &
    = X_k^{(i)} - h(X_k^{(i)}-m_{\pi}(\XX_k)) + \sqrt{4h(\gamma{\mathrm{Id}} + (1-\gamma)C_{\pi}(\XX_k))}\zeta_{k+1}^{(i)},
    \qquad
    i=1,\ldots,M, \ k \in\N,
\end{align}
with $\gamma \in [0,1]$ fixed. Also \eqref{eq:EM_CBS} satisfies
\eqref{eq:PhiM_SigmaM_product}, and, hence, the corresponding ensemble
proposal kernel is of the form \eqref{eq:proposal_structure}
with proposal density
\begin{align*}
	q_{\xx^-}(z, y) 
	& =
	\frac1{\det(\sqrt{2\pi\, h\, \Sigma_\gamma(z,\xx^-)})}
	\exp\left(-\frac{1}{4h}\|\Sigma_\gamma(z,\xx^-)^{-1/2} \left(%
   y- (z + h \Phi_\gamma(z,\xx^-))\right)\|^2\right)
\end{align*}
where
$\Phi_\gamma(x^{(i)},\xx^{-(i)}) = x^{(i)} - h (x^{(i)} - m_{\pi}(\xx))$ and $\Sigma_\gamma(x^{(i)},\xx^{-(i)}) = 2(\gamma{\mathrm{Id}_d} + (1-\gamma)C_{\pi}(\xx))$.
\paragraph{Summary} %

Given \Cref{assum:MA-ALDI} the CBS method \eqref{eq:EM_CBS} %
possesses a positive proposal density, thus establishing ergodicity for ensemble-wise Metropolization, see \Cref{cor:MAPS}. Furthermore, the
structural condition in \eqref{eq:PhiM_SigmaM_product} is met, which also ensures ergodicity for both particle- and block-wise Metropolization as stated in \Cref{theo:MAPS2_seq} and \Cref{theo:MAPS_block}.

\subsection{Metropolis-adjusted Stein variational gradient descent (MA-SVGD)}
Stein variational gradient descent (SVGD) is a particle-based
  sampling method for Bayesian inverse problems \cite{LW2016}.
The algorithm can formally be viewed as a particle and time-discretization of the mean field equation
\begin{equation}\label{eq:SVGD_infty} 
    \dd X_{t} = \Phi(X_t, \pi_t) \dd t, \quad X_0 \sim \pi_0,
    \qquad
    \Phi(x, \rho)
    =
    \int_{\R^d} \left[ K(y, x)\nabla \log\pi(y) + \nabla_y K(y,x)\right] \, \pi_k(\dd y)
\end{equation}
where $\pi_t$ denotes the distribution of $X_t$ and $K$ is a fixed
kernel function \cite{NEURIPS2020_3202111c,SVGDgradientflow}.

The \emph{deterministic} interacting particle system
termed SVGD and originally proposed in \cite{LW2016} reads
\begin{align*}\label{eq:svgdup} 
  X_{k+1}^{(i)} & = X_k^{(i)}  + h_{k+1} \Phi(X_k^{(i)},\XX_k), &
  \Phi(x,\xx) &= \frac{1}{M} \sum_{j=1}^M K(x^{(j)},x) \nabla\log \pi(x^{(j)}) +\nabla_{x^{(j)}} K(x^{(j)},x)
\end{align*}
where $h_k>0$ denote the %
step size in the $k$th step.
While this %
method %
has shown promising results in applications, there is still a lack of
theoretical understanding. Specifically, convergence results %
have only been established
for the long time behavior of the mean field limit in
continuous time, %
where
exponential convergence of the
Kullback-Leibler divergence was shown under strong assumptions on
the underlying kernel $K$ \cite{NEURIPS2020_3202111c,Duncan2019OnTG}.

\paragraph{Ensemble update and proposal density}
Motivated by %
  the Langevin dynamics, a \emph{stochastic} version of SVGD has recently
  been proposed in \cite{Gallego2018StochasticGM}. This approach is driven by
a $\ppi$-invariant coupled system of SDEs \cite{Nsken2021SteinVG,Gallego2018StochasticGM} that describes the time-continuous dynamics of the ensemble $\XX_t$. 
Notably, the mean field limit for $M\to\infty$ of this system of SDEs is %
the same as in
\eqref{eq:SVGD_infty}.

The stochastic particle dynamics reads
\begin{equation}\label{eq:EM_SVGD}
	\XX_{k+1} = \XX_{k} + h\PPhi(\XX_k)+ \sqrt{h \SSigma(\XX_k) } \zeta_{k+1},
\end{equation}
where $\zeta_{k}\sim \Nv(0, \mathrm{Id}_{Md})$,
\begin{align*}
    \PPhi(\xx) &:= (\Phi(x^{(1)},\xx^{(-1)}),\dots,\Phi(x^{(M)},\xx^{(-M)}))^\top,\\
    \Phi(x^{(j)},\xx^{-(j)}) &:=  
    \frac{1}{M} \sum_{i=1}^M K(x^{(j)}, x^{(i)}) \nabla\log \pi(x^{(i)}) +\nabla_{x^{(i)}} K(x^{(j)}, x^{(i)}),
\end{align*}
and 
\[
	\SSigma(\mathbf x) = \frac 2M \, \XS^\top  \diag_d(\mathcal K(\XX),\dots,\mathcal K(\XX)) \XS
      \]
where $\mathcal K(\xx) \in \CM^+(\R^d)$ is given by $\mathcal K(\xx)_{i,j} = K(x^{(i},x^{(j)})$, $i,j=1,\dots,M$, and $\XS \in \R^{Md\times Md}$ denotes a permutation matrix defined by
\[
  \XS = \begin{pmatrix}
    \mathbf{1}_{1,1} %
    & \dots & \mathbf{1}_{M,1}\\
    \vdots & \ddots & \vdots %
    \\
    \mathbf{1}_{1,d} %
    & \dots & \mathbf{1}_{M,d} 
    \end{pmatrix}, 
    \qquad
    \mathbf{1}_{k,\ell}\in\R^{M\times d} \ \text{with}\ (\mathbf{1}_{k,\ell})_{n,m} = \begin{cases} 1, & k=n, \ell = m\\ 0, & \text{else}\end{cases}. 
\] 
Note that $\XS$ is an orthogonal matrix such that $\XS^\top \XS = \XS \XS^\top = {\mathrm{Id}}_{Md}$, and it holds
$$
\sqrt{h\SSigma(\xx)} = \sqrt{2h/M} \XS^\top \diag_d(\sqrt{\mathcal K(\xx)},\dots,\sqrt{\mathcal K(\xx)}) \XS.$$
Moreover, with $\mathbf{K}_{ij} := K(x^{(i)},x^{(j)})\, \mathrm{Id}_{d}$,  we can write 
\[
    \SSigma(\xx) = 
    \frac 2M
    \begin{pmatrix}
    \mathbf{K}_{11} %
    & \cdots & \mathbf{K}_{1M}\\
    \vdots & \ddots & \vdots %
    \\
    \mathbf{K}_{M1} %
    & \ldots & \mathbf{K}_{MM}
    \end{pmatrix}.
\]
Thus, \eqref{eq:EM_SVGD} does not satisfy the structural assumption
\eqref{eq:PhiM_SigmaM_product} for particle-wise Metropolization as in
\Cref{alg:MA-particle2}.  However, the associated proposal kernel
$\QQ$ in $\R^{Md}$ possesses the positive Lebesgue density
\[
    \qq(\xx,\yy)
    =
    \frac1{\det(\sqrt{2\pi\, h\, \SSigma(\xx)})}
	\exp\left(-\frac{1}{2h}\|\SSigma(\xx)^{-1/2} \left(\yy - (\xx + h \PPhi(\xx))\right)\|^2\right)
    > 0,
\]
since $\mathcal K(\xx)$ and, therefore, $\SSigma(\xx)$ is positive definite for distinct particles $x^{(i)} \neq x^{(j)}$, $i\neq j$.

\begin{remark}
  Similar
  to
  CBS and ALDI, we can improve the stability %
  of the %
  system by introducing variance inflation in order
  to avoid the kernel matrix $\mathcal K(\xx)$ from becoming close to singular
  in case the particles collapse.
  To do so, one can replace $\SSigma(\xx)$ by
  $\SSigma_\gamma(\xx) := \gamma {\mathrm{Id}_{Md}} + (1-\gamma)
  \SSigma(\xx)$, $\gamma\in(0,1)$.  However, in our numerical examples
  such a variance inflation was not necessary and, hence, not
  applied.
\end{remark}

\paragraph{Summary}
The stochastic SVGD interacting particle system \eqref{eq:EM_SVGD}
does \emph{not} %
fit into the setting of particle- or block-wise Metropolization as introduced
in %
 \Cref{sec:particle_met} and \ref{sec:block_met}, since the
  structural condition \eqref{eq:PhiM_SigmaM_product} is not
  satisfied. However, it 
  posseses a positive proposal density so that ergodicity holds again for ensemble-wise Metropolization, see \Cref{cor:MAPS}.

\section{Numerical experiments} \label{sec:numerics} To evaluate
  the performance of our proposed Metropolis-adjusted interactive
  particle sampling (MA-IPS) methods, we conduct three numerical
  experiments in which we implement (MA-)ALDI, (MA-)SVGD, and
  (MA-)CBS. The implementation was done in \texttt{MATLAB} and
  we utilized
  the discrete dynamics \eqref{eq:EM_ALDI},
  \eqref{eq:EM_CBS}, and \eqref{eq:EM_SVGD} for the unadjusted
  algorithms, along with the corresponding Metropolis-adjusted methods
  based on Algorithms~\ref{alg:MA-particle} to
  \ref{alg:MA-particle_block}. Let us briefly describe the three
  experiments below:

\begin{enumerate}
\item 
The first experiment illustrates the bias of the unadjusted interactive particle samplers for a one-dimensional non-Gaussian target distribution, thus, emphasizing the need for Metropolization.

\item The second experiment is a more detailed study of MA-IPS using a
  four-dimensional multivariate Gaussian target distribution. It
  demonstrates the benefits of interaction (compared to independent
  parallel Markov chains) and 
  compares the performance of all three Metropolization versions, in particular, under consideration of parallelization.
  Moreover, also the optimal tuning of the considered MA-IPS methods is studied empirically by the relation of the average acceptance rate and the obtained mean squared error of $\SM_N(F)$ for a chosen $F$.

\item 
  The final experiment applies the MA-IPS methods to a high-dimensional target distribution and confirms our observations in the low-dimensional setting.
\end{enumerate} 

We adhere to the following naming conventions for our algorithms:
  
\begin{itemize}
\item {\bf ALDI} corresponds to the chain generated by \eqref{eq:EM_ALDI},
  {\bf CBS} corresponds to the chain generated by \eqref{eq:EM_CBS},
  and {\bf SVGD} corresponds to the chain generated by \eqref{eq:EM_SVGD},
\item {\bf pMALA} corresponds to parallel MALA, i.e.\ to $M$ independent chains each generated
  by \Cref{alg:MHMCMC} with proposal \eqref{eq:pMALA},
\item we use the prefix {\bf MA}, to indicate that it's the Metropolized
  version of an algorithm,
\item we use the suffix {\bf ew} to indicate ensemble-wise Metropolization
  as in \Cref{alg:MA-particle}, e.g.\ MA-ALDI-ew corresponds to \Cref{alg:MA-particle} with proposal \eqref{eq:EM_ALDI},
\item we use the suffix {\bf pw} to indicate particle-wise sequential Metropolization as
  in \Cref{alg:MA-particle_seq}, e.g.\ MA-ALDI-pw corresponds to \Cref{alg:MA-particle_seq} with proposal \eqref{eq:EM_ALDI},
\item we use the suffix {\bf bw} to indicate block-wise
   Metropolization as in \Cref{alg:MA-particle_block}, e.g.\
  MA-ALDI-bw corresponds to \Cref{alg:MA-particle_block} with proposal
  \eqref{eq:EM_ALDI}.
\end{itemize}

\subsection{One-dimensional example: bimodal distribution}
We study the effect of Metropolization on interacting particle systems for ALDI, CBS, and SVGD, by comparing the performance of each method with and without ensemble Metropolization.

As a target we consider the (unnormalized) %
probability density %
\[ 
    \pi(x) \propto \exp\left(-\frac1{2\sigma^2}|x^2-1|^2\right)\cdot \exp\left(-\frac12|x-m_0|^2\right) \qquad x\in\R,
\]
with $\sigma = 0.1$ and $m_0=0.8$. This target density corresponds to a BIP with Gaussian prior distribution $\Nv(m_0,1)$, forward model $G(x) = x^2$ and additive Gaussian noise with zero mean and variance $\sigma^2$.

For each particle sampling method we have simulated $N_{\mathrm{burn}} = 10^4$ iterations of burn-in followed by another $N=10^5$ iterations. 
Both, the Metropolis-adjusted and unadjusted variant, were implemented with $M=10$ particles, same realization of the iid initialization $X_0^{(j)} \sim \Nv(m_0,1)$, $j=1,\dots,M$ and same step size $h>0$ in the Euler-Maruyama scheme. 
For %
both
ALDI versions we %
used step size $h=0.0725$ and for both
CBS versions %
the step size $h=0.05$. For %
both SVGD versions we used step size $h = 0.0001$ and the Gaussian kernel \[k_s(x_1,x_2) = \exp\left(-\frac1{2s^2} |x_1-x_2|^2\right)\]
with $s = 0.01$.  

In Figure~\ref{fig:1d_SVGD}--\ref{fig:1d_CBS} we compare
the %
histograms of the particles obtained by the Metropolis-adjusted and unadjusted algorithms accumulated over all $N$ iterations. %
For all three interactive particle methods, their unadjusted version
do not seem to yield %
a sample from
the target
$\pi$ whereas their Metropolized versions %
do.

\begin{figure}[!htb]
\centering \includegraphics[width=0.4\textwidth]{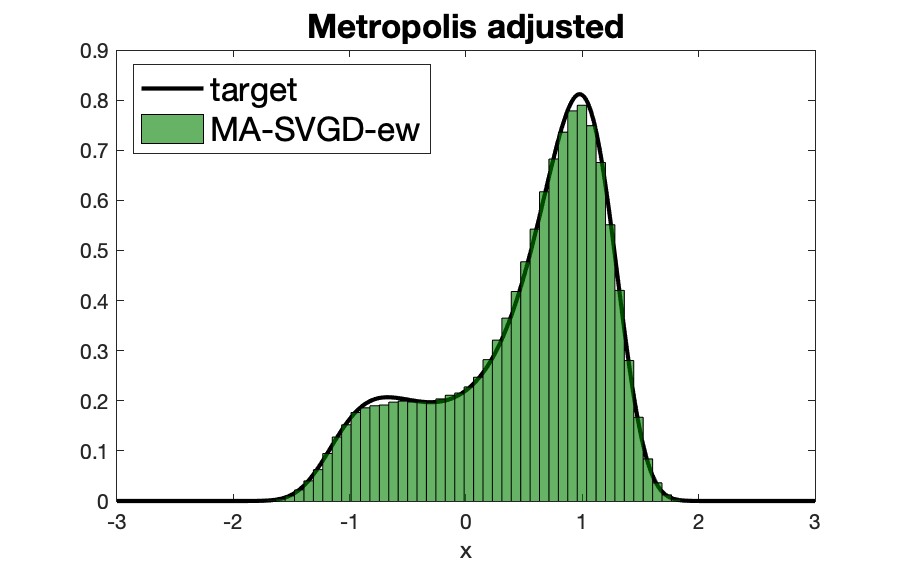}~~\includegraphics[width=0.4\textwidth]{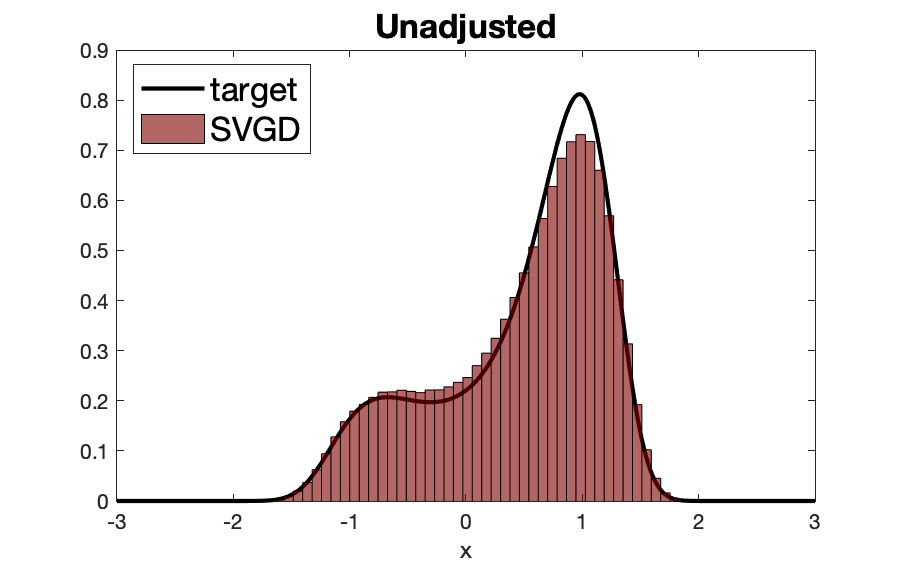}
 \caption{Comparison of MA-SVGD-ew and SVGD with $M=10$ particles and step size $h=0.0001$. MA-SVGD-ew achieved an acceptance rate of $53\%$.}\label{fig:1d_SVGD}
\end{figure} 

\begin{figure}[!htb]
\centering \includegraphics[width=0.4\textwidth]{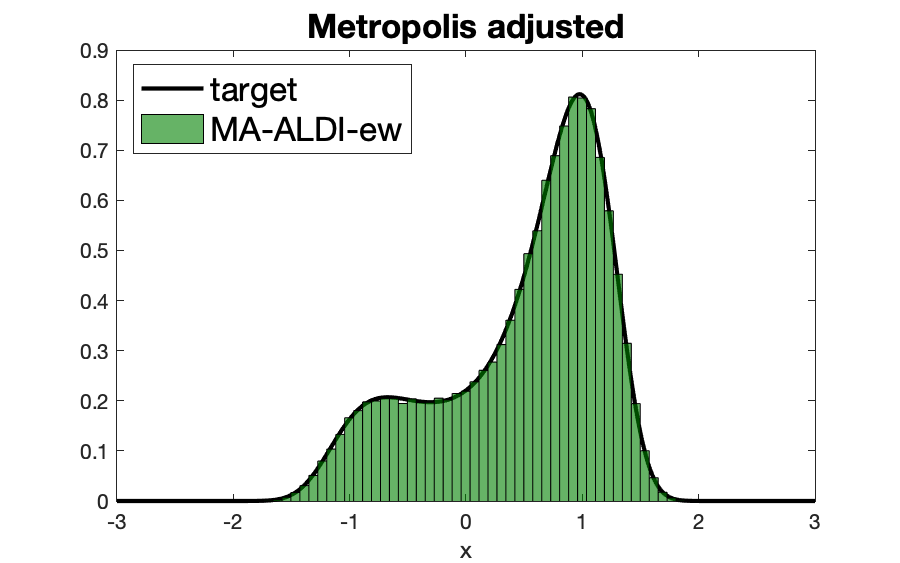}~~\includegraphics[width=0.4\textwidth]{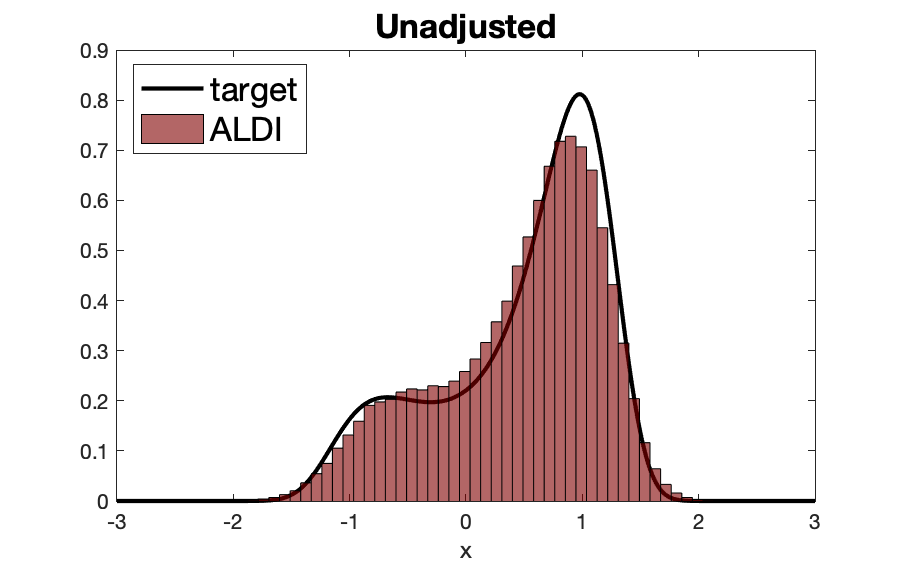}
 \caption{Comparison of MA-ALDI-ew and ALDI with $M=10$ particles and step size $h=0.0725$. MA-ALDI-ew achieved an acceptance rate of $70\%$. }\label{fig:1d_EKS}
\end{figure} 

\begin{figure}[!htb]
\centering \includegraphics[width=0.4\textwidth]{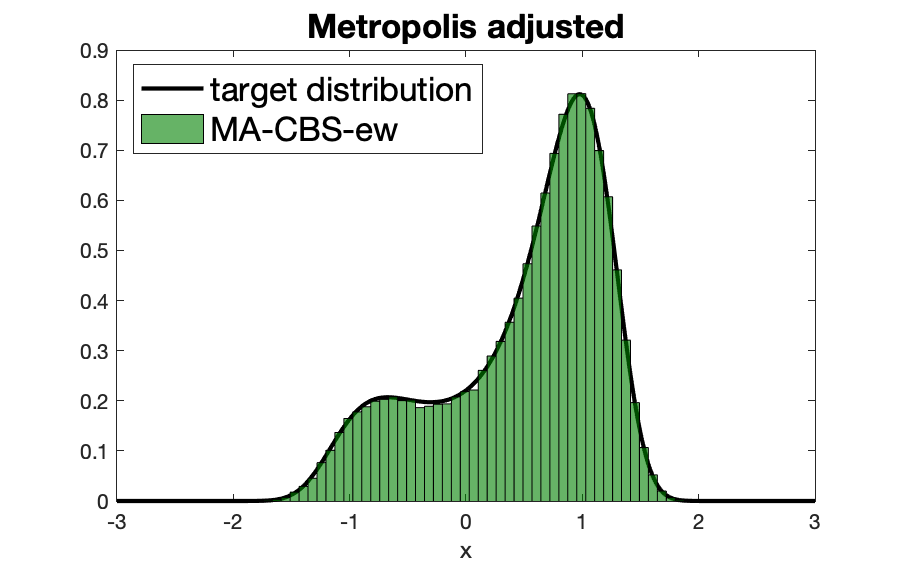}~~ \includegraphics[width=0.4\textwidth]{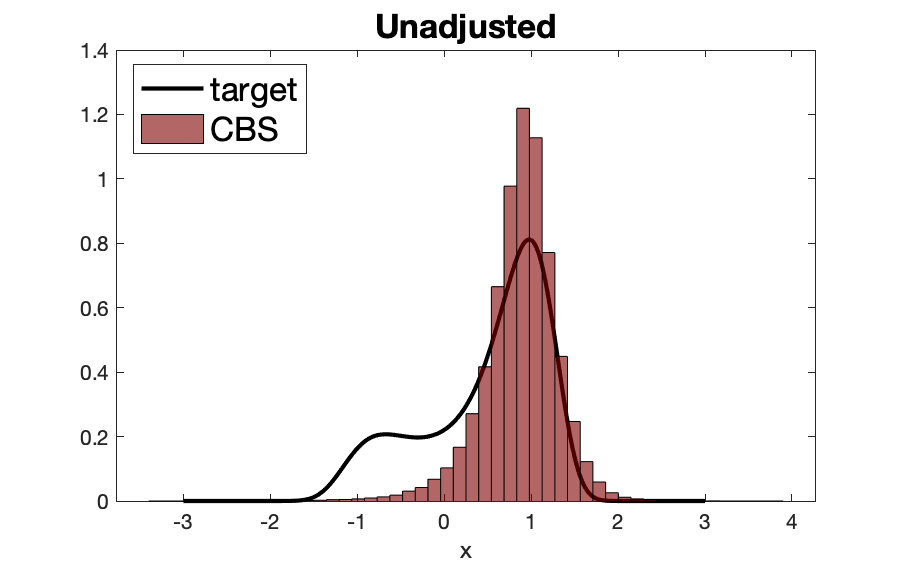}
 \caption{Comparison of MA-CBS-ew and CBS with $M=10$ particles and step size $h=0.05$. MA-CBS-ew achieved an acceptance rate of $52\%$.}\label{fig:1d_CBS}
\end{figure} 

\subsection{Multivariate Gaussian distribution}\label{sec:mvgauss}
We now consider a $4$-dimensional multivariate Gaussian target distribution $\Nv(0,C)$ with target density
\begin{equation*}
    \pi(x) \propto \exp\left(-\frac12\|C^{-1/2}x\|^2\right),
    \qquad
    C = \diag(1,0.1,0.01,0.001).
  \end{equation*}
  As the coordinates are weighted differently, we
  expect a positive effect from the ensemble preconditioner.
In this section we illustrate the benefits of interaction, compare ensemble-wise to particle-wise Metropolization (for MA-ALDI), and consider optimal tuning of MA-ALDI, MA-CBS, and MA-SVGD by controlling the acceptance rate.

To this end, we consider the quantity of interest 
$f(x) = x^\top C^{-1} x$ where for $X\sim\pi$ we have $f(X)\sim \chi^2(4)$.
Our goal is then to estimate
\begin{equation}\label{eq:Pref}
    P_{\mathrm{ref}} := \Prob(f(X) \le q_{0.5} ) = \E_\pi\left[ \mathds{1}_{(-\infty, q_{0.5}]}(f(X))\right] = \frac 12
\end{equation}
where $q_{0.5}$ denotes the $0.5$-quantile of the $\chi^2(4)$ distribution and $\mathds{1}_A$ the indicator function of a set $A$.
The corresponding estimators based on ensemble Markov chains $(\XX_k)_{k\in \N}$ are then
\[
    \widehat P_N
    :=
    \frac 1N \sum_{k=1}^N \FF(\XX_k),
    \qquad
    \FF(\xx) := \frac 1M \sum_{i=1}^M \mathds{1}_{(-\infty, q_{0.5}]}(f(x^{(i)})) \in \R.
\]
For evaluating the efficicency of these estimators we also compute %
the associated autocorrelations $\Corr(\FF(\XX_1), \FF(\XX_{1+k})) \approx \rho_k := c_k/c_0$ where
\[
    c_k := 
    \frac{1}{N-k}
    \sum_{i=1}^{N-k} 
    \left(\FF(\XX_i) - \widehat P_N\right)
    \left(\FF(\XX_{i+k}) - \widehat P_N\right)
\]
since the $\rho_k$ yield information about the asymptotic variance \eqref{eq:sigmaF} in the central limit theorem.

\paragraph{Superiority of interacting particles}
In Figure \ref{fig:quantile_est_gamma}, we compare MA-ALDI-ew %
for different values of the covariance inflation parameter
$\gamma\in\{0.001,0.1,1\}$ to pMALA, in order to show the
improved performance for proposals generated by \emph{interacting}
  particle systems.  Note that for $\gamma=1$ the proposal
of MA-ALDI-ew and %
pMALA are in fact the same,
but MA-ALDI-ew uses
ensemble-wise Metropolization %
while pMALA uses
particle-wise Metropolization.

Figure~\ref{fig:quantile_est_gammaa} %
shows the evolution of $\widehat P_N$ averaged over $100$
independent runs for each method.
For each method, we
tuned the step size $h$ such that the average acceptance rate was
around $50\%$.
For %
$\gamma =0.001$ and $\gamma = 0.1$ the MA-ALDI-ew based %
estimators
converge much faster than pMALA.
This observation is confirmed by the significantly
faster decaying autocorrelation depicted on the right plot in
Figure~\ref{fig:quantile_est_gamma}.  On the other hand %
for
$\gamma = 1$ it is observed that %
MA-ALDI-ew performs worse %
than pMALA, suggesting that for $\gamma=1$ the particle-wise Metropolization is more effective than ensemble-wise Metropolization.

\begin{figure}[!htb]
  \begin{subfigure}{0.48\textwidth}
    \centering
    \includegraphics[width=\textwidth]
    {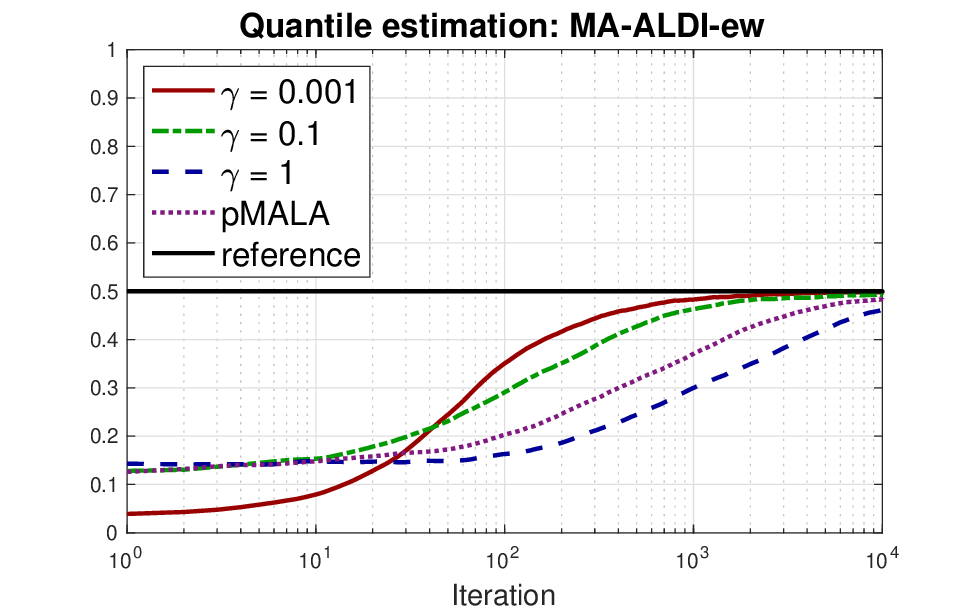}
    \caption{quantile estimation}\label{fig:quantile_est_gammaa}
  \end{subfigure}
  \hfill
  \begin{subfigure}{0.48\textwidth}
    \centering    
    \includegraphics[width=\textwidth]
    {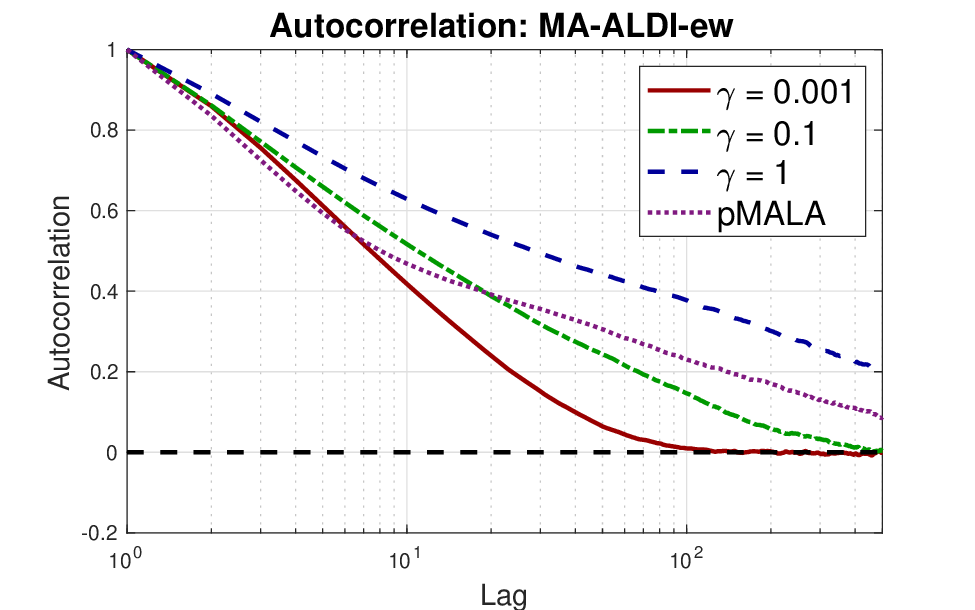}
    \caption{autocorrelation}\label{fig:quantile_est_gammab}
    \end{subfigure}
    \caption{Estimation of $P_{\mathrm{ref}}$ in \eqref{eq:Pref} and the corresponding estimated autocorrelation using either pMALA %
      or %
      MA-ALDI-ew %
      for different choices of $\gamma$. For each %
      method
      the step size $h$ %
      was tuned to obtain an
   acceptance rate of approximately $50\%$. }\label{fig:quantile_est_gamma}
\end{figure}

\paragraph{Superiority of particle- and block-wise Metropolization}
In \Cref{fig:quantile_est_parallelization1} and \Cref{fig:quantile_est_parallelization2}
  we compare the ensemble-wise Metropolization MA-ALDI-ew
  with the %
  particle-
and block-wise Metropolizations MA-ALDI-pw
and MA-ALDI-bw for $\gamma = 0.001$.
The results are obtained by averaging
over $10$ independent runs of the corresponding methods. We used the same
step size $h$ and ensemble size $M$ for all three MA-ALDI versions. Note that we have increased the ensemble size to $M=100$ in this experiment.
Similar as for $\gamma=1$, the particle-wise Metropolization
outperforms the ensemble-wise Metropolization also for
$\gamma = 0.001$ and $\gamma=0.1$.  We see hardly any difference
between %
particle- and block-wise
  Metropolization in \Cref{fig:quantile_est_parallelization1} (a),
where we do not %
take into account possibilities for parallel computing. %
The result %
changes, when allowing access to multiple cores, such that the
computation of the ensemble- and block-wise Metropolzation can be done in parallel. 
In \Cref{fig:quantile_est_parallelization1} (b) and \Cref{fig:quantile_est_parallelization2}, we observe a clear advantage of block-wise and even ensemble-wise Metropolization the
more cores we incorporate. In addition, %
Table~\ref{tab:inte_ac} shows the estimated integrated autocorrelation
scaled by the associated cost of the applied algorithm, i.e.\ the quantity
  \begin{equation}\label{eq:efficiency}
  {c({\mathrm{int\_ac}}, B,{\mathrm{cores}})} = {\mathrm{int}}\_{\mathrm{ac}}\cdot N \cdot M \cdot \frac{1}{\min\{\mathrm{cores},B\}}.
\end{equation}
Here we assume all blocks in \eqref{eq:blocks} to be of equal size
  $|\bb_j|=B \in\N$ for all $j=1,\dots,L$ (and thus $M=LB$
  is a multiple of $B$). Moreover, $\mathrm{cores}\in\N$ refers to the number of processors available for parallelization, $N$ is the
  number of iterations, and $\mathrm{int\_ac}$ refers to the estimated integrated autocorrelation.
  Note that ensemble-wise Metropolization and parallel MALA correspond to $B=M$ and particle-wise Metropolization corresponds to $B=1$. 
  
  \begin{table}[t]
  \centering
  \begin{tabular}{c|ccccc}
    cores & ew & bw, B=50 & bw, B=25 & pw & pMALA \\ \hline
    $1$ & $1366$ & $840.6$ & $708.1$ & $593.2$ & $8996$ \\
    $20$ & $68.3$ & $42$ & $35.4$ & $593.2$ & $449.8$  \\
    $50$ & $27.3$ &$16.8$ & $28.3$ & $593.2$ & $179.9$\\
    $100$ & $13.7$ & $16.8$ & $28.3$ & $593.2$ & $90$
  \end{tabular}
  \caption{Comparison of the quantity in \eqref{eq:efficiency}
      measuring the computational efficiency of the method for different numbers of
      cores available and different variants of Metropolization.}
  \label{tab:inte_ac}
\end{table}

\begin{figure}[!htb]
  \begin{subfigure}{0.48\textwidth}
    \centering \includegraphics[width=\textwidth]{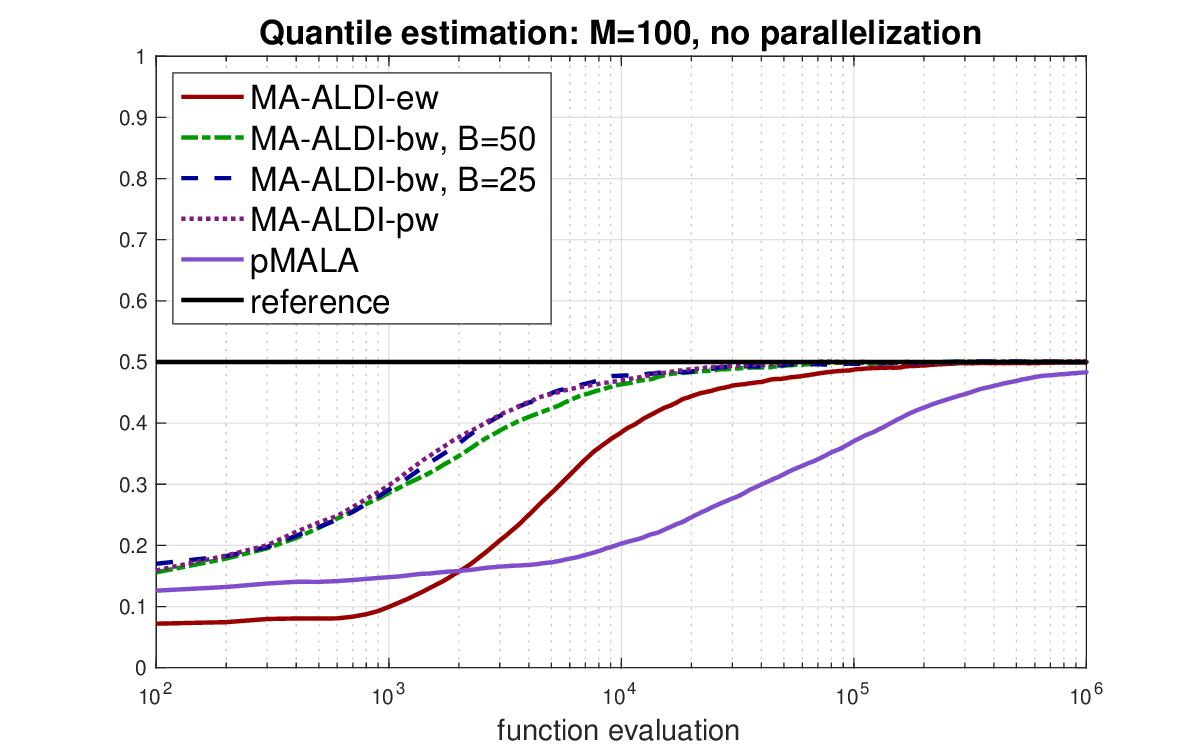}
    \caption{no parallelization}
    \label{fig:quantile_est_pw1a}
  \end{subfigure}%
  \hfill
  \begin{subfigure}{0.48\textwidth}
    \includegraphics[width=\textwidth]{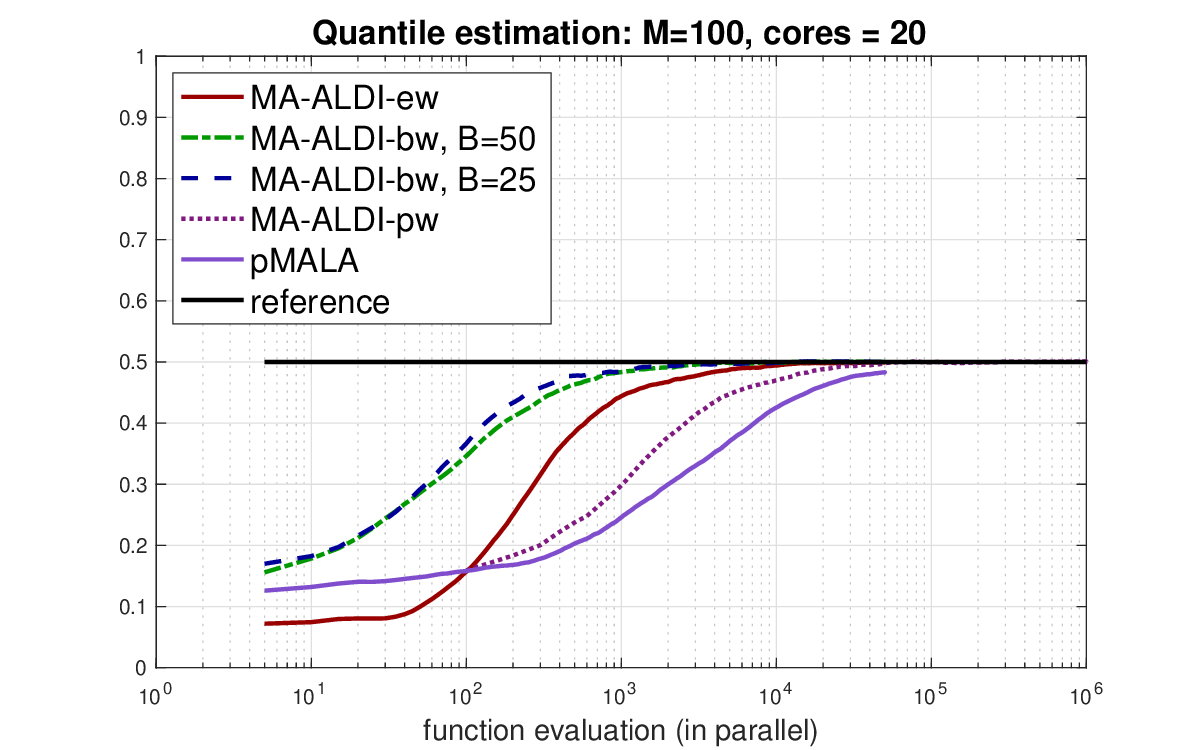}
    \caption{20 cores}
    \label{fig:quantile_est_pw1b}
  \end{subfigure}
  \caption{Estimation of $P_{\mathrm{ref}}$ in \eqref{eq:Pref} %
      for different
      versions of MA-ALDI with $\gamma=0.001$, and pMALA.  
      For MA-ALDI-ew
      the step size $h$ was tuned to obtain an acceptance
      rate of approximately $50\%$. For MA-ALDI-pw and both
      MA-ALDI-bw we use the same $h$ as for MA-ALDI-ew.
  }\label{fig:quantile_est_parallelization1}
\end{figure}

\begin{figure}[!htb]
  \begin{subfigure}{0.48\textwidth}
    \centering \includegraphics[width=\textwidth]{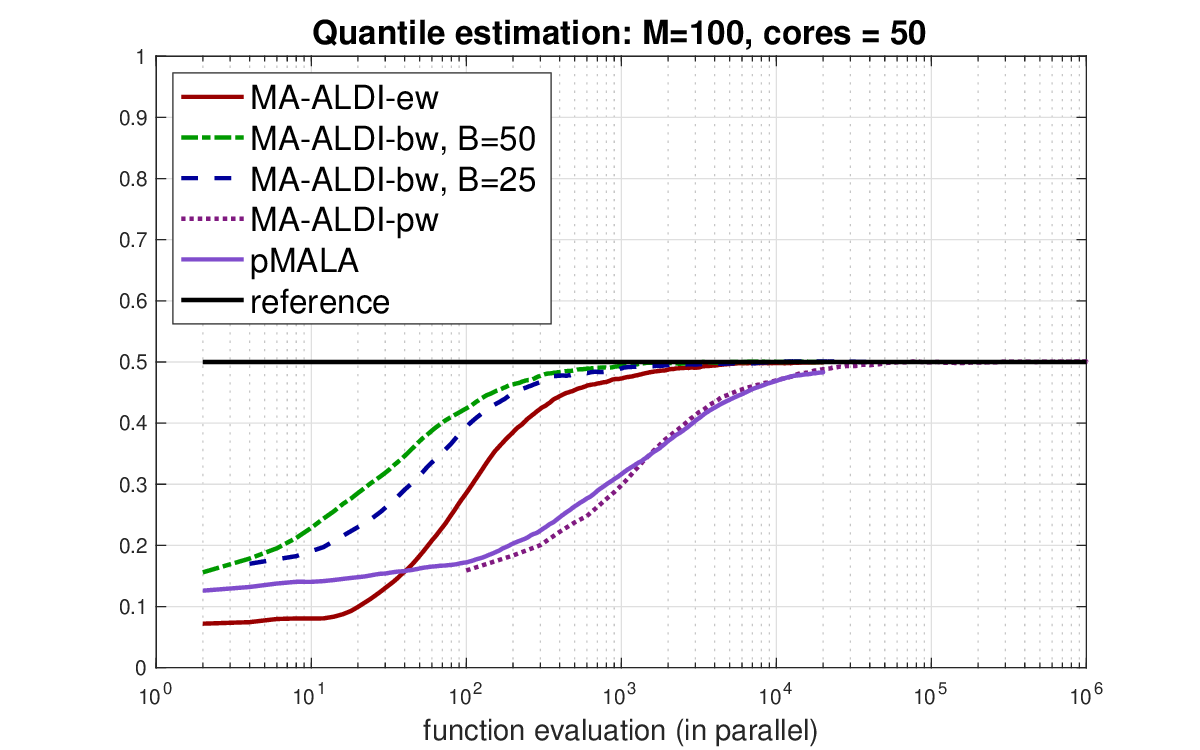}
    \caption{50 cores}
    \label{fig:quantile_est_pw1a}
  \end{subfigure}%
  \hfill
  \begin{subfigure}{0.48\textwidth}
    \includegraphics[width=\textwidth]{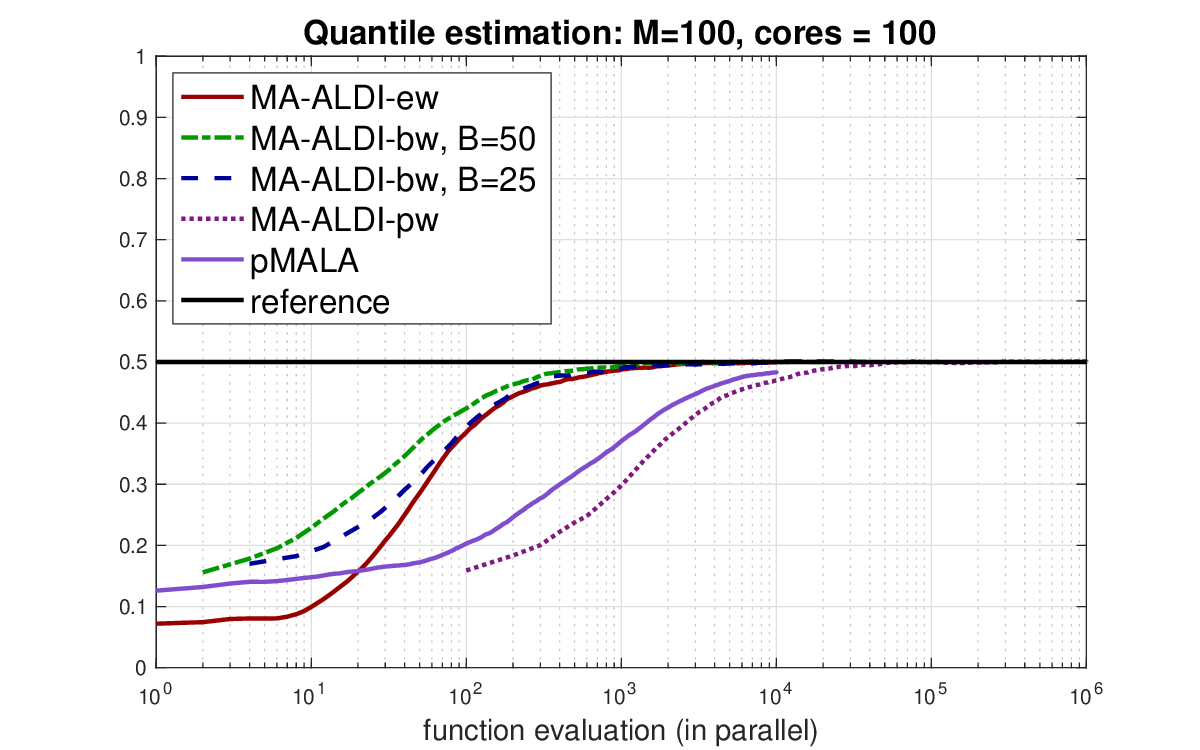}
    \caption{100 cores}
    \label{fig:quantile_est_pw1b}
  \end{subfigure}
  \caption{Same as Figure~\ref{fig:quantile_est_parallelization1} but with different assumptions on accessible cores.
  }\label{fig:quantile_est_parallelization2}
\end{figure}

\paragraph{Optimal tuning}
Finally, We study the dependence of the mean squared error (MSE) of $\SM_N(F)$ on the average acceptance rate for the considered ensemble-wise MA-IPS in Figure~\ref{fig:acc_rate_MHEKS-MHCBS}.
Here, we control the acceptance rate through the step size $h$ of the Euler-Maruyama scheme. 
The MSE was estimated over $100$ independent runs for each method and step size. 
We observe slightly different optimal average acceptance rates for MA-SVGD-ew, MA-ALDI-ew ($\gamma = 0.001$), and MA-CBS-ew ($\gamma=0$), %
with MA-ALDI-ew notably performing most sensitive w.r.t.~the acceptance rate, but also achieving the smallest MSE---by an order of magnitude smaller compared to MA-SVGD-ew and two orders of magnitude compared to MA-CBS-ew.
For MA-SVGD-ew we can additionally control the acceptance rate through the kernel function: We use the Gaussian kernel
\[
    k_s(x_1,x_2) \propto \exp\left(-\frac{1}{2s^2} \|x_1-x_2\|^2\right), 
\]
and can also steer the performance of (MA-)SVGD by the variance
parameter $s^2$. The results are shown in \Cref{fig:accrate_MHSVGD}.
The optimal average acceptance rate seems to almost the same as for
tuning $h$ (right). The dependence of the MSE and average acceptance
rate is explicitly shown in the middle and right plot of
\Cref{fig:accrate_MHSVGD}.

\begin{figure}[!htb]
\centering  \includegraphics[width=0.33\textwidth]{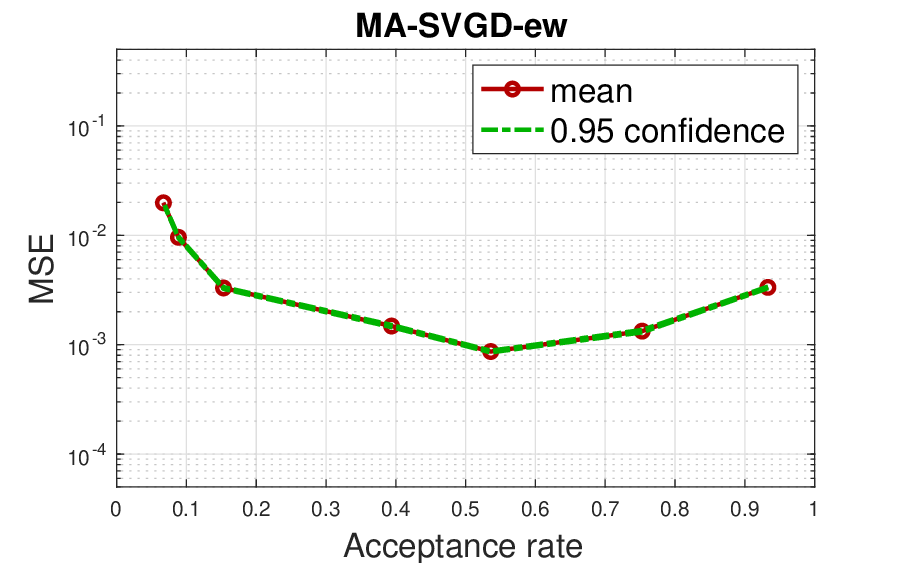}~~\includegraphics[width=0.33\textwidth]{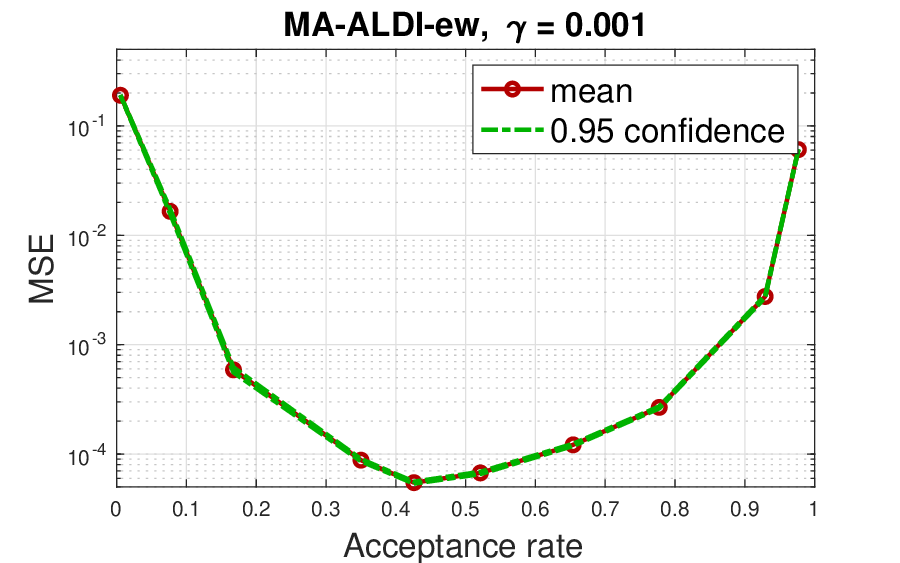}~~\includegraphics[width=0.33\textwidth]{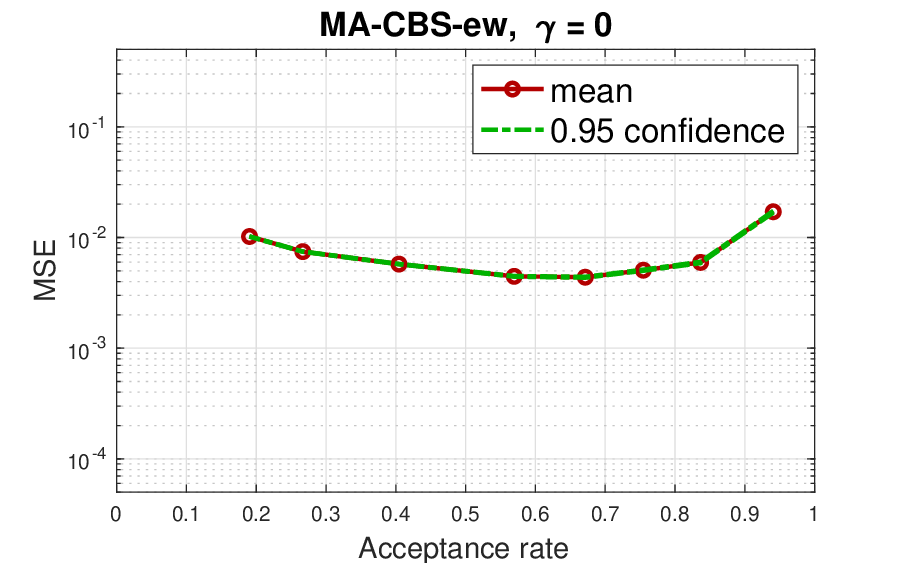}
\caption{Comparison of the expected MSE depending on the expected acceptance rate for MA-SVGD-ew (left),  %
  MA-ALDI-ew with $\gamma=0.001$
   (middle) and %
MA-CBS-ew with $\gamma=0$   (right). } \label{fig:acc_rate_MHEKS-MHCBS}
\end{figure}

\begin{figure}[!htb]
\centering \includegraphics[width=0.33\textwidth]{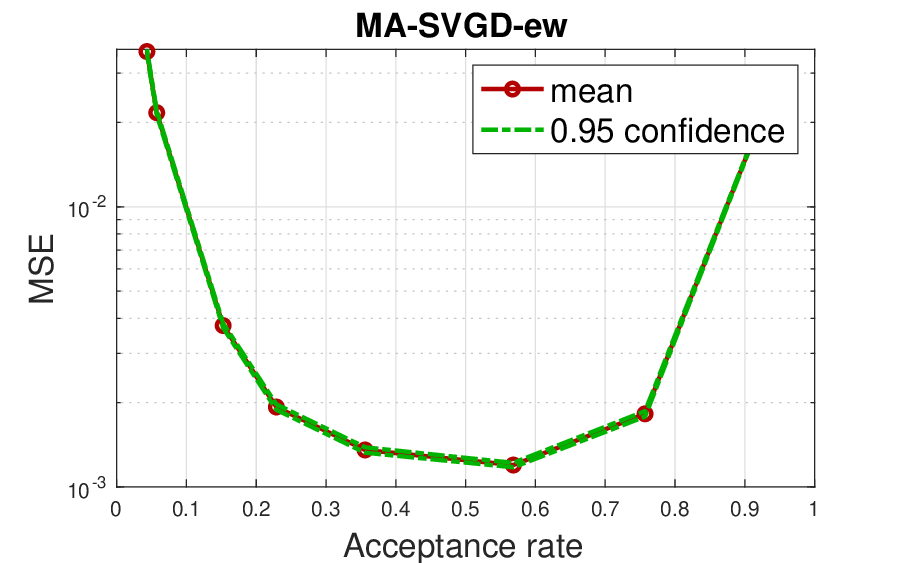}~~ \includegraphics[width=0.33\textwidth]{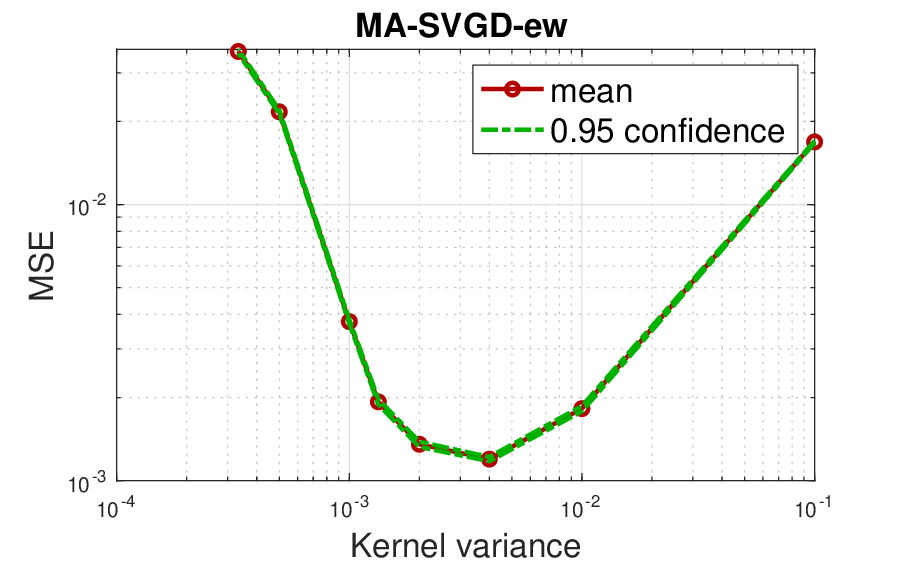}~~ \includegraphics[width=0.33\textwidth]{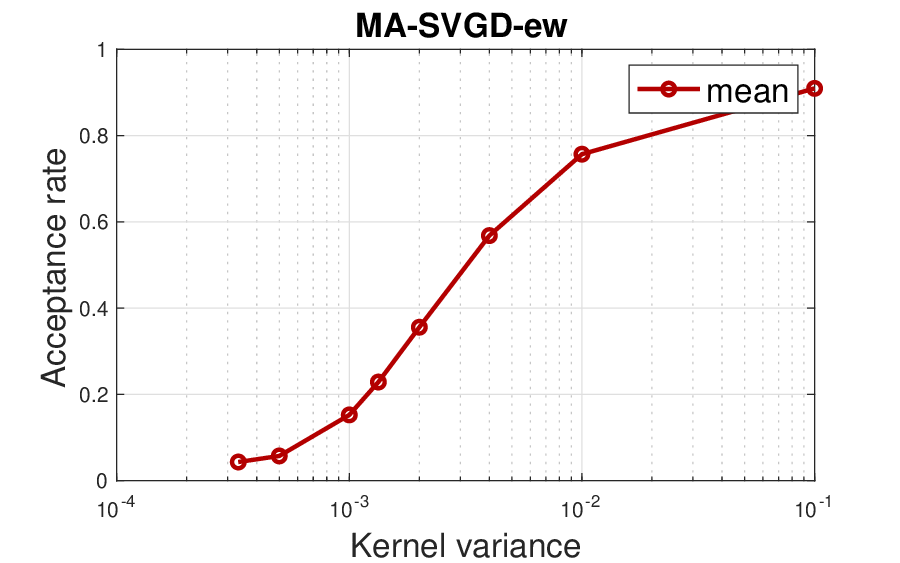}
 \caption{Comparison of the expected MSE depending on the expected acceptance rate and the choice of kernel variance for MA-SVGD-ew.  %
 } \label{fig:accrate_MHSVGD}
\end{figure} 

\subsection{ODE-based linear inverse problem }

We consider the one-dimensional elliptic equation 

\begin{equation}\label{eq:IP_ode}
\begin{cases}
-\frac{{\mathrm d}^2p(s)}{{\mathrm d}^2s} + p(s) = \theta(s) & s\in D:=(0,1)\\
p(s) = 0 & s\in\{0,1\}
\end{cases}
\end{equation}
and the inverse problem of recovering the unknown right-hand side $\theta\in L^\infty(D)$ from noisy observations $y = A(\theta) + \xi\in\R^K$, where $\xi\in\R^K$ denotes observational noise. The forward operator $A:%
H^{-1}(D)\to \R^K$ is defined by
\[
    A = \mathcal O \circ G^{-1}, \quad G = -\frac{{\mathrm d}^2}{{\mathrm d}^2s}+{\mathrm{id}}\ \text{on}\ \mathcal D(G) = H_0^1(D), 
\]
where $\mathcal O: H_0^1(D)\to\R^K$ denotes the observation operator providing function values of $K$ equidistant observation points $s_k = \frac{k}{K}$, $k=1,\dots,K$, such that
for $p\in H_0^1(D)$ we have $\mathcal O p(\cdot) = (p(s_1),\dots,p(s_K))^\top$. 

We consider a Gaussian process prior for $\theta$ given by
\[
    \theta(\cdot,x) 
    = 
    B_{\phi} x
    :=
    \sum_{i=1}^{d} x_i \phi_i(\cdot), 
\]
where $\phi_i(s) = \frac{\sqrt{2}}{\pi}\sin(i \pi s)$ and $x_i \sim \Nv(0, \lambda_i)$ independently with $\lambda_i =i^{-2\tau}$ for some fixed $\tau>1$.
Thus, the resulting inverse problem is to recover the coefficients $x=(x_1,\dots,x_d)^\top\in\R^d$ with prior information $\Nv(0,\Gamma_0)$, where $\Gamma_0 = \diag(\lambda_1,\dots,\lambda_d)
$. 
Assuming additive Gaussian noise $\xi\sim\Nv(0,\Gamma)$ the resulting %
(unnormalized) posterior density is
\[
    \pi(x) \propto \exp\left(-\frac12 \|\Gamma^{-1/2}(y-AB_\phi x)\|^2-\frac12\|\Gamma_0^{-1/2}x\|^2\right). 
\]
For the numerical implementation we replace $G$ by a numerical solution operator for \eqref{eq:IP_ode} on the grid $D_\delta \subset D$ with mesh size $\delta=2^{-6}$ and restriction of the unknown parameter $\theta(\cdot, x)$ to $D_\delta$. 
We consider a fully observed system with $K=2^6$ and $d=10$ terms in the Gaussian process model. 

We %
apply %
different versions of MA-ALDI with different choices of $\gamma\in\{0.01, 0.1, 1\}$ and compare them to %
pMALA. 
Note that the %
posterior $\pi$ is Gaussian with mean
$m_\ast = \Gamma_0 \tilde A^\top (\tilde A \Gamma_0 \tilde A^\top +
\Gamma)^{-1}y$ and covariance matrix
$C_\ast = \Gamma_0 - \Gamma_0 \tilde A^\top (\tilde A\Gamma_0 \tilde
A^\top + \Gamma)^{-1}\tilde A \Gamma_0\in\R^{d\times d}$ for $\tilde A:=
AB_\phi\in\R^{K\times d}$. Therefore, we consider again the quantity of interest
$f(X) = X^\top C_\ast^{-1} X \sim \chi^2(d)$, $X\sim \pi$, and the
corresponding probability
\begin{equation}\label{eq:Pref2}
  {P_{\mathrm{ref}}} = \mathbb P(f(X)\le q_{0.5}),
\end{equation}
  where $q_{0.5}$
denotes the $0.5$-quantile of the $\chi^2(d)$ distribution.  We
construct similar estimators of ${P_{\mathrm{ref}}}$ and the
autocorrelation as in the previous section. The resulting estimation
along the Markov chain and the estimated autocorrelation are plotted
in Figure~\ref{fig:quantile_est_linear_IP}.  Similar to before we
report the average values of 100 independent runs for each method.
Again, we observe that MA-ALDI outperforms pMALA.
In particular, the particle-wise MA-ALDI with highest
interaction ($\gamma = 0.01$) performs best among all considered
algorithms.

\begin{figure}[!htb]
  \begin{subfigure}{0.48\textwidth}
    \centering \includegraphics[width=\textwidth]{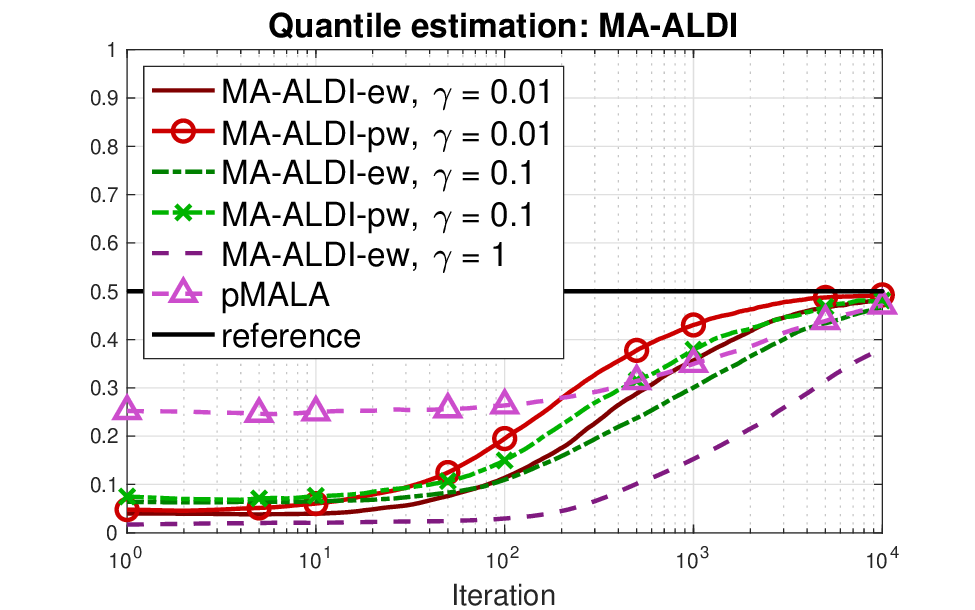}
    \caption{quantile estimation}
  \end{subfigure}\hfill
  \begin{subfigure}{0.48\textwidth}
    \centering
    \includegraphics[width=\textwidth]{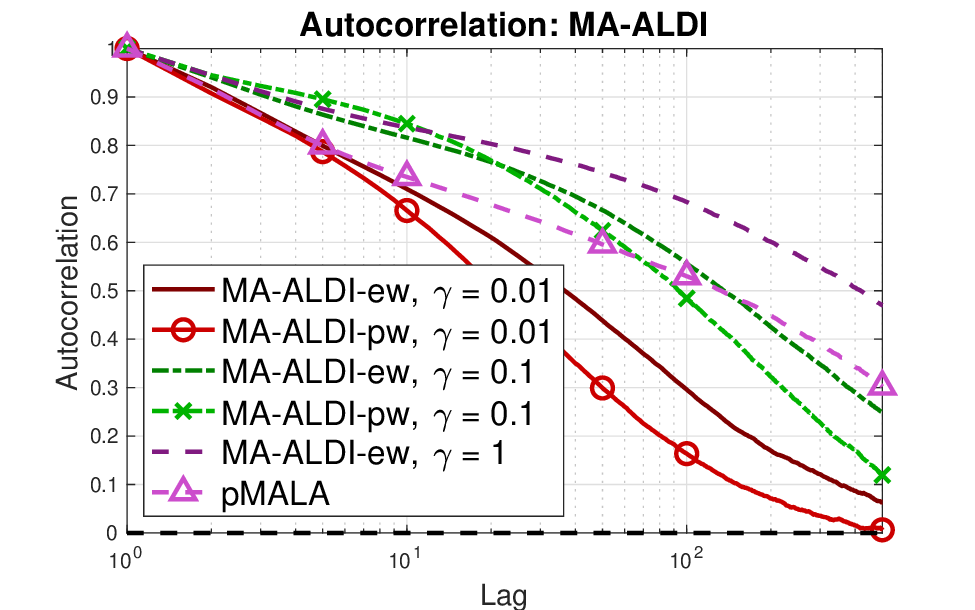}
    \caption{autocorrelation}
    \end{subfigure}
\caption{Estimation of $P_{\mathrm{ref}}$ in \eqref{eq:Pref2} and the
    corresponding estimated autocorrelation for different versions
    of MA-ALDI and $\gamma$,
    and pMALA.  For MA-ALDI-ew and pMALA the step size $h$ was tuned
    to obtain an acceptance rate of approximately $50\%$.  For
    MA-ALDI-pw we use the same $h$ as for
    MA-ALDI-ew.
 }\label{fig:quantile_est_linear_IP}
\end{figure} 

Moreover, we show the resulting posterior approximation averaged over the realized Markov chains pushed forward through the truncated KL-expansion, i.e., $\widehat \theta_k = B_\phi \XX_k$.
In Figure~\ref{fig:posterior_apprx_MHEKS_linear_IP} we plot the pointwise mean $\frac 1N \sum_{k=N_{\mathrm{burn}}}^{N_{\mathrm{burn}}+N}\widehat \theta_k$ plus/minus the pointwise empirical standard deviation. 
Here, we used again a burn-in of $N_{\mathrm{burn}} = 10^3$ iterations and $N=10^4$ and observe a smaller deviation to true posterior for %
MA-ALDI-ew ($\gamma =0.01$) than for %
pMALA.

\begin{figure}[!htb]
  \begin{subfigure}{0.48\textwidth}
    \centering \includegraphics[width=\textwidth]{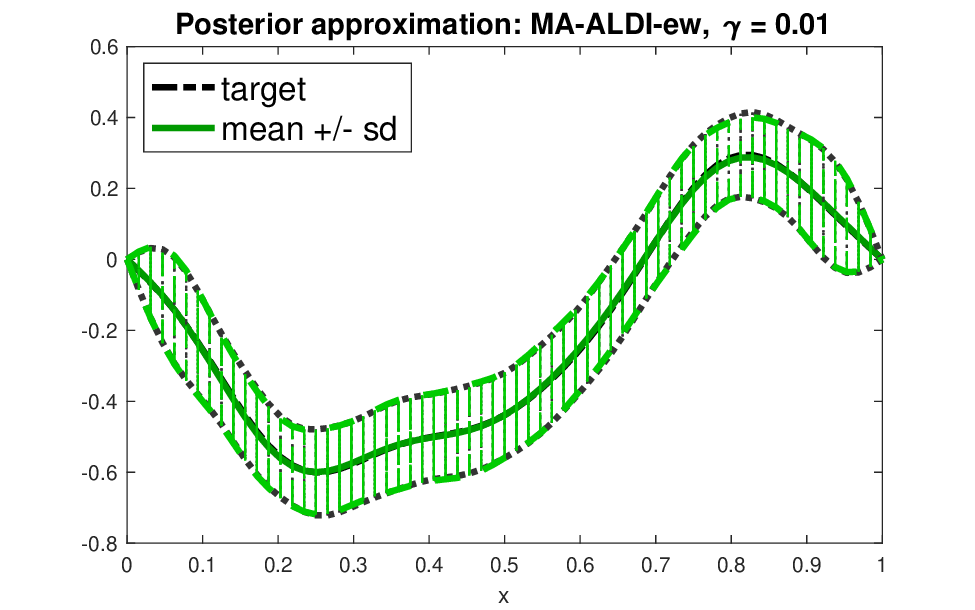}
    \caption{MA-ALDI-ew}
  \end{subfigure}\hfill
  \begin{subfigure}{0.48\textwidth}
    \centering
    \includegraphics[width=\textwidth]{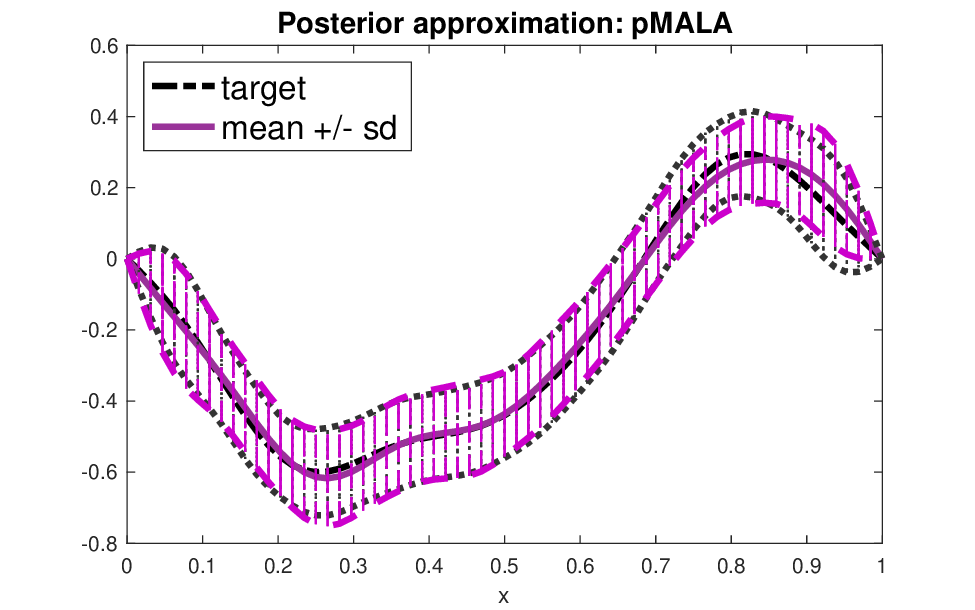}
    \caption{pMALA}
    \end{subfigure}
\caption{Posterior approximation (averaged over all iterations) for %
  MA-ALDI-ew
  with $\gamma = 0.01$ and %
  pMALA.
    The step size $h$ was chosen such that the acceptance rate is approximately $50\%$.}\label{fig:posterior_apprx_MHEKS_linear_IP}
\end{figure} 

\section{Conclusion}
The success of MCMC methods, particularly in high-dimensional
problems, heavily relies on the quality of the proposal
distribution. Ideally, additional information on the target
distribution, such as its covariance, should be used. While this is
not available in practice, it can for example be estimated along the
\emph{path} of the chain which has led to the developement of
  adaptive MCMC methods.  In the present work, we consider an
alternative approach that %
evolves an \emph{ensemble} of $M\in\N$ interacting particles and
leverages the information gained by the entire ensemble to generate a
proposal for the next update. One key advantage %
is that this method provides an effective and natural means of
parallelization which takes full advantage of the additional
information provided by the ensemble. %
This can be crucial in the treatment of real-world problems. For
instance, in engineering and science, solving a (Bayesian) inverse
problems often involves simulating a complex physical process at each
step of the chain. Each of these simulations can take minutes or even
hours to complete, which renders any sequential algorithm and any MCMC
approach that mixes only slowly infeasible.

The present study investigated three fundamental variants of
Metropolizing interacting particle systems that evolve $M\in\N$
particles based on some stochastic ODE. The first variant considers
the update as a proposal for the product of the target distribution in
the product state space $\mathbb{R}^{Md}$. It either accepts or
rejects the entire ensemble. The second %
variant employs
particle-wise Metropolization, where each particle is accepted or
rejected individually and %
sequentially. To allow for parallelization, in the third
variant we partition the ensembles into blocks of equal
size and sequentially accept or reject each block.
While variant two
has been proposed and discussed previously, e.g., in
\cite{CW2021,DS2022,GW2010,LMW2018}, variants one and three are novel
to the best of our knowledge. Furthermore, all three variants allow
for the construction of affine-invariant MCMC methods through
affine-invariant particle dynamics.

We presented a detailed empirical study comparing these %
methods for several common
particle dynamics. Our findings
show that the interaction of the particles can significantly
improve mixing compared to trivially running $M$ independent MCMC
chains (in parallel).
Moreover, depending on the situation, we observed that the
  particle- and block-wise Metropolization seem to outperform the
  ensemble-wise variant.
Overall, our study suggests that proposals based on interacting particle systems can provide significant improvements over traditional MCMC methods. %
Additionally, %
we provide a theoretical analysis of
these methods,
establishing basic ergodicity %
under mild and common assumptions. Finally, in the appendix we also discuss a
``simultaneous'' (instead of sequential) variant, and show why it does in general not yield the correct invariant distribution.
Potential modifications to fix this biasedness are left as an open question for future work.

Other possible directions for future work include additionally using the history of the Markov chain, e.g., estimating the target covariance also along the path of the ensemble chain which may reduce the estimation error for the covariance.
Also %
the application of localization techniques as discussed e.g., in \cite{Huang_2022,RW2021} within the MA-IPS approach seems beneficial.
Finally, while we provide basic convergence results, a solid theoretical analysis of the superiority of interacting ensembles over independent, parallel Markov chains %
remains an open %
and interesting avenue for future work.

\paragraph{Acknowledgements} We thank Daniel Rudolf for very helpful
and enduring discussions.

\bibliographystyle{abbrv} \bibliography{references.bib}

\begin{thebibliography}{10}

\bibitem{Besag1994}
J.~Besag.
\newblock Discussion of ``{R}epresentations of knowledge in complex systems".
\newblock {\em J. Roy. Statist. Soc. Ser. B}, 56(4):591--592, 1994.

\bibitem{MCMCbook}
S.~Brooks, A.~Gelman, G.~Jones, and X.-L. Meng~(Eds.).
\newblock {\em Handbook of Markov Chain Monte Carlo}.
\newblock Chapman and Hall/CRC, New York, NY, 2011.

\bibitem{CHSV2022}
J.~A. Carrillo, F.~Hoffmann, A.~M. Stuart, and U.~Vaes.
\newblock Consensus-based sampling.
\newblock {\em Studies in Applied Mathematics}, 148(3):1069--1140, 2022.

\bibitem{NEURIPS2018_69386f6b}
R.~T.~Q. Chen, Y.~Rubanova, J.~Bettencourt, and D.~K. Duvenaud.
\newblock Neural ordinary differential equations.
\newblock In S.~Bengio, H.~Wallach, H.~Larochelle, K.~Grauman, N.~Cesa-Bianchi,
  and R.~Garnett, editors, {\em Advances in Neural Information Processing
  Systems}, volume~31. Curran Associates, Inc., 2018.

\bibitem{CF2010}
J.~A. Christen and C.~Fox.
\newblock {A general purpose sampling algorithm for continuous distributions
  (the t-walk)}.
\newblock {\em Bayesian Analysis}, 5(2):263 -- 281, 2010.

\bibitem{CRSW2013}
S.~Cotter, G.~Roberts, A.~Stuart, and D.~White.
\newblock {MCMC} methods for functions: Modifying old algorithms to make them
  faster.
\newblock {\em Statistial Science}, 28(3):283--464, 2013.

\bibitem{CW2021}
J.~Coullon and R.~J. Webber.
\newblock Ensemble sampler for infinite-dimensional inverse problems.
\newblock {\em Statistics and Computing}, 31:28, 2021.

\bibitem{MR4065222}
S.~Dolgov, K.~Anaya-Izquierdo, C.~Fox, and R.~Scheichl.
\newblock Approximation and sampling of multivariate probability distributions
  in the tensor train decomposition.
\newblock {\em Stat. Comput.}, 30(3):603--625, 2020.

\bibitem{Duncan2019OnTG}
A.~Duncan, N.~Nuesken, and L.~Szpruch.
\newblock On the geometry of {S}tein variational gradient descent.
\newblock {\em ArXiv}, abs/1912.00894, 2019.

\bibitem{DS2022}
M.~M. Dunlop and G.~Stadler.
\newblock A gradient-free subspace-adjusting ensemble sampler for
  infinite-dimensional {B}ayesian inverse problems.
\newblock {\em ArXiv}, abs/2202.11088, 2022.

\bibitem{Gallego2018StochasticGM}
V.~Gallego and D.~R. Insua.
\newblock Stochastic gradient {MCMC} with repulsive forces.
\newblock {\em ArXiv}, abs/1812.00071, 2018.

\bibitem{GHWS2020}
A.~Garbuno-Inigo, F.~Hoffmann, W.~Li, and A.~M. Stuart.
\newblock Interacting {L}angevin diffusions: gradient structure and ensemble
  {K}alman sampler.
\newblock {\em SIAM Journal on Applied Dynamical Systems}, 19(1):412--441,
  2020.

\bibitem{GNR2020}
A.~Garbuno-Inigo, N.~N\"{u}sken, and S.~Reich.
\newblock Affine invariant interacting {L}angevin dynamics for {B}ayesian
  inference.
\newblock {\em SIAM Journal on Applied Dynamical Systems}, 19(3):1633--1658,
  2020.

\bibitem{GW2010}
J.~Goodman and J.~Weare.
\newblock Ensemble samplers with affine invariance.
\newblock {\em Comm. App. Math. and Comp. Sci.}, (1), 2010.

\bibitem{GM1994}
U.~Grenander and M.~I. Miller.
\newblock Representations of knowledge in complex systems.
\newblock {\em J. Roy. Statist. Soc. Ser. B}, 56(4):549--603, 1994.

\bibitem{H1970}
W.~K. Hastings.
\newblock Monte {C}arlo sampling methods using {M}arkov chains and their
  applications.
\newblock {\em Biometrika}, 57(1):97--109, 1970.

\bibitem{Huang_2022}
D.~Z. Huang, J.~Huang, S.~Reich, and A.~M. Stuart.
\newblock Efficient derivative-free {B}ayesian inference for large-scale
  inverse problems.
\newblock {\em Inverse Problems}, 38(12):125006, oct 2022.

\bibitem{ILS2013}
M.~A. Iglesias, K.~J.~H. Law, and A.~M. Stuart.
\newblock Ensemble {K}alman methods for inverse problems.
\newblock {\em Inverse Problems}, 29(4):045001, mar 2013.

\bibitem{jaini2019sum}
P.~Jaini, K.~A. Selby, and Y.~Yu.
\newblock Sum-of-squares polynomial flow.
\newblock {\em ICML}, 2019.

\bibitem{JKO98}
R.~Jordan, D.~Kinderlehrer, and F.~Otto.
\newblock The variational formulation of the {F}okker--{P}lanck equation.
\newblock {\em SIAM Journal on Mathematical Analysis}, 29:1--17, 1998.

\bibitem{NEURIPS2020_3202111c}
A.~Korba, A.~Salim, M.~Arbel, G.~Luise, and A.~Gretton.
\newblock A non-asymptotic analysis for {S}tein variational gradient descent.
\newblock In {\em Advances in Neural Information Processing Systems},
  volume~33, pages 4672--4682. Curran Associates, Inc., 2020.

\bibitem{LMW2018}
B.~Leimkuhler, C.~Matthews, and J.~Weare.
\newblock Ensemble preconditioning for {M}arkov chain {M}onte {C}arlo
  simulation.
\newblock {\em Statistics and Computing}, 28(2):277--290, 2018.

\bibitem{SVGDgradientflow}
Q.~Liu.
\newblock Stein variational gradient descent as gradient flow.
\newblock In {\em Advances in Neural Information Processing Systems 30}, pages
  3115--3123. Curran Associates, Inc., 2017.

\bibitem{LW2016}
Q.~Liu and D.~Wang.
\newblock Stein variational gradient descent: A general purpose {B}ayesian
  inference algorithm.
\newblock In {\em Proceedings of the 30th International Conference on Neural
  Information Processing Systems}, NIPS’16, page 2378–2386, Red Hook, NY,
  USA, 2016. Curran Associates Inc.

\bibitem{MR1812873}
P.~A. Markowich and C.~Villani.
\newblock On the trend to equilibrium for the {F}okker-{P}lanck equation: an
  interplay between physics and functional analysis.
\newblock volume~19, pages 1--29. 2000.
\newblock VI Workshop on Partial Differential Equations, Part II (Rio de
  Janeiro, 1999).

\bibitem{MR3821485}
Y.~Marzouk, T.~Moselhy, M.~Parno, and A.~Spantini.
\newblock Sampling via measure transport: an introduction.
\newblock In {\em Handbook of uncertainty quantification. {V}ol. 1, 2, 3},
  pages 785--825. Springer, Cham, 2017.

\bibitem{MRRT1953}
N.~{Metropolis}, A.~W. {Rosenbluth}, M.~N. {Rosenbluth}, A.~H. {Teller}, and
  E.~{Teller}.
\newblock {Equation of State Calculations by Fast Computing Machines}.
\newblock {\em J. Chem. Phys.}, 21(6):1087--1092, June 1953.

\bibitem{NR19}
N.~N\"usken and S.~Reich.
\newblock Note on interacting {L}angevin diffusion: {G}radient structure and
  ensemble {K}alman sampler.
\newblock Technical Report arXiv:1908.10890v1, University of Potsdam, 2019.

\bibitem{Nsken2021SteinVG}
N.~N{\"u}sken and D.~R.~M. Renger.
\newblock Stein variational gradient descent: Many-particle and long-time
  asymptotics.
\newblock {\em Foundations of Data Science}, 2023.

\bibitem{PRS2021}
S.~Pathiraja, S.~Reich, and W.~Stannat.
\newblock Mckean--{V}lasov {SDE}s in nonlinear filtering.
\newblock {\em SIAM Journal on Control and Optimization}, 59(6):4188--4215,
  2021.

\bibitem{P14}
G.~Pavliotis.
\newblock {\em Stochastic Processes and Applications: Diffusion Processes, the
  Fokker-Planck and Langevin Equations}.
\newblock Texts in Applied Mathematics. Springer New York, 2014.

\bibitem{PTTM2017}
R.~Pinnau, C.~Totzeck, O.~Tse, and S.~Martin.
\newblock A consensus-based model for global optimization and its mean-field
  limit.
\newblock {\em Mathematical Models \& Methods in Applied Sciences},
  27(1):183--204, 2017.

\bibitem{RW2021}
S.~Reich and S.~Weissmann.
\newblock {F}okker--{P}lanck particle systems for {B}ayesian inference:
  Computational approaches.
\newblock {\em SIAM/ASA Journal on Uncertainty Quantification}, 9(2):446--482,
  2021.

\bibitem{pmlr-v37-rezende15}
D.~Rezende and S.~Mohamed.
\newblock Variational inference with normalizing flows.
\newblock In F.~Bach and D.~Blei, editors, {\em Proceedings of the 32nd
  International Conference on Machine Learning}, volume~37 of {\em Proceedings
  of Machine Learning Research}, pages 1530--1538, Lille, France, 07--09 Jul
  2015. PMLR.

\bibitem{RC2004}
C.~P. Robert and G.~Casella.
\newblock {\em Monte Carlo Statistical Methods}.
\newblock Texts in Statistics. Springer New York, 2004.

\bibitem{RoRo2001}
G.~O. Roberts and J.~S. Rosenthal.
\newblock Optimal scaling for various {M}etropolis--{H}astings algorithms.
\newblock {\em Statistical Science}, 16(4):351--367, 2001.

\bibitem{RoRo2004}
G.~O. Roberts and J.~S. Rosenthal.
\newblock General state space {M}arkov chains and {MCMC} algorithms.
\newblock {\em Probability Surveys}, 1:20--71, 2004.

\bibitem{RoRo2006}
G.~O. Roberts and J.~S. Rosenthal.
\newblock Harris recurrence of {M}etropolis-within-{G}ibbs and
  trans-dimensional {M}arkov chains.
\newblock {\em Annals of Applied Probability}, 16(4):2123--2139, 2006.

\bibitem{MR1440273}
G.~O. Roberts and R.~L. Tweedie.
\newblock Exponential convergence of {L}angevin distributions and their
  discrete approximations.
\newblock {\em Bernoulli}, 2(4):341--363, 1996.

\bibitem{RS2018}
D.~Rudolf and B.~Sprungk.
\newblock On a generalization of the preconditioned {C}rank–{N}icolson
  {M}etropolis algorithm.
\newblock {\em Found. Comput. Math.}, 18:309--343, 2018.

\bibitem{RS2022}
D.~Rudolf and B.~Sprungk.
\newblock Robust random walk-like {M}etropolis--{H}astings algorithms for
  concentrating posteriors.
\newblock {\em arXiv:2202.12127}, 2022.

\bibitem{NEURIPS2019_65a99bb7}
S.~Vempala and A.~Wibisono.
\newblock Rapid convergence of the unadjusted {L}angevin algorithm:
  Isoperimetry suffices.
\newblock In H.~Wallach, H.~Larochelle, A.~Beygelzimer, F.~d\textquotesingle
  Alch\'{e}-Buc, E.~Fox, and R.~Garnett, editors, {\em Advances in Neural
  Information Processing Systems}, volume~32. Curran Associates, Inc., 2019.

\bibitem{Weissmann_2022}
S.~Weissmann.
\newblock Gradient flow structure and convergence analysis of the ensemble
  {K}alman inversion for nonlinear forward models.
\newblock {\em Inverse Problems}, 38(10):105011, sep 2022.

\end{thebibliography}

\appendix

\section{On simultaneous particle-wise Metropolization}\label{sec:pw-sim}
From a computational viewpoint, it would be advantegous to decide
  for each particle independently and in parallel whether to accept or
  reject it, as this facilitates the embarassingly parallel processing
  of all $M$ particles in the ensemble in each step of the
  algorithm. However, as we illustrate in the following, the
  corresponding ``simultaneous'' transition kernel is in general not invariant with respect to the product target measure $\ppi$ or even an $M$-coupling of $\pi$.

To formalize the outlined procedure, we consider the independent and simultaneous application of the particle-wise transition kernel
$P_{\xx^{-(i)}}$ in \eqref{eq:Pi} to the $i$th particle for each $i\in\{1,\dots,M\}$. This yields the transition kernel
$\PP_\mathrm{sim} \colon \R^{Md} \times \mathcal B(\R^{Md}) \to[0,1]$
\begin{equation}\label{eq:P_product}
    \PP_\mathrm{sim}(\xx, \mathrm d \yy)
    =
    P_{\xx^{-(1)}}(x^{(1)} , \mathrm d y^{(1)})
    \otimes
    \cdots
    \otimes
    P_{\xx^{-(M)}}(x^{(M)} , \mathrm d y^{(M)}).
  \end{equation}
The associated algorithmic description is given in \Cref{alg:MA-particle2}.

\begin{algorithm}[t]
  \begin{algorithmic}[1]
    \Statex \textbf{Input:} 
    \begin{itemize}
    \item target density $\pi$ on $\R^d$
    \item ensemble dependent proposal kernel $Q_{\xx^-}$ with density
      $q_{\xx^-}:\R^{d}\times\R^{d}\to (0,\infty)$ in \eqref{eq:qxx}
    \item initial probability distribution $\pi_0$ on $\R^d$
    \end{itemize}
    \Statex \textbf{Output:} ensemble Markov chain $(\XX_k)_{k\in\{1,\dots,N\}}$ in state space $\R^{Md}$
    \State
    draw $\xx_0\sim \otimes_{i=1}^M\pi_0$ and set initial state $\XX_0=\xx_0\in\R^{Md}$
    \For{$k=0,\dots, N$}
      \State
      $\forall i$: given $\XX_k = \xx_k$ draw proposal $y^{(i)} \sim Q_{\xx^{-(i)}_k}(x^{(i)}_k,\cdot)$ independently
    \State $\forall i$: compute particle acceptance probability $\alpha_{\xx^{-(i)}_k}(x^{(i)}_k, y^{(i)})\in [0,1]$ in \eqref{eq:alpha2}
    \State $\forall i$: draw $u_i\sim{\rm U}([0,1])$ independently and set
	$$ 
	X^{(i)}_{k+1} = \begin{cases}
          y^{(i)} &\text{if } u_i \le \alpha_{\xx^{-(i)}_k}(x^{(i)}_k, y^{(i)})\\
          x_{k}^{(i)} &\text{else}
	\end{cases}
	$$
        \EndFor
      \end{algorithmic}
      \caption{Simultaneous particle-wise Metropolized interactive particle sampling
      }\label{alg:MA-particle2}
\end{algorithm}

As for the sequential updates discussed in \Cref{sec:particle_met} and \ref{sec:block_met}, $\ppi$-reversibility does not hold, since in general
\[
    \prod_{i=1}^M
    \left( 
    \alpha_{\xx^{-(i)}}(x^{(i)},y^{(i)}) 
    \, q_{\xx^{-(i)}}(x^{(i)},y^{(i)})
    \, \pi(x^{(i)})
    \right)
    \neq
     \prod_{i=1}^M
    \left( 
    \alpha_{\yy^{-(i)}}(y^{(i)},x^{(i)}) 
    \, q_{\yy^{-(i)}}(y^{(i)},x^{(i)})
    \, \pi(y^{(i)})
    \right)
\]
for $\xx$, $\yy \in \R^{Md}$ with $x^{(i)} \neq y^{(j)}$ for all $i,j=1,\ldots,M$---except for the case of noninteraction, i.e., $q_{\xx^{-}} = q$ and, thus, $\alpha_{\xx^{-}} = \alpha$ does not depend on the other particles in the ensemble.

Regarding the $\ppi$ invariance of $\PP_{\rm sim}$, it is worth noting that each $P_{\xx^{-(i)}}$ is $\pi$-invariant. However, due to the interaction, this does not directly imply $\ppi$ invariance of $\PP_{\rm sim}$. 
Although, the particle-wise marginals of $\ppi$ and $\ppi\PP_\mathrm{sim}$ coincide as shown below in \Cref{propo:P_product} the product transition kernel $\PP_\mathrm{sim}$ is in general not $\ppi$ invariant as illustrated by several counterexamples below.

\begin{proposition}\label{propo:P_product}
For the transition kernel $\PP_\mathrm{sim} \colon \R^{Md} \times \mathcal B(\R^{Md}) \to[0,1]$ given in \eqref{eq:P_product} associated to \Cref{alg:MA-particle2} we have that $\ppi \PP_\mathrm{sim}$ is an \emph{$M$-coupling} of $\pi$, i.e. the particle-wise marginals of $\XX \sim \ppi \PP_\mathrm{sim}$ are $X^{(i)} \sim \pi$ for all $i=1,\ldots,M$. 
\end{proposition}
\begin{proof}
For any $i=1,\ldots,M$ and any $A_i \in \cB(\R^d)$ we have
\begin{align*}
    \ppi \PP_\mathrm{sim}(\R^d \times \cdots \times \R^d \times A_i \times \R^d\times \cdots \times \R^d)
    & =
    \int_{\R^{Md}}
    \PP_\mathrm{sim}(\xx, \R^d \times \cdots \times \R^d \times A_i \times \R^d\times \cdots \times \R^d)
    \, \ppi(\dd\xx)\\
    & =
    \int_{\R^{Md}}
    P_{\xx^{-(i)}}(x^{(i)}, A_i) \,
    \, \ppi(\dd\xx)\\
    & = 
    \int_{\R^{(M-1)d}}
   \left(  \int_{\R^d}
    P_{\xx^{-(i)}}(x^{(i)}, A_i)
    \, \pi(\dd x^{(i)})\right) 
    \, \bigotimes_{j\neq i} \pi(\dd x^{(j)})
    \\
    & = \int_{\R^{(M-1)d}} \pi(A_i) \, \bigotimes_{j\neq i} \pi(\dd x^{(j)})
    = \pi(A_i)    
\end{align*}
due to the $\pi$-invariance of $P_{\xx^{-}}$ for any $\xx^{-} \in \R^{(M-1)d}$. 
\end{proof}

\begin{remark}
If we \emph{assume} $\ppi$-invariance of $\PP_\mathrm{sim}$, then ergodicity and a strong law of large numbers follow under the same conditions as in \Cref{theo:MAPS2_seq}; this can be proven by a slight modification of the arguments in \cite{RoRo2006}.
\end{remark}

\begin{example}
Let us consider a discrete state space $\mathcal X = \{x_1,x_2\}$ of two elements $x_1\neq x_2$ and the uniform distribution $\pi(x_1) = \pi(x_2) = \frac 12$ on $\mathcal X$ as the target measure.
Consider parametrized proposal kernels
$Q_z \colon \mathcal X \times \mathcal X \to [0,1]$ with parameter
$z \in \mathcal X$ which we write as right stochastic matrices
$\mathrm Q_z \in [0,1]^{|\mathcal X|}$ where the element $q_{ij}$ of
$\mathrm Q_z$ in the $i$th row and $j$th column denotes the
probability $Q_z(x_i, x_j)$:
\[
    \mathrm Q_{x_1}
    =
    \begin{pmatrix}
    2/3 & 1/3 \\
    2/3 & 1/3
    \end{pmatrix},
    \qquad
    \mathrm Q_{x_2}
    =
    \begin{pmatrix}
    1/2 & 1/2 \\
    1/2 & 1/2
    \end{pmatrix}.
\]
The resulting parametrized acceptance probabilities
$\alpha_z \colon \mathcal X \times \mathcal X \to [0,1]$
are given by (cp.~\eqref{eq:alpha0})
\[
    \alpha_{x_1}(x,y)
    =
    \begin{cases}
    1/2, & \text{if } x = x_2, y = x_1\\
    1, & \text{else}
    \end{cases},
    \qquad
    \alpha_{x_2}(x,y)
    \equiv
    1,
  \]
  and the rejection probabilies by (cp.~\eqref{eq:Pr})
\[
    r_{x_1}(x)
    =
    \begin{cases}
    1/3, & x = x_1\\
    0, & x = x_2
    \end{cases},
    \qquad
    r_{x_2}(x) \equiv 0.
\]
This
yields as $\pi$-invariant MH transition kernel $P_z\colon \mathcal X \times \mathcal X \to [0,1]$ again written as right stochastic matrices $\mathrm P_z$ (cp.~\eqref{eq:Pr})
\[
    \mathrm P_{x_1}
    =
    \begin{pmatrix}
    2/3 & 1/3 \\
    2/3 & 1/3
    \end{pmatrix},
    \qquad
    \mathrm P_{x_2}
    =
    \begin{pmatrix}
    1/2 & 1/2 \\
    1/2 & 1/2
    \end{pmatrix}.
\]
The resulting product transition kernel 
\[
    \PP_\mathrm{sim}(\xx, \yy) 
    = 
    P_{x^{(2)}}(x^{(1)}, y^{(1)}) \cdot P_{x^{(1)}}(x^{(2)}, y^{(2)})
\]
is then, written as well as a $(4\times 4)$-right stochastic matrix where the rows and columns in the matrix correspond to the lexicographically ordered states $(x_1,x_1), (x_1,x_2), (x_2,x_1), (x_2,x_2)$ in $\mathcal X^2$,
\[
    \PP_\mathrm{sim}
    \simeq
    \begin{pmatrix}
    4/9 & 2/9 & 2/9 & 1/9\\
    1/6 & 2/6 & 1/6 & 2/6 \\
    1/6 & 2/6 & 1/6 & 2/6 \\
    1/4 & 1/4 & 1/4 & 1/4
    \end{pmatrix}.
\]
The associated invariant measure $\boldsymbol\nu\colon \mathcal X^2 \to [0,1]$ is given by
\[
    \boldsymbol\nu(x_1,x_1) = \frac{45}{173},
    \quad
    \boldsymbol\nu(x_1,x_2) = \frac{49}{173},
    \quad
    \boldsymbol\nu(x_2,x_1) = \frac{35}{173},
    \quad
    \boldsymbol\nu(x_2,x_2) = \frac{44}{173}
\]
which does not correspond to $\boldsymbol\pi \equiv \frac 14$.
Moreover, also the particlewise marginals $\nu^{(i)}$ of $\boldsymbol\nu$ do not coincide with $\pi$:
\[
    \nu^{(1)}(x_1) = \frac{94}{173},
    \quad
    \nu^{(1)}(x_2) = \frac{79}{173},
    \qquad
    \nu^{(2)}(x_1) = \frac{80}{173},
    \quad
    \nu^{(2)}(x_2) = \frac{93}{173}.
\]
\end{example}

\begin{example}\label{exam:counter_num1}
Let us consider the continuous state space $\mathcal X = \mathbb R$ equipped with the triangular target distribution $\pi$ given by the Lebesgue density
\[
    \pi(x)
    =
    \begin{cases}
    \frac 12 + 2 \min\{x , 1-x\}, & \text{if } x\in[0,1]\\
    0, & \text{else.}
    \end{cases}
\]
Again, we consider an ensemble Markov chain of $M=2$ interacting particles $\XX_k = (X_k^{(1)}, X_k^{(2)})^\top \in \mathbb R^2$.
As ensemble proposal kernel we choose the following
\[
    \QQ(\xx, \dd \yy)
    =
    Q_{x^{(2)}}(x^{(1)}, \dd y^{(1)})
    \otimes
    Q_{x^{(1)}}(x^{(2)}, \dd y^{(2)}),
    \qquad
    Q_{z}(x, \dd y)
    =
    \Nv\left(\frac {x+z}2, \frac14 \right),
\]
which corresponds particle-wise to $\Nv\left( m(\XX_k), h\right)$ for
$h=\frac 14$. %
Consider the transition kernel $\PP_\mathrm{sim}$ resulting from
particle-wise Metropolization.  Note that for $M=2$ we can decompose
$\PP_\mathrm{sim}$ as follows
\begin{align*}
    \PP_\mathrm{sim}((x_1,x_2), \dd y_1, \dd y_2)
    & =
    \alpha_{x_2}(x_1,y_1)
    \alpha_{x_1}(x_2,y_2)
    q_{x_2}(x_1,y_1)
    q_{x_1}(x_2,y_2)
    \dd y_1 \dd y_2\\
    & \quad +
    \alpha_{x_2}(x_1,y_1)
    r_{x_1}(x_2)
    q_{x_2}(x_1,y_1)
    \dd y_1 \delta_{x_2}(\dd y_2)\\
    & \quad +    
    r_{x_2}(x_1)
    \alpha_{x_1}(x_2,y_2)
    q_{x_1}(x_2,y_2)
    \dd y_2 \delta_{x_1}(\dd y_1)\\
    & \quad +
    r_{x_2}(x_1)
    r_{x_1}(x_2)
    \delta_{x_1}(\dd y_1)
    \delta_{x_2}(\dd y_2)
\end{align*}
where $q_{z}$ denotes the Lebesgue density of $Q_{z}$.  We then
discretize the state space to obtain a transition matrix
$\PP_\mathrm{sim} \in \R^{n\times n}$ and compute its invariant
measure as approximation to the true invariant measure of the operator
$\PP_\mathrm{sim}$.  Since $\alpha_{z}(x,y) = 0$ for $y\notin[0,1]$ it
suffices to discretize $[0,1]^2$.  Here we use a uniform grid with
grid size $\Delta x = 0.01$ in each dimension.  Thus, $n = 101^2$.The
invariant measure of the matrix $\PP_\mathrm{sim} \in \R^{n\times n}$
is then computed numerically and rearranged to yield
$\ppi_\text{inv} \in [0,1]^{101 \times 101}$.  It is displayed in
comparison to an analogously discretized version of the true product
target $\ppi = \pi \otimes \pi$ in Figure \ref{fig:2d_comp}.  We do
notice a bias, although a small one of relative size $10^{-3}$ to
$10^{-2}$.  Since the crucial object for sampling purposes is not
necesarilly the invariant measure in the ensemble space but the
particle-wise marginals of it we also compare these in Figure
\ref{fig:2d_comp_marg}.  However, the results are similar here. A
small bias is observable, again the relative size compared to the true
target are of order $10^{-3}$.

\begin{figure}[!htb]
\includegraphics[width=0.32\textwidth]{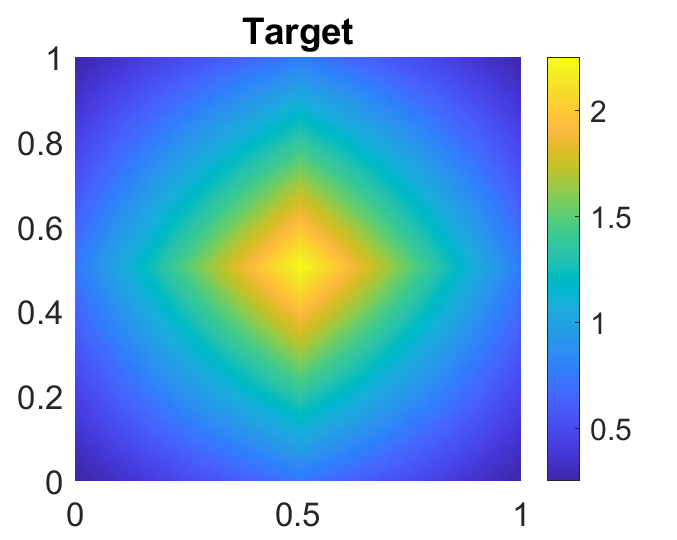}\hfill\includegraphics[width=0.32\textwidth]{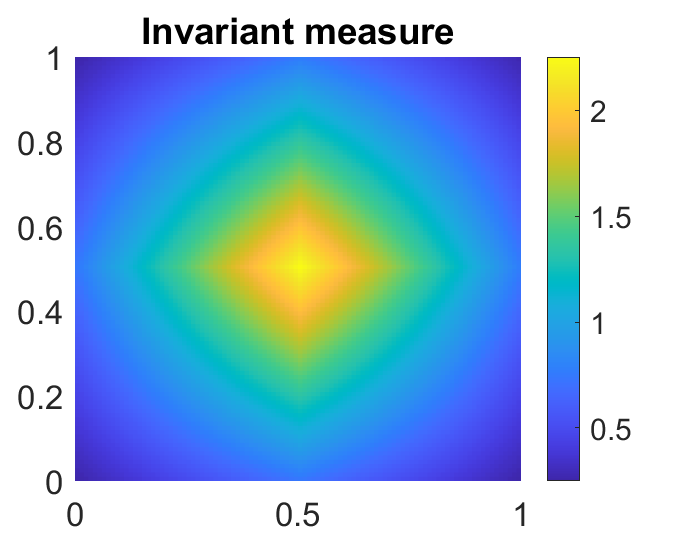}\hfill\includegraphics[width=0.32\textwidth]{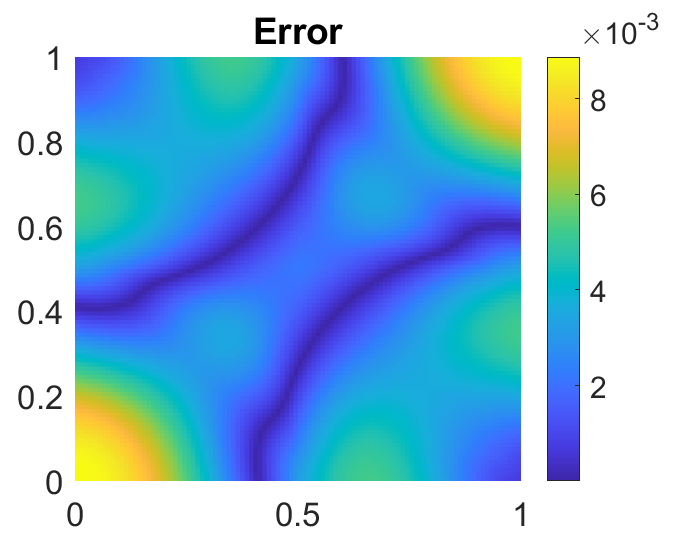}
 \caption{Comparison true product target and numerically computed invariant measure of simultaneous particle-wise Metropoliziation $\PP_\mathrm{sim}$ for $M=2$ particles in \Cref{exam:counter_num1}.}\label{fig:2d_comp}
\end{figure} 

\begin{figure}[!htb]
\includegraphics[width=0.32\textwidth]{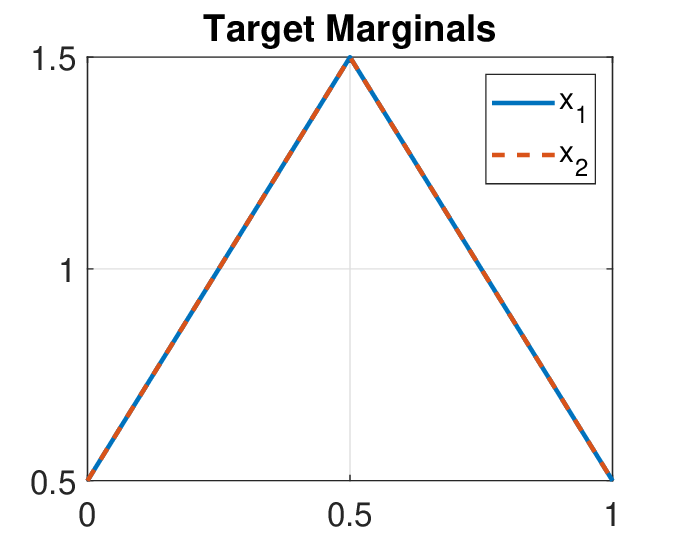}\hfill\includegraphics[width=0.32\textwidth]{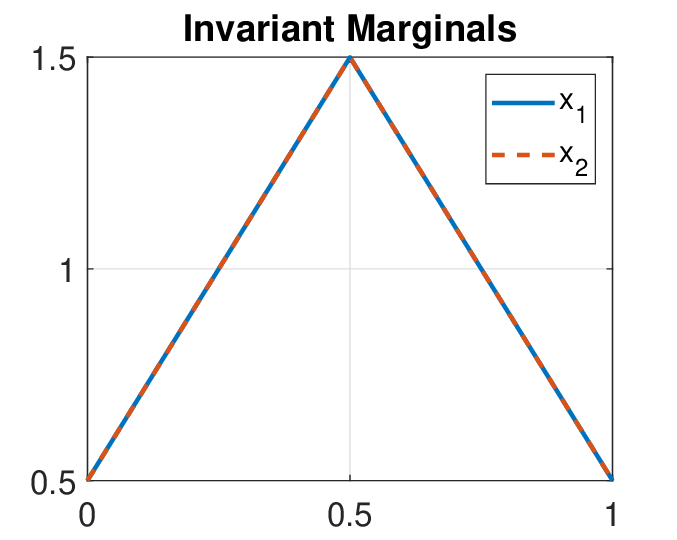}\hfill\includegraphics[width=0.32\textwidth]{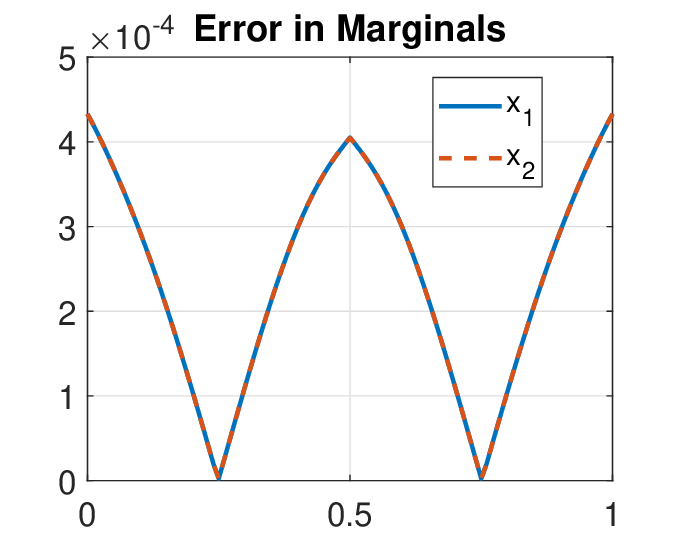}
 \caption{Comparison true product target and numerically computed invariant measure of simultaneous particle-wise Metropoliziation $\PP_\mathrm{sim}$ for $M=2$ particles in \Cref{exam:counter_num1}.}\label{fig:2d_comp_marg}
\end{figure} 
\end{example}

\begin{example}\label{exam:counter_num2}
We provide another numerical example similar to the previous one.
Here again, $\mathcal X = \mathbb R$ but now $\pi = \U[0,1]$ is the uniform distribution on $[0,1]$.
We consider $M=2$ interacting particles $\XX_k = (X_k^{(1)}, X_k^{(2)})^\top \in \mathbb R^2$ based on the following ensemble proposal kernel
\[
    \QQ(\xx, \dd \yy)
    =
    Q_{x^{(2)}}(x^{(1)}, \dd y^{(1)})
    \otimes
    Q_{x^{(1)}}(x^{(2)}, \dd y^{(2)}),
    \qquad
    Q_{z}(x, \dd y)
    =
    \Nv\left( x, \frac 14 \left(0.01 + \frac{0.99}2(z-x)^2\right) \right)
\]
which corresponds particle-wise to $\Nv\left( x, h (\gamma + (1-\gamma)C(\XX_k)) \right)$ with $\gamma = 0.01$ and $h=\frac 14$.  
Analogously, we discretize the state space or $[0,1]^2$, respectively, using a 
uniform grid with grid size $\Delta x = 0.01$ and compute numerically the invariant measure of the resulting transition matrix $\PP_\mathrm{sim} \in \R^{101\times 101}$.
The results are shown in Figure \ref{fig:2d_comp2} and \ref{fig:2d_comp_marg2}.
Also for this example we do notice a bias which is even larger than in the previous example, i.e., we observe a relative error of order $10^{-2}$ to $10^{-1}$ for the joint target and $10^{-2}$ for the particle-wise marginals.

\begin{figure}[!htb]
\includegraphics[width=0.32\textwidth]{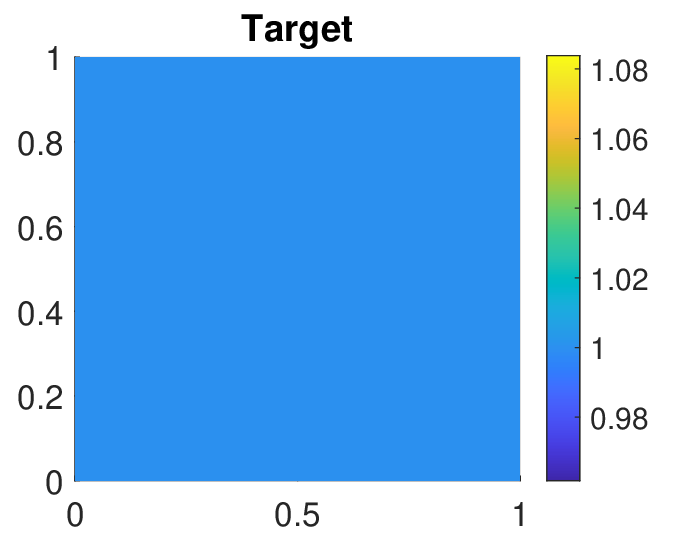}\hfill\includegraphics[width=0.32\textwidth]{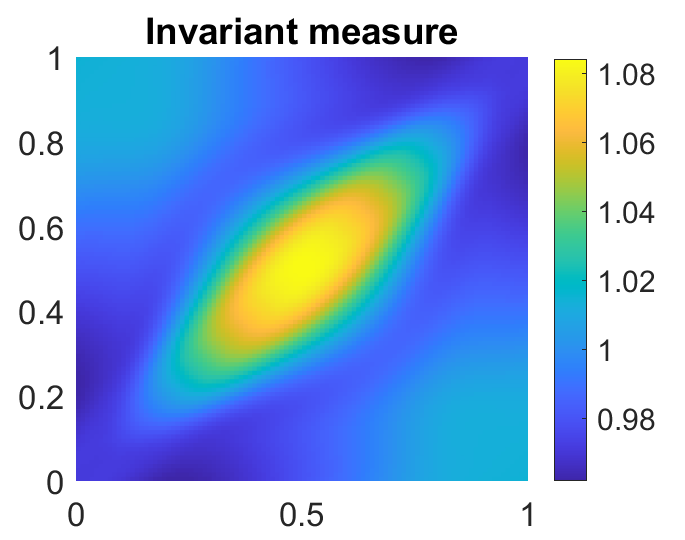}\hfill\includegraphics[width=0.32\textwidth]{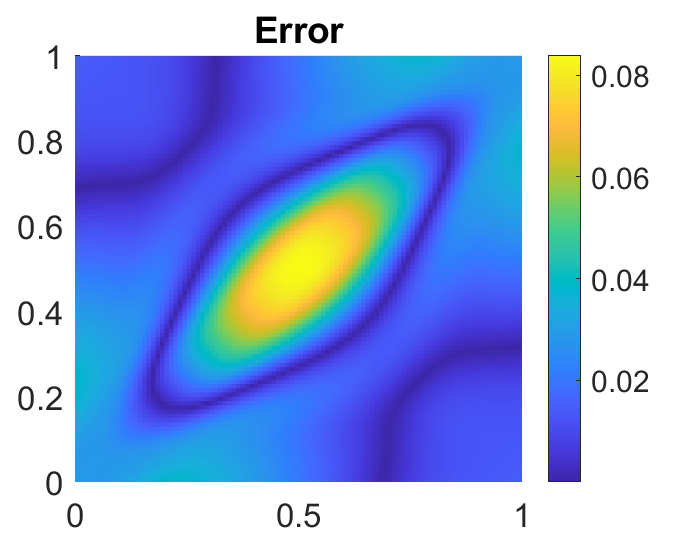}
 \caption{Comparison true product target and numerically computed invariant measure of simultaneous particle-wise Metropoliziation $\PP_\mathrm{sim}$ for $M=2$ particles in \Cref{exam:counter_num2}.}\label{fig:2d_comp2}
\end{figure} 

\begin{figure}[!htb]
\includegraphics[width=0.32\textwidth]{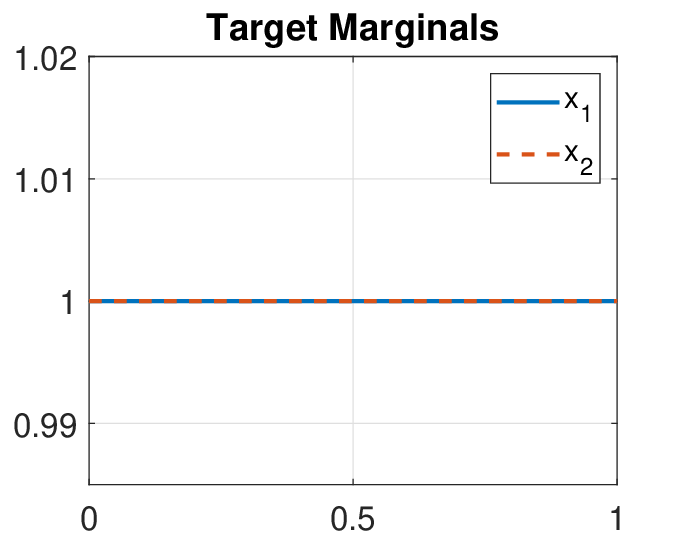}\hfill\includegraphics[width=0.32\textwidth]{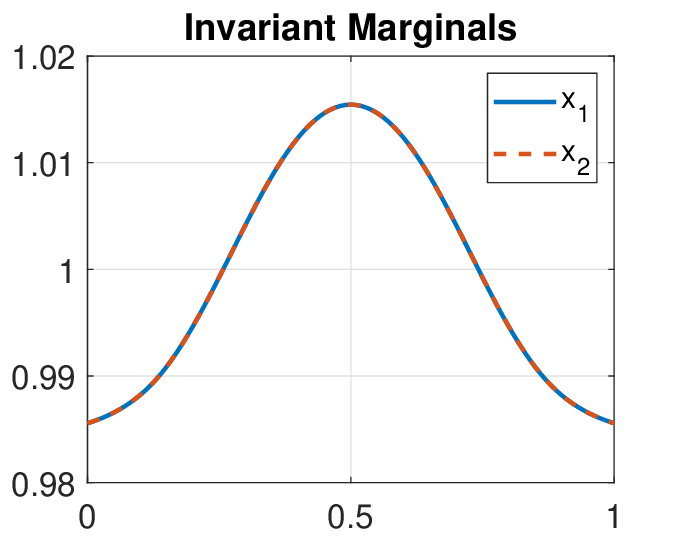}\hfill\includegraphics[width=0.32\textwidth]{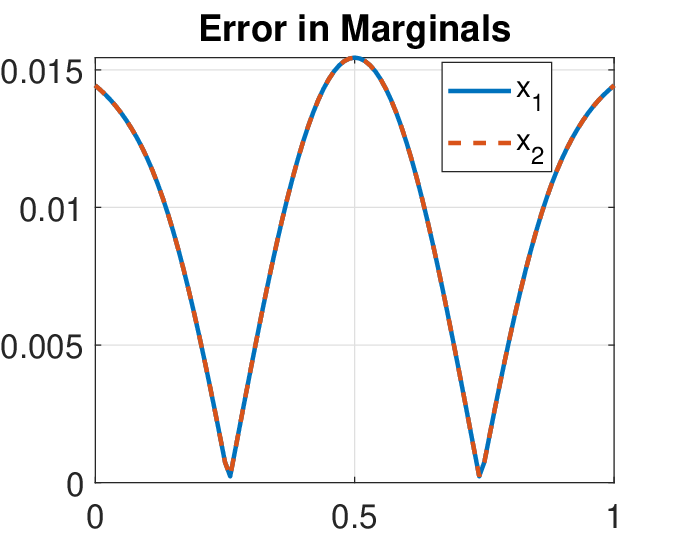}
 \caption{Comparison true product target and numerically computed invariant measure of simultaneous particle-wise Metropoliziation $\PP_\mathrm{sim}$ for $M=2$ particles in \Cref{exam:counter_num2}.}\label{fig:2d_comp_marg2}
\end{figure} 
\end{example}

We suspect that the bias of simultaneous particle-wise Metropolization is larger for smaller ensemble sizes than for bigger ones. %
In particular, the bias may vanish as $M\to\infty$ for suitable interacting particle systems, i.e., if the dynamics of each particle converge to their only time-discretized mean field limit as $M\to\infty$, then the corresponding proposal distributions $Q_{x^{-(i)}}(x^{(i)},\cdot)$ should also converge to a limit proposal distribution $Q_\infty(x^{(i)},\cdot)$ which does not depend on the other particles anymore, e.g., $Q_\infty(x^{(i)},\cdot) = \Nv(x^{(i)} + h C(\pi_t) \nabla \log \pi(x^{i}), 2h C(\pi_t) )$ in case of ALDI.
However, for independent proposal kernels $Q_{x^{-(i)}}(x^{(i)},\cdot) = Q(x^{(i)},\cdot)$ the transition kernel of simultaneous particle-wise Metropolization $\PP_\mathrm{sim}$ is in fact $\ppi$-invariant.
Therefore, we suspect that the bias of $\PP_\mathrm{sim}$ is largest for the $M=2$ particle case considered in the numerical counterexamples  above.

\end{document}